\documentclass[11pt]{article}

\usepackage{amsmath}

\usepackage{amsfonts}

\usepackage{cmap}

\usepackage{mathtools}

\usepackage{amssymb}

\usepackage{icomma}

\usepackage{latexsym}

\usepackage[utf8]{inputenc}	

\def\keywordfont{\fontsize{9}{10}\selectfont\leftskip1pc\rightskip5pc plus1fill}%

\newenvironment{jelcode}{%
	\par\addvspace{13pt plus2pt minus1pt}
	\keywordfont\noindent{\bfseries{JEL CLASSIFICATION}: }\ignorespaces%
}

\newenvironment{keywords}{%
	\par\addvspace{13pt plus2pt minus1pt}
	\keywordfont\noindent{\bfseries{KEYWORDS: }}\ignorespaces%
}

\usepackage{icomma}

\usepackage[round,longnamesfirst]{natbib}

\usepackage{xr-hyper}

\usepackage[nodisplayskipstretch]{setspace}

\usepackage[english, colorlinks=true]{hyperref}

\usepackage[english]{varioref}

\usepackage{amsthm}

\usepackage{framed}

\usepackage{graphicx}

\usepackage{microtype}

\usepackage[english]{babel}

\usepackage{bm}

\usepackage{enumitem}

\usepackage{booktabs}

\usepackage{mathrsfs}

\usepackage{placeins}

\usepackage{adjustbox}

\theoremstyle{definition}

\usepackage{mdframed}
\theoremstyle{plain}

\newtheorem{theorem}{Theorem}[section]
\newtheorem{proposition}[theorem]{Proposition}
\newtheorem{lem}[theorem]{Lemma}

\theoremstyle{definition}

\newtheorem{example}{Example}

\newtheorem{asm}{Assumption}

 \newtheorem{remark}{Remark}

\usepackage{chngcntr}
\usepackage{apptools} 
\newcommand{\air}{\vspace{0.1cm}} 
\usepackage{etoolbox}
\usepackage[svgnames]{xcolor}

\usepackage{tikz}

\usepackage{framed}

\usepackage[utf8]{inputenc}
\usepackage[T1]{fontenc}

\usepackage{caption}
\usepackage{lmodern} 
 
\usepackage{tikz}

\DeclareMathOperator{\E}{\mathbb{E}}
\DeclareMathOperator{\var}{\mathrm{Var}}

\usepackage[right=1in, left=1in,top=1in, bottom=1in]{geometry}

\DeclarePairedDelimiter{\abs}{\lvert}{\rvert} 
\newcommand{\norm}[1]{\left\lVert#1\right\rVert}

\usepackage{matlab-prettifier}
 
\usepackage{extarrows}

\DeclareMathOperator{\I}{\mathbb{I}}
\DeclareMathOperator{\R}{\mathbb{R}}

\newcommand{\Mcal}{\mathcal{M}}

\DeclarePairedDelimiter{\curl}{\lbrace}{\rbrace}
\DeclarePairedDelimiter{\floor}{\lfloor}{\rfloor}

\usepackage{apptools}
\AtAppendix{\counterwithin{thm}{section}}
\AtAppendix{\counterwithin{prop}{section}}
\AtAppendix{\counterwithin{ex}{section}}

\newcommand\blfootnote[1]{%
	\begingroup
	\renewcommand\thefootnote{}\footnote{#1}%
	\addtocounter{footnote}{-1}%
	\endgroup
} 
 
\begin{document}

	\title{\vspace{-2.8cm} Inference on Extreme Quantiles of \\ Unobserved Individual Heterogeneity  }

	\author{Vladislav Morozov\footnote{\scriptsize \href{mailto:morozov@uni-bonn.de}{morozov@uni-bonn.de};  Universität Bonn, Institute of Finance and Statistics}\blfootnote{ \scriptsize Acknowledgments: I am profoundly grateful to Andrea Sy for her generosity and knowledge of the CBI data. I am indebted to  Christian Brownlees and Kirill Evdokimov  for their  support and guidance.  I am also grateful to Iván Fernández-Val, Adam Lee, Zoel Martin, Geert Mesters, Katerina Petrova,  Barbara Rossi, Alina Shirshikova, André Souza,	   Piotr Zwiernik, and the participants at the 2021 European  Winter Meeting of the Econometric Society, the 2021 Spanish Economic Association Symposium, the  2022 Royal Economic Society Annual Conference, the 2022 BSE Summer Forum,  the 27th International Panel Data Conference, the 10th ICEEE, and the 3rd CESC for comments and discussions. }
	}

	\date{March 13, 2025}

\vspace{-1.2cm}

	\maketitle

\vspace{-0.2cm}

\begin{abstract}
\noindent We develop a methodology for conducting inference on extreme quantiles of unobserved individual heterogeneity (e.g., heterogeneous coefficients, treatment effects) in panel data and meta-analysis settings. Inference is challenging in such settings: only noisy estimates of heterogeneity are available, and central limit approximations perform poorly in the tails. We derive a necessary and sufficient condition under which noisy estimates are informative about extreme quantiles, along with sufficient rate and moment conditions. Under these conditions, we establish an extreme value theorem and an intermediate order theorem for noisy estimates. These results yield simple optimization-free confidence intervals for extreme quantiles. Simulations show that our confidence intervals have favorable coverage and that the rate conditions matter for the validity of inference. We illustrate the method with an application to firm productivity differences between denser and less dense  areas.

\end{abstract}
 


 \newpage 
 
\section{Introduction}

%

Extreme quantiles of unobserved individual heterogeneity (UIH) are of interest in the analysis of economic panel data and in  meta-analysis. UIH includes heterogeneous coefficients, treatment effects, and other latent variables. For example, in the setting of \cite{Combes2012b}, one might seek to estimate the lowest level of firm-specific productivity compatible with firm survival ---   the zeroth quantile of the productivity distribution among surviving firms.\footnote{Here and in the empirical application, ``\emph{productivity}'' refers to firm-specific total factor productivity within a Cobb-Douglas production function, in line with the production function estimation literature (e.g., \cite{Ackerberg2007EconometricToolsAnalyzing}), rather than the production frontier estimation literature (e.g., \cite{Kumbhakar2020-1}). See remark \ref{remark:production-frontier}.}

However, inference on extreme quantiles is challenging as  UIH is not directly observed --- only noisy estimates derived from individual time series, studies, or clustered data are available.  It is not  a priori clear when such estimates yield useful information about the quantiles of interest. Unlike means, extreme quantiles do not benefit from noise ``averaging out''. Moreover, estimation noise is often correlated with true UIH due to dependence between UIH and covariates used in estimation (e.g., \cite{Heckman2001, Browning2007, Browning2010}), further complicating inference.

 	This paper develops a methodology for inference on extreme quantiles of UIH using noisy estimates and establishes sharp conditions under which such estimates are informative. The key requirement is pointwise asymptotic tail equivalence --- the tails of the noisy estimates' distribution must converge in a certain weak pointwise sense to the tails of the latent UIH distribution. We construct confidence intervals and hypothesis tests using self-normalizing ratios of extreme or intermediate order statistics and derive   extreme and intermediate value theorems  for noisy estimates.

  Our inference methods rely on two asymptotic approximations:  {extreme order} and  {intermediate order}. Extreme order methods exploit ratios of the highest order statistics, and we show that the limiting distributions of these ratios can be estimated via subsampling or simulation. Intermediate order methods use ratios of statistics that are asymptotically in the tail but not the most extreme. We construct a ratio statistic that is asymptotically standard normal. This ratio requires no tuning parameters, in contrast to conventional intermediate order approaches \citep[ch. 3]{DeHaan2006}. Our framework complements \cite{Jochmans2019}, who develop central order approximations methods for central quantiles of UIH. \label{page:editor:intermedite-novelty}

	We show that inference is valid if and only if tail equivalence holds, under minimal assumptions on the marginal distributions of UIH and estimation noise. The result allows for complex dependence structures between noise and true UIH, an important feature in non-experimental settings \citep{Heckman2001,Browning2007,Browning2010}.

	For a broad class of distributions, we derive sufficient conditions for tail equivalence in terms of rates on the cross-sectional and individual sample sizes $N$ and $T$. These conditions require standard moment or normality assumptions on the noise and a lower bound on the EV index. If the true UIH distribution has an infinite tail, tail equivalence holds under mild restrictions. For instance, if the noise has at least eight finite moments, inference remains valid provided $N/T^4 \to 0$, matching the central quantile inference condition in \cite{Jochmans2019}. If the noise is Gaussian, $N$ may grow almost exponentially relative to $T$. In contrast, rate conditions are stricter for distributions with finite endpoints, depending on the relative heaviness of UIH and noise tails.
 
We propose a rule of thumb for choosing inference methods. Broadly, the rule depends on  $N$ and the quantile of interest. For smaller $N$, extreme order methods are preferable for tail quantiles. In larger samples, one may also use the simpler intermediate order methods.

 Our simulation studies show if rate conditions hold, our methods offer favorable coverage properties in the tails. The rate conditions are important, and their failure may lead to  distorted inference.  We also show that our rule of thumb for method choice performs well. 
  
	We illustrate our methodology with an application to firm productivity in denser and less dense areas in the setting of \cite{Combes2012b}. Our analysis addresses two key aspects of their study. First, we examine firm selection, which is hypothesized to left-truncate the productivity distribution by imposing a minimum survival threshold. We find no evidence of such truncation, reinforcing the conclusion of \cite{Combes2012b} that any truncation must be identical across areas. Second, \cite{Combes2012b} assume that productivity distributions differ only in mean, variance, and truncation. We nonparametrically show that their tails are similar after adjusting for mean and variance, lending credence to that assumption.
 
This paper contributes to several strands of literature. First,  it relates to recent work  on  distributional properties of UIH \citep{Arellano2012IdentifyingDistributionalCharacteristics, Okui2019, Barras2021, Sasaki2022, Jochmans2019}.  \label{page:editor:lit-review}
 First, \cite{Jochmans2019} develop results for central quantiles of UIH, but their approach relies on the central limit theorem for quantiles and thus unsuitable for extremes.  Second, to the best of our knowledge, \cite{Sasaki2022} is the only paper that discusses inference on extreme quantiles of UIH. Specifically, they consider extreme quantiles of coefficients in a simple linear model and describe a high-level condition for validity of inference based on estimates. We focus on the general problem of inference using noisy estimates that covers linear, nonlinear, and nonparametric estimators. We obtain   necessary and sufficient conditions under which such estimates are informative, along with explicit rate conditions.
Second, our extreme order approximations relate to fixed-$k$ tail inference methods \citep{Muller2017, Sasaki2022}, though our approach does not require bounds the EV index or optimization.  Third, we contribute the literature on intermediate order inference \citep{DeHaan2006, Davison2015StatisticsExtremes}). We show that one can use self-normalized ratios of intermediate order statistics to obtain an asymptotically pivotal statistic that involves no tuning parameters. Such a construction is novel both in the noiseless and the noisy cases. Finally, our results contrast with work on extreme conditional quantiles and treatment effects \citep{Chernozhukov2005, Chernozhukov2011, Zhang2018}, which focus on the extreme (conditional) quantiles  of observable rather than latent heterogeneity.

The rest of the paper is organized as follows. Section \ref{section:setup} formalizes our setup and  assumptions. In section \ref{section:distributions} we lay  down the probabilistic foundations of our inference theory by proving  extreme and intermediate extreme value theorems for noisy estimates. Building on these results, in section \ref{section:inference} we discuss three approaches to inference. 
  In section \ref{section:simulations} we explore performance of our CIs in a Monte Carlo setting.
Section \ref{section:empirical} contains the empirical application.
 All proofs   are collected in the appendix. We provide additional results    in the Online Appendix, available from the author's website.

\section{Setting and Assumptions}\label{section:setup}

\subsection{Problem statement}

Suppose that unobserved individual heterogeneity $\theta_i$  is sampled in an IID fashion   from some distribution $F$.   {$\theta_i$ may be a  treatment effect, effect size, a coordinate of a vector of individual-specific coefficients,  value of a function at a point, etc.}  $i$ indexes cross-sectional units or individual studies; in what follows  we will refer to them as ``units''.

  Our goal is to estimate quantiles $F^{-1}(q)$ 
  where $q$ is close to 0 or 1, and to conduct inference on them.   {In particular, we are interested in constructing confidence intervals $F^{-1}(q)$, along with  hypothesis tests for hypotheses like $H_0: \theta_i\geq 0$, which is equivalent to  $H_0: F^{-1}(0)\geq 0$}.  Without loss of generality we focus on the right tail of $F$.  
  
  \subsection{Data generating process}

We do not observe $\theta_i$ directly; instead we only see noisy observations $\vartheta_{i, T}$:
\begin{equation}\label{equation:noisyThetaDefinition}
\vartheta_{i, T}= \theta_i + \dfrac{1}{T_i^{p}}\varepsilon_{i, T_i}, \quad i=1, \dots, N, 
\end{equation}
 where $T_i$ is the sample size available for unit $i$,  $\varepsilon_{i, T_i}$  
is the scaled estimation error, $\varepsilon_{i, T}=O_p(1)$ for all $T$,  and $p>0$ is the convergence rate the estimator.
For clarity of exposition, in what follows we assume that for all units $T_i=T$ for some common $T$; in   \textcolor{black}{the Online Appendix} we show that  our results   extend to the unbalanced case. 
$p$   is determined by the estimation method used in a given case. For estimators convergent at the parametric rate $T^{-1/2}$, $p$ is equal to $1/2$, but we allow other rates. $\E[\varepsilon_{i, T}]$ may be nonzero and need not converge to 0 as $T\to\infty$, but   bias may not diverge.

Representation \eqref{equation:noisyThetaDefinition} is   compatible with any estimation method with a known rate of convergence;  $\varepsilon_{i, T}$ can always be   defined 	 via the identity $\varepsilon_{i, T}= T^p(\vartheta_{i, T}-\theta_i)$.  
  Intuitively, $\vartheta_{i, T}$ are    estimators of $\theta_i$ that have growing precision. $\vartheta_{i, T}$ can potentially be     biased. Our setting nests the setup considered by \cite{Jochmans2019} and the typical setup of 
  meta-analysis (see e.g. \cite{Higgins2009}).  
  
  We provide several examples of how $\vartheta_{i,  T}$ may be constructed.

 \begin{example}[Unit-wise OLS and IV]
 	\label{example:IV}
 Let $y_{it} =\theta_i x_{it} + u_{it}$. First suppose that $\E[u_{it}x_{it}]=0$.   We estimate heterogeneous coefficients $\theta_i$   by unit-wise OLS: $\vartheta_{i, T}=  (T^{-1}\sum_{t=1}^T x_{it}^2  )^{-1}T^{-1}\sum_{t=1}^T x_{it}y_{it}$. Define $\varepsilon_{i,T} =  (T^{-1}\sum_{t=1}^T x_{it}^2  )^{-1}T^{-1/2}\sum_{t=1}^T x_{it}u_{it} $ to write $\vartheta_{i, T}$ as in eq. \eqref{equation:noisyThetaDefinition} for $p=1/2$.   
 Now  let $x_{it}$ be endogenous in the sense that $\E[u_{it}x_{it}]\neq 0$. Suppose a valid instrument $z_{it}$ is available: $\E[z_{it}u_{it}]=0$, $\E[z_{it}x_{it}]\neq 0$. The unit-wise IV estimators are given by $\vartheta_{i, T} = (T^{-1}\sum_{t=1}^T z_{it}x_{it}  )^{-1}T^{-1}\sum_{t=1}^T z_{it}y_{it}$.  In this case $\varepsilon_{i, T}  =  (T^{-1}\sum_{t=1}^T z_{it}x_{it}  )^{-1}T^{-1/2}\sum_{t=1}^T z_{it}u_{it}$. Note that the distribution of $\varepsilon_{i, T}$ may be strongly  dissimilar to the normal  distribution even for large $T$ if $z_{it}$ is not a strong instrument \citep{Nelson1990}.
 \end{example}

 \begin{example}[Nonparametric regression]
 	\label{example:nonparametric}
 	Let   $\theta_i\coloneqq \E[y_{it}|x_{it}=x_0]$ for some fixed value $x_0$. 
 	Let $\vartheta_{i, T} =  (\sum_{t=1}^T K\left({(x_{it}-x_0)}/{h} \right)  )^{-1}{\sum_{t=1}^T y_{it}K ({(x_{it}-x_0)}/{h}  )}{}$ by the Nadaraya-Watson estimator of $\theta_i$ where $K$ is a kernel function and $h$ is a bandwidth parameter.  
 	Let $u_{it} = y_{it} -\theta_i$, and set $\varepsilon_{i, T} =   (\sum_{t=1}^T K ({x_{it}-x_0}/{h}  )  )^{-1}{\sqrt{Th}\sum_{t=1}^T u_{it}K ({(x_{it}-x_0)}/{h}  )}$.   Let $h=T^{-s}, s\in (0, 1)$.  It holds that $\varepsilon_{i, T}= O_p(1)$ and  $\vartheta_{i, T} = \theta_i +T^{-(1-s)/2} \varepsilon_{i, T}$  under suitable conditions on $h$. If $h$ is picked to optimize the convergence rate,   $\varepsilon_{i, T}$   has a non-zero mean even in the limit. 
 \end{example}

\paragraph{Notation}

Let $\vartheta_{1, N, T}\leq \dots \leq \vartheta_{N, N, T}$ be the  order statistics of $\curl{\vartheta_{1, T}, \dots, \vartheta_{N, T}}$, and similarly let $\theta_{1, N}\leq \dots \leq \theta_{N, N}$ be the order statistics of the latent noiseless $\curl{\theta_1, \dots, \theta_N}$.  
$\Rightarrow$ denotes weak convergence, both of random variables and functions.

 \subsection{Assumptions}
\label{subsection:assumptions}

\begin{asm}
	\label{assumption:independence}
	For each $T$, 	$\curl{(\theta_i, \varepsilon_{i, T})}_{i=1, \dots, N}$ are  independent and identically distributed random vectors indexed by $i$.   
\end{asm}

Observations $\vartheta_{i, T}$ are sampled in an IID fashion.  Note that  $\vartheta_{i, T}$ can be  conditionally heteroskedastic: $\var(\vartheta_{i,T}|\theta_i)=\var(\varepsilon_{i, T}|\theta_i)$ can depend on $\theta_i$, provided this   variance exists.


\begin{asm}[]\label{assumption:EV}
	
The distribution	$F$ of $\theta$ is in the weak domain of attraction of an extreme value distribution with extreme value index $\gamma\in \R$.
	
%
 
\end{asm}

Under assumption \ref{assumption:EV}, the classical extreme value theorem  of \cite{Gnedenko1943} applies to the latent noiseless distribution $F$. 
  This will serve as the basis for extending this extreme value convergence to the observed noisy data. Without assumption \ref{assumption:EV}, we would not be able to  conduct  inference using the asymptotic behavior of the sample maximum even if we had access to the true $\theta_i$.
Assumption
\ref{assumption:EV} is a   mild assumption, satisfied by almost all textbook continuous distributions and many discrete ones. Assumption \ref{assumption:EV} is equivalent to certain regular variation conditions on $F$  \citep[theorem 1.2.1]{DeHaan2006}.

Two  key features of our analysis are that we do not restrict the dependence structure between $\theta_i$ and $\varepsilon_{i, T}$ and that we impose no   distributional assumptions on $\varepsilon_{i, T}$.  This approach is motivated by the following practical challenges.
First,  $\theta_i$ and $\varepsilon_{i, T}$ will typically  be related in a complex and unobservable manner  outside of tightly controlled experimental settings. Such dependence may arise if the agent chooses some  covariates with knowledge of $\theta_i$, and these covariates are in turn used in estimation of $\theta_i$.
Second, the distribution of $\varepsilon_{i, T}$ might not be well-approximated by a normal distribution, as in  the IV case of example \ref{example:IV} (though we show  how to exploit normality of $\varepsilon_{i, T}$).

 Instead, we only impose a weak assumption on the marginal distribution $G_T$ of $\varepsilon_{i, T}$.
\begin{asm}[Tightness of $\curl{G_T}_{T=1}^{\infty}$]
	\label{assumption:tightness}
	For any $\varepsilon>0$ there exists some $C_{\varepsilon}\geq 0$ such that for all $T$ it holds that $G_T(-C_{\varepsilon})\leq \varepsilon$ and $1-G_T(C_{\varepsilon})\leq \varepsilon$. 
\end{asm}
Intuitively, assumption \ref{assumption:tightness} 
requires that $\varepsilon_{i, T}$ be defined in such a way that, as $T\to\infty$, the distributions  $G_T$ of $\varepsilon_{i, T}$ do not escape to infinity.
  Assumption \ref{assumption:tightness} holds  automatically if $T^p(\vartheta_{i, T}-\theta_i)$ has a non-degenerate asymptotic distribution. 
Together with definition \eqref{equation:noisyThetaDefinition}, assumption \ref{assumption:tightness}  implies that each $\vartheta_{i, T}$ is consistent for $\theta_i$, but
 we allow $\vartheta_{i, T}$ to be biased in finite samples.
In addition, assumption \ref{assumption:tightness} allows the mean of $\varepsilon_{i, T}$ to be nonzero in the limit, as may occur in nonparametric regression with MSE-minimizing choice of bandwidth (example \ref{example:nonparametric}).

%
%
%

  \section{Extreme Value Theory For Noisy Estimates}\label{section:distributions}

\subsection{Extreme Value Theorem For Noisy Estimates}

The key step towards inference on extreme quantiles is to establish distributional results for the sample maximum  $\vartheta_{N, N, T}$.  The following result gives necessary and sufficient conditions under which   $\vartheta_{N, N, T}$ and the latent maximum $\theta_{N, N}=\max\curl{\theta_1, \dots, \theta_N}$ have the same limit distribution. It introduces the notion of tail equivalence and serves as a stepping stone towards our inference results. 
  
 	\begin{theorem}
 		\label{theorem:evtNoisy}
 		Let assumption \ref{assumption:independence}  hold.
 		Let   constants $a_N, b_N$ and a random variable $X$ be  such that   ${(\theta_{N, N} - b_N)}/{a_N}\Rightarrow X$  as $N\to\infty$. Let $F$ be the cdf of $\theta_i$ and $G_{T}$ be the cdf of $\varepsilon_{i, T}$. 
 		\begin{enumerate}[noitemsep,topsep=0pt,parsep=0pt,partopsep=0pt, label={(\arabic*)}, leftmargin=*]
 			\item \emph{(Transferral of convergence)} Let    the following tail equivalence (TE) conditions hold: for each $\tau\in (0, \infty)$ as $N, T\to\infty$	there exists some $\epsilon\in(0, 1)$
 			\begin{align}
 				& \hspace{-30pt} 	\sup_{u\in \left[0, \epsilon\right]}  \dfrac{1}{a_N}\left( F^{-1}\left(1- \dfrac{1}{N\tau}-u \right)  -F^{-1}\left(1- \dfrac{1}{N\tau} \right) + \dfrac{1}{T^p} G_T^{-1}\left(u\right)  \right) \to 0, & \tag{TE-Sup}\label{equation:tailEquivalenceSup} \\
 				& \hspace{-30pt} 	\inf_{u\in \left[0, \frac{1}{N\tau} \right]}   \dfrac{1}{a_N}\left(F^{-1}\left(1- \dfrac{1}{N\tau} + u\right) -  F^{-1}\left(1- \dfrac{1}{N\tau} \right) +\dfrac{1}{T^p} G_T^{-1}\left(1-u\right) \right)  \to 0. & \tag{TE-Inf} \label{equation:tailEquivalenceInf}
 			\end{align}
 			Then  $
 			(\vartheta_{N, N, T}- b_N)/{a_N} \Rightarrow X$  as $N, T\to\infty$.
 			\item \emph{(Sharpness)} Consider the following TE condition: for each $\tau\in(0, \infty)$ as $N, T\to\infty$ 
 			\begin{equation}
 				\sup_{u\in \left[0, 1-\frac{1}{N\tau}\right]} \dfrac{1}{a_N}\left( F^{-1}\left(1- \dfrac{1}{N\tau}-u \right)  -F^{-1}\left(1- \dfrac{1}{N\tau} \right) + \dfrac{1}{T^p} G_T^{-1}\left(u\right)  \right) \to 0,  \tag{TE-Sup'}\label{equation:tailEquivalenceSupGeneral}
 			\end{equation}
 			Then (i)   \eqref{equation:tailEquivalenceSup} and  \eqref{equation:tailEquivalenceInf} together imply \eqref{equation:tailEquivalenceSupGeneral}; (ii)
 			  \eqref{equation:tailEquivalenceInf}  and \eqref{equation:tailEquivalenceSupGeneral} are sharp in the following sense:  if at least one of the conditions fails,  there exists a sequence of joint distributions of $(\theta_i, \varepsilon_{i, T})$ with given marginal distributions $F, G_T$ such that	$({\vartheta_{N, N, T}- b_N})/{a_N}$ weakly converges to a limit different from $X$ or does not converge at all.
 		\end{enumerate}

 	\end{theorem}
 
	Theorem \ref{theorem:evtNoisy} establishes the precise conditions under which the maximum of the noisy sample, $\vartheta_{N, N, T}$, inherits the  asymptotic distribution as the latent maximum $\theta_{N, N}$. The result is general, holding regardless of the dependence structure between $\theta_i$ and $\varepsilon_{i, T}$. If the tail equivalence (TE) conditions \eqref{equation:tailEquivalenceSupGeneral} and \eqref{equation:tailEquivalenceInf} fail, there exists at least one possible limiting distribution for $\vartheta_{N, N, T}$ that differs from that of $\theta_{N, N}$, and convergence may   occur only along a subsequence.

   At the core of this result are \eqref{equation:tailEquivalenceSup} and \eqref{equation:tailEquivalenceInf}, which characterize the relationship between the right tails of the laws of $\vartheta_{i, T}$ and $\theta_i$. These conditions ensure that, asymptotically, the noisy sample provides the same extreme value information as the unobserved noiseless sample. The expressions in \eqref{equation:tailEquivalenceSup} and \eqref{equation:tailEquivalenceInf} reflect a pointwise  limit equivalence of tails: the term $F^{-1} + T^{-p}G^{-1}$ captures the approximate quantiles of the noisy observations, while $(-F^{-1})$ are the quantiles of the target distribution. The supremum and infimum adjust for the unknown dependence between $\theta_i$ and $\varepsilon_{i, T}$.

These TE conditions provide a general framework for extreme value analysis in the presence of noise. They accommodate a broad range of estimators for $\theta_i$ (see eq. \eqref{equation:noisyThetaDefinition}), including linear, nonlinear, and nonparametric estimators. In addition, they sharpen and generalize the sufficient condition derived by \cite{Sasaki2022} for the special case of univariate linear regression (see remark \ref{remark:relation-to-sasaki-wang-conditions} below).  \label{page:editor:evt-novelty}

   We highlight four   aspects of conditions \eqref{equation:tailEquivalenceSup} and \eqref{equation:tailEquivalenceInf}. First, the infimum in \eqref{equation:tailEquivalenceInf} is always greater than or equal to the supremum values in \eqref{equation:tailEquivalenceSup} and \eqref{equation:tailEquivalenceSupGeneral}, as shown in the proof. Second, the conditions hold automatically in the absence of estimation noise. In this case the function under the supremum in \eqref{equation:tailEquivalenceSup} is non-positive for all admissible \( u \), and equal to zero only at \( u = 0 \). Likewise, the function under the infimum in \eqref{equation:tailEquivalenceInf} is non-negative and zero only at \( u = 0 \). 
   Third, when estimation noise is present, similar logic applies but the impact on the  conditions is asymmetric : if \( G_T(0) \in (0,1) \), the function under the infimum is eventually always non-negative, with \eqref{equation:tailEquivalenceInf} requiring its minimum to converge to zero. Meanwhile, the function under the supremum in \eqref{equation:tailEquivalenceSup} can be negative over some range of \( u \) but may change sign at larger values. Here, \eqref{equation:tailEquivalenceSup} requires that the supremum converge to zero, whether from above or below. We stress that this is not a requirement of uniform convergence. Finally, for sharpness, the weaker condition \eqref{equation:tailEquivalenceSupGeneral} is necessary. However, for practical purposes, \eqref{equation:tailEquivalenceSup} is sufficient and more interpretable; it may be viewed as  controlling the contribution of the left tail of \( G_T \).

\vspace{1mm}

When do    \eqref{equation:tailEquivalenceInf}  and \eqref{equation:tailEquivalenceSup} hold?  Proposition \ref{proposition:ratesExtremeHalfUniformity} provides sufficient conditions, expressed as constraints on the growth rate of $N$ relative to $T$, given assumptions on the tails of $G_T$. Such conditions depend on $\gamma$, the extreme value index of $F$, which is typically unknown. However, they remain valid when only a lower bound $\gamma' \leq \gamma$ is available. Intuitively, when $\gamma$ is large, $F$ has a heavier tail, making the contribution of $G_T$ less pronounced and allowing for larger values of $N$.

 \begin{proposition}
 	\label{proposition:ratesExtremeHalfUniformity}
 	Let assumptions \ref{assumption:independence} and \ref{assumption:tightness} hold. Let one  of the following conditions hold:
 	\begin{enumerate}[noitemsep,topsep=0pt,parsep=0pt,partopsep=0pt, label={(\arabic*)}, leftmargin=*]
 		\item Let $\sup_T\E\abs{\varepsilon_{i, T}}^{\beta}<\infty$ for some $\beta>0$, and  let  ${N^{1/\beta-\gamma'}(\log (T))^{1/\beta}}/{T^{p}}\to 0 $ for some $\gamma'$.  
 		\item For all $T$, let $\varepsilon_{i, T}\sim N(\mu_T, \sigma_T^2)$, and  let  ${N^{-\gamma'} \sqrt{\log (N)}}/{T^{p}}\to 0 $ for some $\gamma'$.   
 	\end{enumerate} 
 	In addition, let $F$ satisfy assumption \ref{assumption:EV} with EV index $\gamma>\gamma'$. Then   \eqref{equation:tailEquivalenceInf} and  \eqref{equation:tailEquivalenceSup} hold for $F$ and $G_T$ {for any sequence $\curl{(a_N, b_N)}_{N=1}^{\infty}$ such that $(\theta_{N, N}-b_N)/a_N$ converges to a non-degenerate random variable.}
 	
 \end{proposition}
  
 	Proposition \ref{proposition:ratesExtremeHalfUniformity} highlights two key factors that determine how restrictive these conditions are. First, heavier tails of $F$ (larger $\gamma'$) permit a larger $N$ relative to $T$. Second, lighter tails of $G_T$ result in milder constraints on $N$. The conditions on $G_T$ are captured by (1) and (2). Condition (1) requires only that $G_T$ has uniformly bounded $\beta$th moments, while (2) assumes exact normality, which can be viewed as a limiting case of (1) as $\beta \to \infty$. In practice, $G_T$ is expected to lie between these two extremes, with its tails becoming lighter as $T \to \infty$.

 		A useful comparison can be made to the conditions for inference on central quantiles established by \cite{Jochmans2019}. They show that $N/T^4 \to 0$ is sufficient for validity when $p = 1/2$, under broad conditions on $F$ (see theorem \ref{theorem:JW}). A similar rate condition arises in our extreme   setting when $G_T$ has more than eight moments and $\gamma \geq 0$ (as is necessarily the case if $F$ has an infinite tail).\footnote{It is fairly straightforward to impose conditions under which $G_T$ has a given number of moments for linear estimators, such as those of examples \ref{example:IV} and \ref{example:nonparametric}. Further, \cite{Brownlees2022} provide some conditions on existence of moments of estimation error for nonlinear estimators.}  \label{page:ref2:minor-comment-1}  Proposition \ref{proposition:ratesExtremeHalfUniformity} also allows for nearly exponential growth of $N$ relative to $T$ when $\varepsilon_{i,T}$ is normally distributed and $\gamma \geq 0$. This suggests that valid inference on extreme quantiles may still be feasible in cases where central quantile inference is not. Conversely, when $\gamma < 0$, the sufficient conditions for TE may become more restrictive than those for central quantiles.
 	
 	Importantly, the TE conditions differ   from the \cite{Jochmans2019} conditions in their underlying mechanisms. The latter are primarily concerned with the rate of estimation noise bias, while the TE conditions control the relative contribution of the tail of $F$ and the tails of $G_T$. As a result, TE conditions benefit from lighter tails of $G_T$, and they hold trivially in the noiseless setting, as noted after Theorem \ref{theorem:evtNoisy}.

	\begin{remark}[Relation to \cite{Sasaki2022}]\label{remark:relation-to-sasaki-wang-conditions}
		\cite{Sasaki2022} analyze the special    case of univariate linear regression with heterogeneous coefficients (OLS in example \ref{example:IV}) and derive a condition similar to those in Proposition \ref{proposition:ratesExtremeHalfUniformity}. Their condition 2.7 aligns closely with our result, modulo certain high-level uniformity assumptions which concern averages of residuals and covariates. In this linear case, Proposition \ref{proposition:ratesExtremeHalfUniformity} quantifies these uniformity conditions and shows that they impose an additional restriction that depends on $G_T$, as reflected in the $N^{1/\beta}$ or $\log(N)$ terms. 
	\end{remark}

\begin{remark}\label{remark:lowerBoundGammaEconomics}
	
	An appropriate value of $\gamma'$ might often be apparent in a given application. For example, \cite{Gabaix2009, Gabaix2016} documents that many economic relations follow a power law and    outlines some general theoretical mechanisms under which a power law arises. If such a mechanism is likely to hold in a given situation, it is  reasonable to assume that the distribution of the  data is well approximated by a power law.   Since a power law distribution has to have $\gamma>0$, it is sufficient to check the hypothesis of proposition \ref{proposition:ratesExtremeHalfUniformity} with $\gamma'=0$.  Further, to allow the distribution of $\theta$ to potentially have a light  infinite tail, it is sufficient to check conditions for any $\gamma'<0$, potentially arbitrarily close to 0.
	
\end{remark}

	\subsection{Intermediate Order Statistics}\label{section:ivt}
 
To conduct inference on extreme quantiles, we also need to develop an asymptotic theory for intermediate order statistics.  
Formally,  $\theta_{N-k(N), N}$ is called an intermediate   order statistic if $k(N)\to \infty$ as $N\to\infty$ and $k(N)=o(N)$. Intuitively, such statistics asymptotically stay in the tail, but are not the top statistics. We generally suppress dependence of $k$ on $N$.

 To derive asymptotic properties of intermediate order statistics, we   impose an additional assumption on  $F$ that refines assumption \ref{assumption:EV}.

\begin{asm}[$F$ satisfies a first order von Mises condition] \label{assumption:vonMises}
\emph{$F$ is twice  differentiable with density $f$, $f$ positive in some left neighborhood of $F^{-1}(1)$ ($F^{-1}(1)$ might be finite or infinite),  and for some $\gamma\in \R$ it holds that  $
	\lim_{t\uparrow F^{-1}(1)} \left( [1-F]/{f}\right)'(t) =\gamma.$ }

\end{asm}

Assumption \ref{assumption:vonMises} is a slight strengthening of assumption \ref{assumption:EV}.  See  \cite{Dekkers1989} and  \cite{DeHaan2006} for a discussion of the condition and its plausibility.

 We now state a theorem describing the asymptotic behavior of intermediate order statistics:  

\begin{theorem}[Intermediate value theorem (IVT)] 
	\label{theorem:intermediateNormality}
	Let assumptions \ref{assumption:independence}  and \ref{assumption:vonMises} hold and let $k=k(N)\to\infty, k=o(N)$ as $N\to\infty$.  Define $c_N$ as  the derivative of the inverse of $1/(1-F)$ evaluated at $N/k$ and multiplied by $N/k$, that is, $c_N=  (N/k)\times  \left[\left(\left({1}/({1-F})  \right)^{-1}\right)'\left( N/k \right)\right]$.
		Let $U_1, \dots, U_N$ be iid Uniform[0, 1] random variables. 
		\begin{enumerate}[noitemsep,topsep=0pt,parsep=0pt,partopsep=0pt, label={(\arabic*)}, leftmargin=*]
			\item \emph{(Transferral of convergence)}  	Let   the following tail equivalence conditions hold as  $N, T\to\infty$ for some $\epsilon>0$:
			\begin{align}
				&  \hspace{-30pt}  \sup_{u\in \left[0, \epsilon\right]}  \dfrac{\sqrt{k}}{c_N} \left( F^{-1}(1-U_{k, N}-u)  -F^{-1}\left(1- U_{k, N} \right)  + \dfrac{1}{T^p} G_T^{-1}\left( u\right) \right) \xrightarrow{p} 0, \label{equation:intermediateConditionRandomSup} \\
				& \hspace{-30pt} 	 \inf_{u\in \left[0, U_{k, N}\right]}  \dfrac{\sqrt{k}}{c_N }  \left(F^{-1}\left(1- U_{k, N} + u \right)  - F^{-1}\left(1-  U_{k, N} \right) +\dfrac{1}{T^p} G_T^{-1}\left(1-u\right)\right)   \xrightarrow{p} 0. \label{equation:intermediateConditionRandomInf}
			\end{align}
			Then as $N, T\to\infty$
			\begin{equation}\label{equation:intermediateTheoremStatistic}
				\dfrac{\sqrt{k}}{c_N} \left(\vartheta_{N-k, N, T} - F^{-1}\left(1-  k/N\right) \right)\Rightarrow N(0, 1).
			\end{equation}
			\item \emph{(Sharpness)} Consider the following   condition: 
			\begin{equation}
				  \sup_{u\in \left[0, 1-U_{k, N	}\right]}  \dfrac{\sqrt{k}}{c_N} \left( F^{-1}(1-U_{k, N}-u)  -F^{-1}\left(1- U_{k, N} \right)  + \dfrac{1}{T^p} G_T^{-1}\left( u\right) \right) \xrightarrow{p} 0, \label{equation:intermediateConditionRandomSupGeneral}
			\end{equation}	Conditions \eqref{equation:intermediateConditionRandomInf} and \eqref{equation:intermediateConditionRandomSupGeneral}   are sharp in the sense of theorem \ref{theorem:evtNoisy}.
		\end{enumerate}
	 
\end{theorem}

 Theorem \ref{theorem:intermediateNormality} is the intermediate order counterpart of theorem \ref{theorem:evtNoisy};  conditions \eqref{equation:intermediateConditionRandomSup}, \eqref{equation:intermediateConditionRandomInf}, and \eqref{equation:intermediateConditionRandomSupGeneral} are  analogs of   \eqref{equation:tailEquivalenceSup}, \eqref{equation:tailEquivalenceInf}, and \eqref{equation:tailEquivalenceSupGeneral}, respectively.   The two sets of conditions differ in the region where tail equivalence is imposed. Theorem \ref{theorem:evtNoisy}  concerns quantiles of the form $1- 1/(N\tau)$, $\tau>0$ fixed,  whereas theorem \ref{theorem:intermediateNormality} looks at quantiles of the order $1-k/N$.    Since $k\to\infty$, the two regions are asymptotically distinct.  A second point of  difference between   the two pairs of conditions  is that conditions of theorem \ref{theorem:intermediateNormality} take a randomized form, though proposition \ref{proposition:ratesIntermediateHalfUniformity} below provides deterministic sufficient conditions.   As before, the conditions allow for general   dependence structures and natures of $\varepsilon_{i, T}$. To the best of our knowledge, there are no directly comparable results in the literature.\footnote{\label{page:editor:footnote-girard-paper} \cite{Girard2021} obtain a somewhat related condition for estimating intermediate order expectiles in regressions with heavy-tailed noise. Their condition requires that estimation noise be uniformly small (see their condition (2)). This condition is sufficient for transferral of convergence for any order statistic regardless of order (modulo the scaling factors involved). In contrast, Theorem \ref{theorem:intermediateNormality} targets only the differences between intermediate order statistics and does not restrict the difference between central or top order statistics.}
 
 \vspace{1mm}

As for \eqref{equation:tailEquivalenceInf} and \eqref{equation:tailEquivalenceSup}, sufficient conditions for \eqref{equation:intermediateConditionRandomSup} and \eqref{equation:intermediateConditionRandomInf} take form of rate restrictions on $N$ and $T$ that depend on the EV index $\gamma$.  A sufficient condition is possible if a lower bound for $\gamma$ is available. The following proposition  is an analog of proposition  \ref{proposition:ratesExtremeHalfUniformity}.

 \begin{proposition}
 	\label{proposition:ratesIntermediateHalfUniformity}
 	Let assumptions \ref{assumption:independence} and \ref{assumption:tightness} hold. Let     $\delta\in (0, 1)$. Let one  of the following conditions hold  
 	\begin{enumerate}[noitemsep,topsep=0pt,parsep=0pt,partopsep=0pt, label={(\arabic*)}, leftmargin=*]
 		\item   $\sup_T\E\abs{\varepsilon_{i, T}}^{\beta}<\infty$   and for some $\nu>0, \gamma'$  it holds that  ${ N^{\delta/2(1+1/\beta) + (1-\delta)(-\gamma'+1/\beta) + \nu/\beta }}/{T^{p}} \to 0 $
 		\item  For all $T$ let $\varepsilon_{i, T}\sim N(\mu_T, \sigma_T^2)$, and for some $\gamma'$ it holds that   ${N^{\delta/2 + (1-\delta)(-\gamma')} \sqrt{\log (N)}}/{T^p}\to 0$.
 	\end{enumerate} 
 	In addition, let  $F$ satisfy assumption \ref{assumption:vonMises} with EV index $\gamma>\gamma'$. Then  conditions \eqref{equation:intermediateConditionRandomSup} and  \eqref{equation:intermediateConditionRandomInf} hold for $F$ and $G_T$ for  $k=N^{\delta}$.

 \end{proposition}

\begin{remark}[Comparison of sufficient conditions for the EVT and the IVT]
	\label{remark:comparisonConditionsEVTIVT}
	
	Depending on $\gamma$, the rate conditions for the EVT may be more or less restrictive than the conditions for the IVT. 
 For example, suppose that $\varepsilon_{i\, T}$ is normally distributed. 
If $\gamma\leq -1/2$, then the EVT condition implies the IVT condition; the opposite holds if $\gamma>-1/2$, regardless of the value of $\delta$.
In particular, if $\gamma>0$, for the EVT there are no restrictions on relative sizes of $N$ and $T$, but there are  restrictions for the IVT.

\end{remark}

\begin{remark}[Dependence on $\delta$]
	\label{remark:dependenceonDelta}
	
	Conditions   of proposition \ref{proposition:ratesIntermediateHalfUniformity} depend on $\delta$, the parameter that determines the magnitude of $k=N^{\delta}$. If $\delta$ is close to zero, conditions for the IVT are close to those for the EVT.  
	Intuitively, in this case conditions \eqref{equation:intermediateConditionRandomSup} and  \eqref{equation:intermediateConditionRandomInf} require asymptotic tail equivalence in approximately the same section of the tail as conditions \eqref{equation:tailEquivalenceSup} and \eqref{equation:tailEquivalenceInf}, and so the resulting sufficient conditions are similar.
 As $\delta$ grows, the region controlled by conditions \eqref{equation:intermediateConditionRandomSup} and  \eqref{equation:intermediateConditionRandomInf} becomes distinct from the right endpoint of the distribution  (while still staying the tails by requirement that $k=o(N)$).

\end{remark}

  \section{Inference   Using Noisy Estimates}\label{section:inference}

We now turn to  inference.  In this section, we introduce confidence intervals (CIs), estimators,  and tests based on  extreme and intermediate order asymptotic approximations.  We  also briefly discuss the central order approximations of \cite{Jochmans2019}.

The extreme, intermediate, and central order approximations differ in their appropriate use case.  Let $\tau$ be the quantile of interest.
Based on the simulation evidence of section \ref{section:simulations} and the Online Appendix, we propose the following rule of thumb that is valid for all values of $N$:
\begin{enumerate}[noitemsep,topsep=0pt,parsep=0pt,partopsep=0pt, label={(\arabic*)}, leftmargin=*]
	\item If $(1-\tau)N\leq 100$ and $\tau\geq 0.9$, we recommend extreme order approximations --- the approach based on theorem \ref{theorem:feasibleEVT} below.  In this case  fewer than 100 order statistics lie to the right the sample $\tau$th quantile, and the central limit theorem for quantiles is unlikely to provide a good approximation.
	\item If $(1-\tau)N>100$, we recommend using a central order approximation. Specifically, we recommend basing inference on theorem \ref{theorem:JW} below.
	
\end{enumerate}
In larger cross-sections ($N\gtrsim 10000$), the intermediate order approach of  theorem \ref{theorem:feasibleIVTnoiseless}  is a viable and  particularly simple-to-compute option for $\tau\geq [0.9, 1)$, provided the corresponding value of  $k\equiv (1-\tau)N$ satisfies $k\geq 100$.

  \subsection{Inference Using Extreme Order Approximations}
  \label{subsection:evtInference}
  
  Extreme order approximations use theorem \ref{theorem:evtNoisy} as the basis for inference.   The quantile of interest is modeled as drifting to 1 at a rate proportional to $N^{-1}$.
Formally,  we select $b_N = F^{-1}(1-l/N)$  in theorem \ref{theorem:evtNoisy} for some fixed $l>0$. 
  Further, we make an explicit choice of scaling constants $a_N$. Then the following version of theorem \ref{theorem:evtNoisy} holds:

 \begin{lem}
	
	\label{corollary:evtNoisyCenteredQuantile} Let assumptions of theorem \ref{theorem:evtNoisy} hold.  Let $l>0$ be fixed.   Let $E_1^*$ be a standard exponential random variable.
	\begin{enumerate}[noitemsep,topsep=0pt,parsep=0pt,partopsep=0pt, label={(\arabic*)}, leftmargin=*]
		\item  If $F$ satisfies assumption \ref{assumption:EV} with EV index $\gamma>0$, then
		\begin{equation}
		\dfrac{1}{ F^{-1}\left(  1 - \frac{1}{N} \right)} \left[  \vartheta_{N, N, T}- F^{-1}\left( 1-\dfrac{l}{N} \right)\right] \Rightarrow  (E_1^*)^{-\gamma} - l^{-\gamma}\text{ as } N, T\to\infty.
		\end{equation}
		\item  If $F$ satisfies assumption \ref{assumption:EV} with EV index $\gamma=0$, there for some positive function $\hat{f}$  
		\begin{equation} 
		\dfrac{1}{\hat{f}\left(F^{-1}\left(  1 - \frac{1}{N} \right)\right)} \left( \vartheta_{N, N, T} - F^{-1}\left(  1 - \frac{l}{N} \right) \right)  \Rightarrow -\log(E_1^*) - \log(l)\text{ as } N, T\to\infty.
		\end{equation}
		\item If $F$ satisfies assumption \ref{assumption:EV} with EV index $\gamma<0$, then   $F^{-1}(1)<\infty$ and
		\begin{equation} 
		\dfrac{1}{F^{-1}(1) - F^{-1}\left(1- \frac{1}{N} \right) }\left(\vartheta_{N, N, T} -  F^{-1}\left(  1- \frac{l}{N}   \right) \right)
		\Rightarrow  -(E_1^*)^{-\gamma} +  l^{-\gamma}\text{ as } N, T\to\infty.
		\end{equation}
		
	\end{enumerate}

\end{lem}

 Lemma \ref{corollary:evtNoisyCenteredQuantile} cannot be used for inference directly, as  the scaling constants involved are unknown.  These constants   involve the $(1-1/N)$th quantile of $F$, and so are not covered by sample information. 
 
 To address this challenge, we first establish an intermediate result. The following lemma extends theorem \ref{theorem:evtNoisy} to cover a vector of $q$ top  order statistics  for $q$ fixed.

  \begin{lem}[Joint EVT]
  	\label{theorem:jointEVT} Let assumptions of theorem \ref{theorem:evtNoisy} hold. Let  $q$ be a fixed natural number 	and $E_1^*, \dots, E_{q+1}^*$ be iid standard exponential random variables. 
  	\begin{enumerate}[noitemsep,topsep=0pt,parsep=0pt,partopsep=0pt, label={(\arabic*)}, leftmargin=*]

  		\item If $F$ satisfies assumption \ref{assumption:EV} with  $\gamma>0$, then as $N, T\to\infty$
  		\begin{multline} 
  		\begin{pmatrix}
  		\dfrac{\vartheta_{N, N, T}}{ F^{-1}(1-1/N)}, \dfrac{\vartheta_{N-1, N, T} }{ F^{-1}(1-1/N)}, \dots, \dfrac{\vartheta_{N-q, N, T}  }{  F^{-1}(1-1/N)}
  		\end{pmatrix} \\ \Rightarrow \begin{pmatrix}
  		{(E_1^*)^{-\gamma} }, {(E_1^*+E_2^*)^{-\gamma} }, \dots, {(E_1^*+E_2^*+ \dots + E^*_{q+1})^{-\gamma}  }
  		\end{pmatrix}.
  		\end{multline}

  		\item If $F$ satisfies assumption \ref{assumption:EV} with  $\gamma=0$, then as $N, T\to\infty$   	for $\hat{f}$ as in lemma \ref{corollary:evtNoisyCenteredQuantile}.
  		\begin{multline} 
  		\begin{pmatrix}
  		\dfrac{  \vartheta_{N, N, T} - F^{-1}\left(  1 - \frac{1}{N} \right)}{\hat{f}\left(F^{-1}\left(  1 - \frac{1}{N} \right)\right)} , 	\dfrac{  \vartheta_{N-1, N, T} - F^{-1}\left(  1 - \frac{1}{N} \right)}{\hat{f}\left(F^{-1}\left(  1 - \frac{1}{N} \right)\right)}, \dots, 	\dfrac{  \vartheta_{N-q, N, T} - F^{-1}\left(  1 - \frac{1}{N} \right)}{\hat{f}\left(F^{-1}\left(  1 - \frac{1}{N} \right)\right)} 
  		\end{pmatrix} \\ \Rightarrow \begin{pmatrix}
  		{-\log(E_1^*) }, {-\log(E_1^*+E_2^*)}, \dots, {-\log(E_1^*+E_2^*+ \dots + E^*_{q+1}) }
  		\end{pmatrix}.
  		\end{multline}
  	
  	  		\item If $F$ satisfies assumption \ref{assumption:EV} with EV index $\gamma<0$, then as $N, T\to\infty$
  	\begin{multline} 
  	\begin{pmatrix}
  	\dfrac{\vartheta_{N, N, T}-F^{-1}(1)}{F^{-1}(1) - F^{-1}(1-1/N)}, \dfrac{\vartheta_{N-1, N, T} -F^{-1}(1)}{F^{-1}(1) - F^{-1}(1-1/N)}, \dots, \dfrac{\vartheta_{N-q, N, T	} -F^{-1}(1)}{F^{-1}(1) - F^{-1}(1-1/N)}
  	\end{pmatrix} \\ \Rightarrow \begin{pmatrix}
  	{-(E_1^*)^{-\gamma} }, -{(E_1^*+E_2^*)^{-\gamma} }, \dots, -{(E_1^*+E_2^*+ \dots + E^*_{q+1})^{-\gamma}  }
  	\end{pmatrix}.
  	\end{multline}
 
  	\end{enumerate}

\end{lem}

  Lemma \ref{theorem:jointEVT} allows us to solve the issue of unknown scaling rates by using a self-normalization trick  similar to the one employed by \cite{Chernozhukov2011} for quantile regression. By taking the ratio of two elements in the joint EVT \ref{theorem:jointEVT}, we eliminate scaling factors completely, while the form of the limit is explicitly known up to the EV index $\gamma$.

  Combining lemmas \ref{corollary:evtNoisyCenteredQuantile} and \ref{theorem:jointEVT}, we obtain the following version of the EVT that can be used to conduct inference on extreme quantiles under tail equivalence conditions: 
  \begin{theorem}[Feasible EVT]
  	\label{theorem:feasibleEVT} Let assumptions of theorem \ref{theorem:evtNoisy} hold, in particular, let $F$ have EV index $\gamma\in \R$. Let $q\geq 1, r\geq 0$ be  fixed integers and $l> 0$ be a fixed real number; let  $E_1^*, E_2^*, \dots$ be  iid standard exponential RVs. 
  		Then as $N, T\to\infty$ 
  				\begin{align} 
  			  		\dfrac{\vartheta_{N-r, N, T}-   F^{-1}\left(  1- \frac{l}{N}   \right) }{\vartheta_{N-q, N, T}-\vartheta_{N, N, T}}  
  	&	\Rightarrow  \dfrac{  (E_1^*+ \dots + E_{r+1}^*)^{-\gamma}-  l^{-\gamma}}{(E_1^*+ \dots + E_{q+1}^*)^{-\gamma}-(E_1^*)^{-\gamma} } \label{equation:feasibleEVT:-q},
  				\end{align}
  		 where for $\gamma=0$  the right hand side  means $[\log(E_1^* )-\log(l)]/[\log(E_1^*+ \dots + E_{q+1}^*)- \log(E_1^*) ]$.
  		 In addition, if $F$ satisfies assumption \ref{assumption:EV} with $\gamma<0$, then $F^{-1}(1)<\infty$ and  we may set $l=0$
  		  in eq. \eqref{equation:feasibleEVT:-q}.	 

  \end{theorem}

Theorem \ref{theorem:feasibleEVT} allows us to conduct inference on extreme quantiles with 
no knowledge of the value of $\gamma$.  
   The left hand side of eq.  \eqref{equation:feasibleEVT:-q}
 does not depend on $\gamma$; this statistic is the basis of our inference procedures. 
While the right hand side limit distribution of eq. \eqref{equation:feasibleEVT:-q}	 is non-pivotal and does depend on $\gamma$,  we show below how the critical values can be consistently estimated either by subsampling without estimating $\gamma$ (theorem \ref{theorem:subsampling}) or by plugging in a consistent estimator for $\gamma$ (remarks \ref{remark:gammaEstimation}-\ref{remark:critical-values-simulation}).

We show how to construct confidence intervals, estimators and hypothesis tests for extreme quantiles based on theorem \ref{theorem:feasibleEVT} with a series of examples.
\begin{example}[Location-scale equivariant confidence interval for 95th percentile]
	
	\label{example:CI:95evt}
	
	Let $l=10$ and $N=200$, in which case $F^{-1}(1-l/N)=F^{-1}(0.95)$. There are only 10 observations to the right of the sample quantile, and it is appropriate to use extreme order approximation described by theorem \ref{theorem:feasibleEVT}. 
	Let $q\geq 1$, $r \geq 0$ be fixed integers, see remark \ref{remark:choiceOfq} below on choice of $r$ and $q$.  Let   $\hat{c}_{\alpha}$ be a consistent estimator of the $\alpha$th quantile of $ [ (E_1^* + \dots + E_{r+1}^*)^{-\gamma} - 10^{-\gamma}  ]/ [(E_1^*+ \dots + E_{q+1}^*)^{-\gamma}- (E_1^*)^{-\gamma} ]$. Then  let the confidence interval  for $F^{-1}(0.95)$ based on sample size $N=200$ be given by:
	\begin{equation}
	 CI_{\alpha}=\left[  \vartheta_{N-r, N, T} - \hat{c}_{\alpha/2} \left(\vartheta_{N-q, N, T} - \vartheta_{N, N, T} \right),   \vartheta_{N-r, N, T} - \hat{c}_{1-\alpha/2} \left(\vartheta_{N-q, N, T} - \vartheta_{N, N, T} \right)     \right].
		\end{equation}
	 $CI_{\alpha}$ is location-scale equivariant, as the statistic of eq. \eqref{equation:feasibleEVT:-q} is location-scale invariant.
 
	Care must be exercised in interpreting the asymptotic properties of   $CI_{\alpha}$:   
	it is a $(1-\alpha)\times 100\%$  asymptotic confidence intervals for $F^{-1}\left(1-l/N \right)$. The target quantity shifts with $N$, and the value of $l$ determines which quantile is targeted for a given sample size. %
	
\end{example}

\begin{example}[Median-unbiased estimator for 95th percentile]
	\label{example:medianUnbiasedEstimator}

	Let $q, r$, and $\hat{c}_{\alpha}$ be as in example \eqref{example:CI:95evt}. 
	By theorem \ref{theorem:feasibleEVT} 
	$ P\left(  ({\vartheta_{N-r, N, T}- F^{-1}(1-l/N) })/({\vartheta_{N-q, N, T}-\vartheta_{N, N, T}}) \leq \hat{c}_{1/2} \right) \to 1/2$. Rearranging,   we obtain that
	the estimator
	\begin{equation}\label{equation:medianUnbiasedFullSample}
	\Mcal_{N, T} = \vartheta_{N-r, N, T} - \hat{c}_{1/2} ( \vartheta_{N-q, N, T} - \vartheta_{N, N, T} )
	\end{equation}
	is  asymptotically median-unbiased  for $F^{-1}(1-l/N)$ (see \cite{Chernozhukov2011} for a similar construction in a quantile regression setting).
Note that $\Mcal_{N, T}$ is always contained in  $CI_{\alpha}$ (unlike $\vartheta_{N-r, N, T}$ which  lies in $CI_{\alpha}$ if $\hat{c}_{\alpha/2}$ and $\hat{c}_{1-\alpha/2}$ have different signs).

\end{example}

 \begin{example}[Hypothesis tests about support]
 	\label{example:test-support-hypothesis}

 	We can also use theorem \ref{theorem:feasibleEVT} to test hypotheses about the support of $F$.     	Let $\gamma<0$
 	and suppose we wish to test  $H_0: F^{-1}(1)\leq  C$ 
 	vs. $H_1: F^{-1}(1)>C$. 
 	Define the test statistic  $
 	W_C=  {(\vartheta_{N, N, T}- C)}/{(\vartheta_{N-q, N, T} -\vartheta_{N, N, T})}.$
 	The test rejects $H_0$ if $W_C<\hat{c}_{\alpha}$ where $\hat{c}_{\alpha}$ is a consistent estimator of the $\alpha$th quantile of  ${  (E_1^*)^{-\gamma}  }/[{(E_1^*+ \dots + E_{q+1}^*)^{-\gamma}-(E_1^*)^{-\gamma} }]$. 
 	The test is   asymptotically size $\alpha$ and consistent against point alternatives, since $P(W_C<\hat{c}_{\alpha}|F^{-1}(1)=C)\to \alpha$, and for any $\delta>0$, 	$P(W_C<\hat{c}_{\alpha}|F^{-1}(1)=C-\delta)  \to 0$,  	$P(W_C<\hat{c}_{\alpha}|F^{-1}(1)=C+\delta) \to 1$.
 \end{example}

  We now describe a subsampling estimator for the limit distribution of theorem \ref{theorem:feasibleEVT} \citep{Politis1994, Politis1999}.
    Let $q>1, r\geq 0$, and $l\geq 0$. Define $J(x)$ to be the limit distribution in eq. \eqref{equation:feasibleEVT:-q}.   
  Split the set of units $\curl{1, \dots, N}$ into all    subsamples of size $b$  and index the  subsamples by $s$, $s=1, \dots, \binom{N}{b}$ (see remark \ref{remark:choice-of-b} below on the choice of $b$).  Let $\vartheta^{(s)}_{b-k, b, T}$ be the $(b-k)$th  order statistic in subsample $s$.  
  Define the subsampling estimator $L_{b, N, T}$ for $J$ as 
  \begin{align} 
   L_{b, N, T}(x) & = \dfrac{1}{\binom{N}{b}}\sum_{s=1}^{\binom{N}{b}}  \I\curl*{ W_{s, b, N, T}\leq x   }, \quad 
  W_{s, b, N, T}   =  \dfrac{ \vartheta_{b-r, b, T}^{(s)} - \vartheta_{N-Nl/b, N, T} }{\vartheta_{b-q, b, T}^{(s)}  - \vartheta_{b, b, T}^{(s)} }, \quad  \dfrac{Nl}{b}\leq N.
 \end{align}
Observe that the subsample statistic $W_{s, b, N, T}$ is centered at $\vartheta_{N-Nl/b, N, T}$. Intuitively,  this corresponds to the $(1-l/b)$th quantile,  correctly centering  the subsampled statistics.  If we are interested in $F^{-1}(1)$, then $l=0$, and the statistic is centered at $\vartheta_{N, N, T}$.

  Define the estimated critical value $\hat{c}_{\alpha}$ as the $\alpha$th quantile of $L_{b, N, T}$. The following result shows that $\hat{c}_{\alpha}$ is consistent for the true critical values  of interest for all $\alpha\in(0, 1)$.

  \begin{theorem}
  
 	\label{theorem:subsampling}  
 	Let $b=N^m, m \in (0, 1)$.  	  If $l>0$, let conditions of propositions \ref{proposition:ratesExtremeHalfUniformity} and \ref{proposition:ratesIntermediateHalfUniformity} hold with $\delta=1-m$.  If $\gamma<0$ and $l=0$, then let conditions of proposition \ref{proposition:ratesExtremeHalfUniformity} hold. 
 	 Then the subsampling estimator $L_{b, N, T}(x)\xrightarrow{p} J(x)$ at all $x$ and  $\hat{c}_{\alpha}\xrightarrow{p} c_{\alpha}= J^{-1}(\alpha)$    for all $\alpha\in(0, 1)$.
 	 
  \end{theorem}

Theorem \ref{theorem:subsampling} shows that subsampling may be applied in the case of noisy observations.  The key step in establishing the validity of subsampling is to control the estimation noise in the subsamples and to leverage tail equivalence. 
Theorem \ref{theorem:subsampling} parallels a result derived by \cite{Chernozhukov2011} for inference in extreme quantile regression.
%

\begin{remark}[Choice of $b$]
	\label{remark:choice-of-b}
We suggest two possible approaches for choosing $b$: a  minimal  interval volatility criterion \citep[algorithm 5.1]{Romano2001SubsamplingIntervalsAutoregressive} and a criterion based on the stability of the subsampled distribution \citep[p.\ 971]{Bickel2008ChoiceOutBootstrapa}.  	In both cases, subsampling is applied for a range of candidate values of $b$. The value of $b$ is selected by minimizing an approach-specific variability criterion. The former approach minimizes the variance of the endpoints of the confidence intervals. The latter one minimizes the distance between the subsampling distributions for pairs of consecutive candidate values of $b$.
	Provided the conditions of theorem \ref{theorem:subsampling} hold for each candidate value of $b$, either approach will  select a valid $b$. In the simulations of section \ref{section:simulations},  
	  choosing $b$ with the minimal volatility method leads to favorable performance of confidence intervals.
\end{remark}

\begin{remark}[Estimation of the EV index]
	\label{remark:gammaEstimation}
	Inference based on theorem \ref{theorem:feasibleEVT} does not require an estimate   of $\gamma$. However, $\gamma$ may be consistently estimated.
	Let $k=k(N)$ satisfy $k\to\infty, k=o(N)$. Let 
	$A_{N, T} =k^{-1}\sum_{i=0}^{k-1} \vartheta_{N-i, N, T}- \vartheta_{N-k, N, T}$ and $B_{N, T} = k^{-2}\sum_{i=0}^{k-1} i \left( \vartheta_{N-i, N, T}- \vartheta_{N-k, N, T}\right)$.
	The \cite{Hill1975} estimator $\hat{\gamma}_H$ and probability weighted moment (PWM) estimator $\hat{\gamma}_{PWM}$ of \cite{Hosking1987} are  defined as 
	\begin{align}
	\hat{\gamma}_H &  = \dfrac{1}{k} \sum_{i=0}^{k-1} \log (\vartheta_{N-i, N, T} ) - \log(\vartheta_{N-k, N, T}), \quad 	\hat{\gamma}_{PWM}  = \dfrac{A_{N, T}- 4 B_{N, T}}{A_{N, T}- 2B_{N, T}},
	\end{align}
	 If conditions of theorems \ref{theorem:evtNoisy} and \ref{theorem:intermediateNormality} hold, then $\hat{\gamma}_H\xrightarrow{p} \gamma$ if $\gamma> 0$; 
	 $\hat{\gamma}_{PWM}\xrightarrow{p}\gamma$  if  $\gamma<1$ (approximately if $F$ has a finite first moment). The proof of consistency  is given in \textcolor{black}{proof appendix}. 
	 In practice, the value of $k$ may be chosen in a data-driven way, and \cite{Caeiro2016ThresholdSelectionExtreme} discuss a number of approaches. In the simulations of section \ref{section:simulations}, we use the modified semiparametric bootstrap of \cite{Caers1999StatisticsModelingHeavy} \citep[algorithm 4.3]{Caeiro2016ThresholdSelectionExtreme}.	 
\end{remark}

\begin{remark}[Simulated critical values]
\label{remark:critical-values-simulation}
	
Estimating $\gamma$ provides a second way of estimating the critical values necessary for application of theorem \ref{theorem:feasibleEVT}. Simulation-based critical values  
	can be obtained by drawing samples  from the limit distribution of eq.  \eqref{equation:feasibleEVT:-q}   after plugging in $\hat{\gamma}_H$ or $\hat{\gamma}_{PWM}$ in place of $\gamma$. 
	%
\end{remark}

\begin{remark}[Choice of $r$ and $q$]
	\label{remark:choiceOfq}
	
	 Applying the methods of examples \ref{example:CI:95evt}-\ref{example:test-support-hypothesis}
	requires choosing the tuning parameters $r$ and $q$.
	We suggest the following choices.
		For $N\leq 5000$, we suggest taking $q=2$ with subsampled critical values, and $q\in [2, 4]$ with simulated critical values.
		For larger cross-sections, larger values of $q$ may be taken, up to $q=30$.   
	In both cases, the numerator parameter $r$ should be picked to match $\vartheta_{N-r, N, T}$ to the sample $\tau$th quantile, that is,  $r=\floor{(1-\tau)N}$.
	While this choice of $r$ is somewhat improper in the context of theorem \ref{theorem:feasibleEVT}, we note that $r\leq 100$ regardless of $N$   under the rule of thumb outlined at the beginning of section \ref{section:inference}. 
The simulations of section \ref{section:simulations} and the Online Appendix show that these choices yield favorable coverage and length properties.  

\end{remark}

 \begin{remark}
 	An alternative approach for constructing CIs for extreme quantiles using an extreme order approximation is proposed by \cite{Muller2017}.  
They use lemma \ref{theorem:jointEVT} as the foundation of inference by treating    the vector of top $q$   order statistics as a single draw from the corresponding limit distribution.
Based on such a joint EVT, \cite{Muller2017} propose two methods: inverting the likelihood ratio 
and minimizing the average expected length where the average is taken over a pre-specified range of values for $\gamma$. 
Both methods lead to valid CIs for extreme quantiles in our setting,  as long as lemma \ref{theorem:jointEVT} holds; though we suggest using CIs based on theorem \ref{theorem:feasibleEVT} that require no optimization and no bounds on the parameter $\gamma$.

\end{remark} 

\subsection{Inference Using Intermediate Order Approximations}

The intermediate  order approximation of theorem \ref{theorem:intermediateNormality} provides an alternative approach to inference that is based on convergence of  intermediate order statistics (eq. \eqref{equation:intermediateTheoremStatistic}).  In this case, the quantile of interest is modeled as drifting to 1 at a rate   $k/N$ where $k\to\infty, k=o(N)$ as $N\to\infty$; this rate is slower than the rate $N^{-1}$ of extreme order approximations. 
 The resulting statistic is asymptotically standard normal.
 

To eliminate the unknown scaling rate $c_N$, we again use an additional order statistic.\footnote{ There are alternative approaches to inference using intermediate order statistics. For example, the subsampling approach
  	of \cite{Bertail1999, Bertail2004} can be used to account for unknown scaling rate $c_N$. However, the presence of a slowly varying component in $c_N$  requires using multiple subsampling sizes, which may be problematic if $N$ is not extremely large.
  	Alternatively, see ch. 4  of \cite{DeHaan2006} for inference under a second-order condition in a setting without estimation noise.
  }  
Unlike theorem \ref{theorem:feasibleEVT}, the statistics in the denominator of the self-normalized statistic are asymptotically perfectly dependent and only differ by a non-zero deterministic factor.  
  The following theorem first establishes that such a technique works in   the noiseless case,  which may be of independent interest; the result is then transferred  to noisy observables.
  
  \begin{theorem}\label{theorem:feasibleIVTnoiseless} 
  Let assumptions \ref{assumption:independence} and \ref{assumption:vonMises} hold. Let  $k=o(N), k\to\infty$. Let  $f=F'$ be non-increasing or non-decreasing in some left neighborhood of $F^{-1}(1)$ $(F^{-1}(1)\leq \infty)$. 
  	\begin{enumerate}[noitemsep,topsep=0pt,parsep=0pt,partopsep=0pt,  label={(\arabic*)}]
  	\item   Then 
  	\begin{equation} %
  	\dfrac{\theta_{N-k, N}- F^{-1}\left(1- \frac{k}{N} \right)}{\theta_{N-k, N}-\theta_{N-k-\floor{\sqrt{k}}, N}} \Rightarrow N(0, 1), \quad N\to\infty.
  	\end{equation} 
  	\item In addition, let conditions of theorem \ref{theorem:intermediateNormality} hold when evaluated at $k$ and $k+\sqrt{k}$. Then 
  	\begin{equation}\label{equation:IVTnoisy}
  	\dfrac{\vartheta_{N-k, N, T}- F^{-1}(1-k/N)}{\vartheta_{N-k, N, T}-\vartheta_{N-k-\floor{\sqrt{k}}, N, T}} \Rightarrow N(0, 1), \quad N, T\to\infty.
  	\end{equation}
  	\end{enumerate}
\end{theorem}
The statistics of theorem \ref{theorem:feasibleIVTnoiseless} are particularly simple to use. First,  the limiting distribution does not involve any unknown parameters. Second, these statistics do not involve any tuning parameters. The choice of $k$   determines  the centering quantile of interest. When $k$ is chosen, the denominator is uniquely determined by $k$:

\begin{example}
	\label{example:ivt-interval}
	 Let $N=200$ and let    $k=k(N)$  be such that  $k(200)=16$ and $k\to\infty, k=o(N)$ as $N\to\infty$. Then 
	  by theorem \ref{theorem:feasibleIVTnoiseless} a  confidence interval for $F^{-1}(0.92)$ based on sample of size $N=200$:
\begin{equation*} 
	\left[\vartheta_{N-k, N}- z_{1-\alpha/2}\left(\vartheta_{N-k, N} - \vartheta_{N-k-\floor{\sqrt{k}}, N} \right), \vartheta_{N-k, N}- z_{\alpha/2}\left(\vartheta_{N-k, N} - \vartheta_{N-k-\floor{\sqrt{k}}, N} \right)   \right], 
\end{equation*}
  where   $z_{\alpha}$ the $\alpha$th quantile of the standard normal distribution.   
\end{example}

 \begin{remark}[Limitations of theorem \ref{theorem:feasibleIVTnoiseless}]
 	\label{remark:feasibleIVT:limitation}
 	
 	The approximation of theorem \ref{theorem:feasibleIVTnoiseless} should not  be used if $k$ or $N$ are small. If $\floor{\sqrt{k}}$ is small but positive,  $\vartheta_{N-k, N}$ and $\vartheta_{N-k-\floor{\sqrt{k}}, N}$ will be close, and the distribution of the statistic may be far from normality. If $\floor{\sqrt{k}}=0$, the statistic cannot be used. 	Since $k$ must satisfy $k=o(N)$, $N$ must also be suitably large.
 	A value of $k\approx 100$ is generally the minimum threshold for theorem \ref{theorem:feasibleIVTnoiseless}  to provide a useful approximation, coupled with the requirement that $N  \gtrsim 10^4$. We refer to the simulations reported in the Online Appendix.

 \end{remark}
 
 \begin{remark}[Practical difference between extreme and intermediate order CIs]
 	\label{remark:difference-extreme-intermediate}
 	
 	Although the formulas for the two CIs are visually similar, they differ in terms of their construction and applicability. First,  for the extreme order CI of example \ref{example:CI:95evt}, there is flexibility in the choice of the component order statistics $(\vartheta_{N-r, N, T}, \vartheta_{N-q, N, T})$, regardless of the target quantile. In contrast, for the intermediate order CI of example \ref{example:ivt-interval} the   statistics $(\vartheta_{N-k, N}, \vartheta_{N-k-\floor{\sqrt{k}}, N}) $ are rigidly determined by the target quantile.
 	Second, the CI of example \ref{example:CI:95evt} can be used even for small values of $N$, where the CI of example \ref{example:ivt-interval} should only be applied in sufficiently large samples (remark \ref{remark:feasibleIVT:limitation}).
 	 \end{remark}

\subsection{Inference Using Central Order Approximations}
\label{subsection:centralOrderApproximations}
 
 The third method of inference is based on the central limit theorem for quantiles. 
The quantile of interest  $F^{-1}(\tau)$ is modeled as fixed and  independent from $(N, T)$, in contrast to the  extreme and intermediate order approximations given above. Such ``central'' order approximations require that a sufficient number of observations   be available on both sides of the corresponding  sample order statistic $\vartheta_{\floor{N\tau}, N, T}$ (at least 100 in the simulations of section \ref{section:simulations}).

 \cite{Jochmans2019} derive such approximations in the context of our problem, and we  briefly state their results. They study a version of \eqref{equation:noisyThetaDefinition} given by $\vartheta_{i, T} = \theta_i + T^{-1/2} \varepsilon_i$ (that is, $p=1/2$ and $G_T=G$ for all $T$).
We introduce some additional notation: let $K$ be a kernel function, $h$  a bandwidth parameter, and define
 \begin{align}%
 b_F(x)   & = \left[\dfrac{ \E\left[\sigma^2_i|\theta_i=x \right]f(t) }{2}\right]', \quad  \sigma_i^2   = \var(\varepsilon_{i}|\theta_i),\\  
  \hat{b}_F  &  = - \dfrac{(nh^2)^{-1}\sum_{i=1}^n \sigma^2_i K'((\vartheta_{i, T}-\theta)/h ) }{2}, \quad 
  \hat{\tau}^*  = \tau + \dfrac{\hat{b}_F(\vartheta_{\floor{N\tau}, N, T}) }{T}.
 \end{align} 
$\sigma_i^2$ is assumed known and invariant over time. 
 
 \begin{theorem}[Propositions 2 and 4 of \cite{Jochmans2019}]
\label{theorem:JW}
 	
Let conditions of proposition 3 in \cite{Jochmans2019} hold, and in particular  let for all $T$ $\varepsilon_{i, T}=\varepsilon_i$, $\E[\varepsilon_i|\theta_i]=0$, and $\varepsilon_i$ be independent from $\theta_i$ given $\sigma_i^2$.    Let $\tau\in (0, 1)$.
\begin{enumerate}[noitemsep,topsep=0pt,parsep=0pt,partopsep=0pt,  label={(\arabic*)}]
\item If $N/T^2\to c <\infty$, then as   $N, T\to\infty$
\begin{equation} 
\sqrt{N}\left(\vartheta_{\floor{N\tau}, N, T} - F^{-1}(\tau) + \dfrac{1}{T}\dfrac{b_F(F^{-1}(\tau))}{f(F^{-1}(\tau))} \right)\Rightarrow N\left(0, \dfrac{\tau(1-\tau)}{f(F^{-1}(\tau))^2} \right).
\end{equation}
\item If $N/T^4\to 0$, then as $N, T\to\infty$
	\begin{equation} 
\sqrt{N}\left(\vartheta_{\floor{N\hat{\tau}^*}, N, T}-F^{-1}(\tau)  \right)\Rightarrow N\left(0, \dfrac{\tau(1-\tau)}{f(F^{-1}(\tau))^2} \right).
\end{equation}
\end{enumerate} 

 \end{theorem}
Theorem \ref{theorem:JW} is a restatement of propositions 2 and 4 of \cite{Jochmans2019}.
It shows that the sample $\tau$th quantile $\vartheta_{\floor{N\tau}, N, T}$  is a consistent and asymptotically normal estimator for $F^{-1}(\tau)$ with standard asymptotic variance. However, $\vartheta_{\floor{N\tau}, N}$ is subject to bias of leading order $1/T$.  
\cite{Jochmans2019} show that this bias can be reduced by instead considering the  sample $\hat{\tau}^*$th quantile: $\vartheta_{\floor{N\hat{\tau}^*}, N, T}$ is  consistent and asymptotically normal with the same variance, but   the leading order of the bias is instead given by $1/T^2$. This bias  is eliminated if $\sqrt{N}/T^2\to 0$.

 	\begin{remark}  
 		\label{remark:debiasing}
 For central order approximations,	 the order of the bias incurred by using   $\vartheta_{\floor{N\tau}, N, T}$ in place of $\theta_{\floor{N\tau}, N}$  
 is the same for a broad class of distributions, and equal to $T^{-1}$. This invariance of bias order enables construction of the debiased estimator $\vartheta_{\floor{N\hat{\tau}^*}, N, T}$. 
 		 The situation is more complex for extreme and intermediate order approximations. 
 		The magnitude of the impact of estimation noise   is   determined by the interaction of $a_N$ and $T$ in theorem \ref{theorem:evtNoisy}; 
 	  $a_N$ itself may behave like $N^{\gamma}$ for   $\gamma\in \R$ depending on $F$,   up to slowly varying components. 
 		 
 	\end{remark}
 
\section{Simulation Study}\label{section:simulations}

 We assess the performance of our confidence intervals with a simulation study.
We consider a linear   model with unit-specific coefficients where the outcome $y_{it}$ is generated by covariates $(x_{it}, z_{it})$ as 
\begin{equation} 
y_{it} = \alpha_i + \eta_i x_{it} + \theta_i z_{it} + \sqrt{\dfrac{\var(\theta_i)}{\var(u_{it})}}\times u_{it}.
\end{equation} 
The parameter of interest $\theta_i$ is the coefficient on $z_{it}$.   We are interested in the  coverage  and length properties of a  nominal 95\% confidence interval (CI) for  the 0.9-0.9995th quantiles of $\theta_i$.

The data is generated as follows. The coefficients $(\alpha_i, \eta_i, \theta_i)$ are drawn from  a   Gaussian copula with correlation 0.3 and  marginals $t_3$, where $t_{\nu}$ is   Student's $t$-distribution with $\nu$ degrees of freedom. The value of $\nu$ matches the data of our empirical application. 
 Covariates $x_{it}$ are drawn as $0.3\eta_i + \left(1+0.3\norm{(\alpha_i, \eta_i, \theta_i)} \right)^{1/2}(0.1+U_i)$ where $U_i$ is a Uniform[0, 1]; $z_{it} $ are generated similarly with $\theta_i$ in place of $\eta_i$. 
 In this stylized setup,  the UIH and covariates are   dependent.
     $u_{it}$ is sampled independently  from $G_{\beta}$, where $G_{\beta}$ has density $g_{\beta} = \beta(1+\abs{x})^{-\beta-1}/2$ for $x\in \R$ and $\beta=8$. $G_{\beta}$ is a two-sided power law with finite moments of order $<\beta$.  $u_{it}$ is rescaled so that  its variance matches that of $\theta_i$.
   Coefficients are estimated using OLS.     As a result, $\theta_i$ and  estimation noise $\varepsilon_{i, T}$ are dependent.
    We consider $N=200, 2000$ and $T=10, 50$ and draw 10000 MC samples.
    
 \vspace{1mm}

  We construct CIs using  the three approximations of section \ref{section:inference}:
   \begin{enumerate}[noitemsep,topsep=0pt,parsep=0pt,partopsep=0pt, label={(\arabic*)}, leftmargin=*]
  	\item \emph{Extreme}:   we report two CIs based on theorem \ref{theorem:feasibleEVT} --- with subsampled (theorem \ref{theorem:subsampling}) and simulated (remark \ref{remark:critical-values-simulation}) critical values. 
  	For subsampling, we draw 5000 subsamples; subsample size $b$ is chosen using the minimum volatility method \citep{Politis1999} (remark \ref{remark:choice-of-b}). 
  	For the critical values of remark \ref{remark:critical-values-simulation}, we estimate $\gamma$ with the PWM estimator $\hat{\gamma}_{PWM}$ (remark \ref{remark:gammaEstimation}). The tuning parameter $k$ for $\hat{\gamma}_{PWM}$ is chosen using  algorithm 4.3 in \cite{Caeiro2016ThresholdSelectionExtreme} --- a modified version on the semiparametric bootstrap of \cite{Caers1999StatisticsModelingHeavy}.
For construction of the statistic itself, we take $r=\floor{l}$, so that the CI is centered on the sample quantile. The denominator tuning parameter $q$ is selected in line with remark \ref{remark:choiceOfq}, following additional simulation results in the Online Appendix. For the subsampled CI we take $q=2$, and for the simulated CI $q=4$.

  	\item \emph{Intermediate}: we report the CI based on theorem \ref{theorem:feasibleIVTnoiseless}. The value of $k$ is mechanically determined by the target quantile as in example \ref{example:ivt-interval} (see also remark \ref{remark:feasibleIVT:limitation}). \label{page:ref2:major-comment-2-choice-k-sim-description} 
  	
  	We also include a ``textbook'' CI based on extrapolation, see theorem 4.3.1 in \cite{DeHaan2006}. Unlike the CI of example \ref{example:ivt-interval}, the ``textbook'' CI requires choosing the intermediate sequence $k$ as a tuning parameter.  \label{page:ref2:major-comment-2-contrast-intermediate-CIs}
  	We choose $k$ using the method of \cite{Caeiro2016ThresholdSelectionExtreme} and use the PWM estimator for $\gamma$. This CI can only be constructed for sufficiently high quantiles.     The validity of the extrapolation interval hinges on a second-order condition \citep{DeHaan1996} which we do not examine in the current paper.

  	\item \emph{Central}: we reports two CIs: The first interval is a binomial CI based on the raw data.  {The same approach is implemented in the Stata command \texttt{centile}.}  
  	The second interval uses the analytical correction of \cite{Jochmans2019}. The corresponding critical values are computed using the bootstrap with 1000 bootstrap samples. 
  \end{enumerate}

\begin{figure}[htp]
	\centering 
	\includegraphics[width=\linewidth]{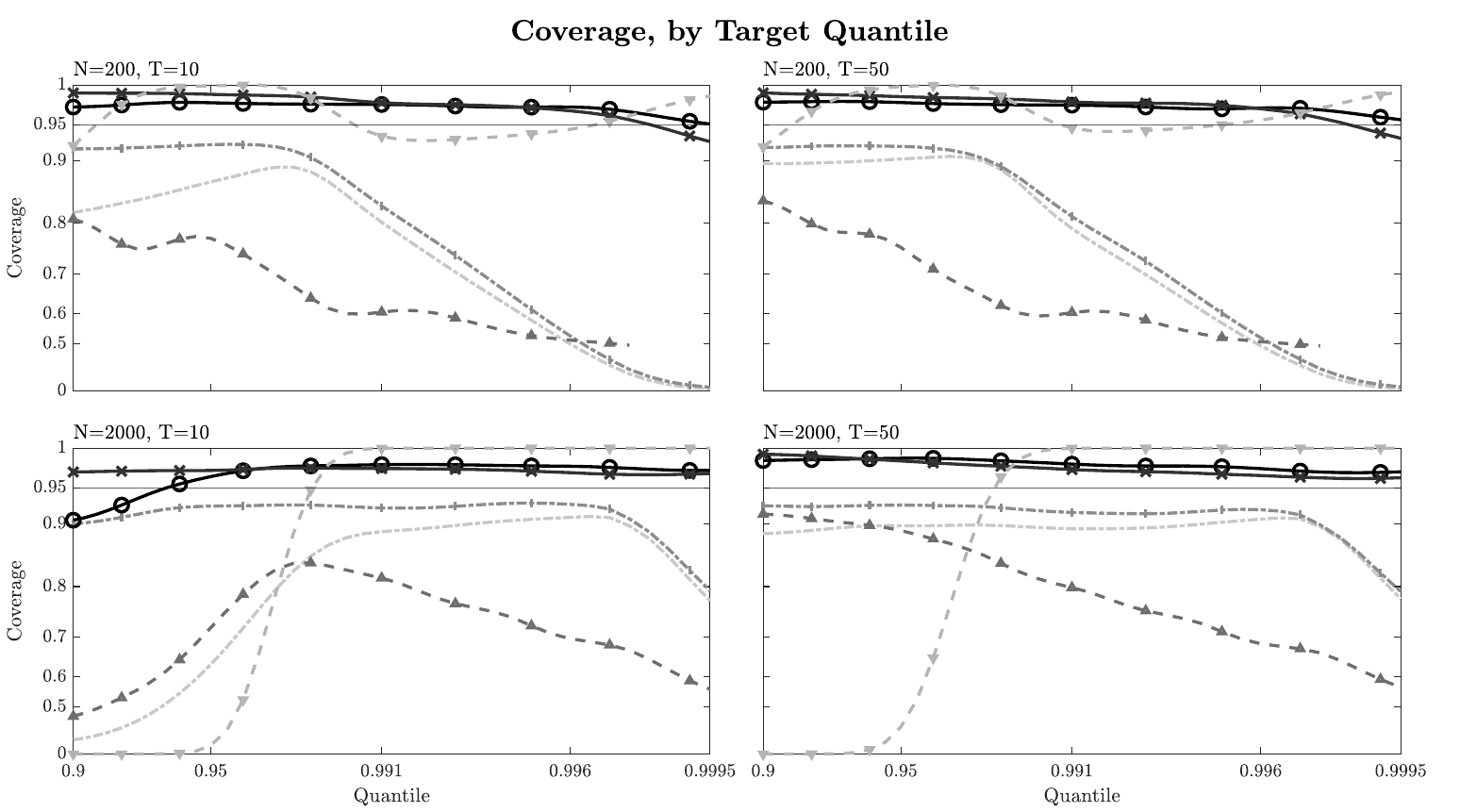}
 	\includegraphics[width=\linewidth]{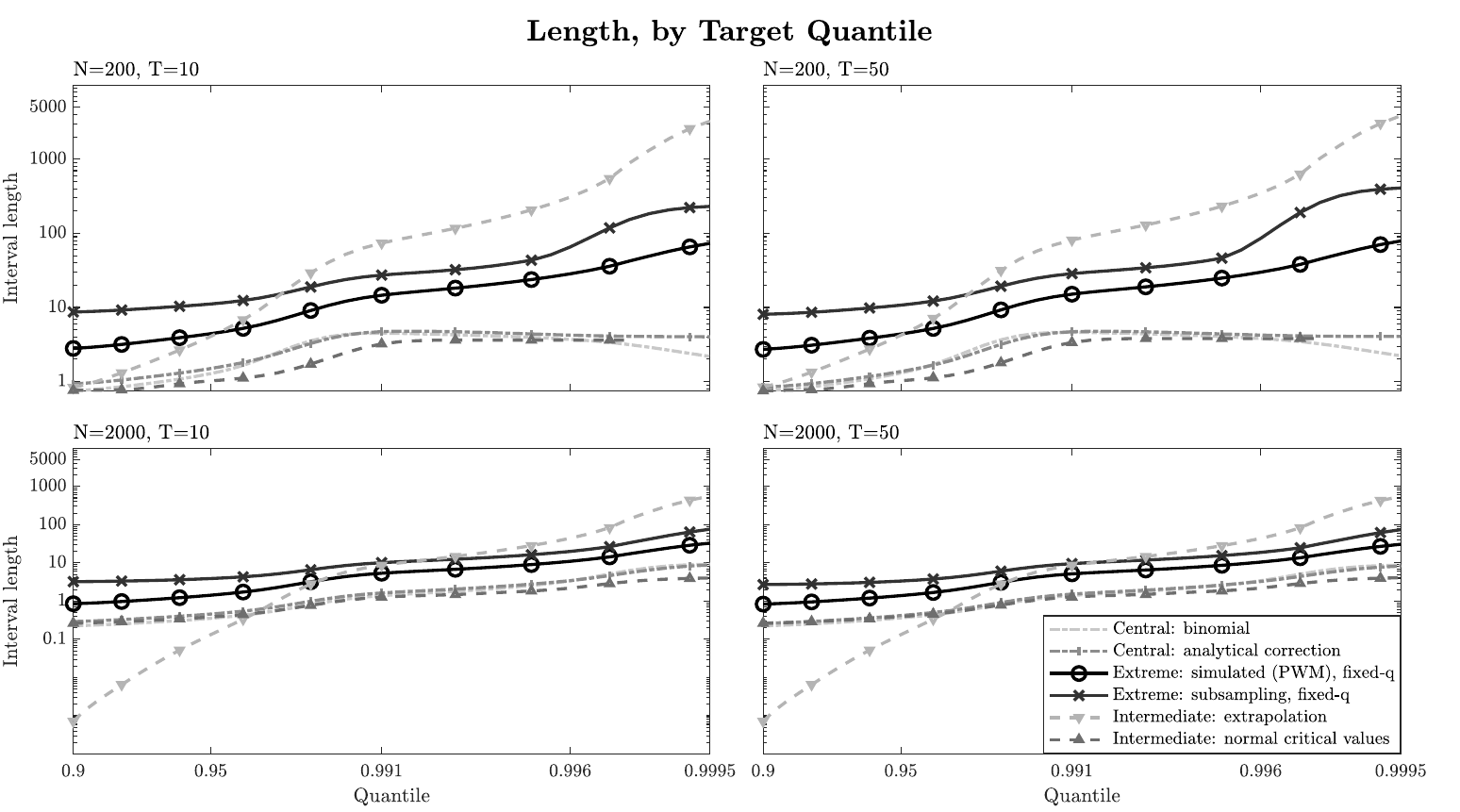}
	
	\caption{Coverage and length for 95\% nominal confidence interval.     $F=t_3$.   $u_{it}\sim G_\beta, \beta=8$ (8 finite moments).  Notes: (1) nonlinear $x$ and $y$-axes; (2) intermediate CIs cannot be constructed for some quantiles (remark \ref{remark:feasibleIVT:limitation}).}  \label{figure:simulation-noisy}
\end{figure}

\vspace{1mm}

We briefly discuss the validity of the above approximations.
For the extreme approximations, the rate conditions hold in light of proposition \ref{proposition:ratesExtremeHalfUniformity}  --- the estimation noise has more moments than $\theta_i$.
For intermediate approximations, the rate conditions hold only for a range of higher quantiles  (remark \ref{remark:dependenceonDelta}).
  For central order approximations, the rate conditions for validity of using raw data do not hold, particularly for $(N, T)=(2000, 10)$.  

 \vspace{1mm}

 Fig.\ \ref{figure:simulation-noisy} depicts our core simulation results.  It depicts coverage rates and lengths for the above confidence intervals. In order to assess the impact of estimation noise, in fig. \ref{figure:simulation-noiseless} we also plot the same results in a noiseless setting, that is, with $u_{it}=0$. 

 \vspace{1mm}
 
Our key recommendation for inference on $\tau$th quantiles reflects that given in section \ref{section:inference}:
\begin{enumerate}[noitemsep,topsep=0pt,parsep=0pt,partopsep=0pt, label={(\arabic*)}, leftmargin=*]
	\item If $(1-\tau)N\leq 100$, we recommend extreme order CIs. Both extreme CIs offer favorable coverage and length properties, and are overall  robust to estimation noise. Between the two extreme CIs, the choice may be viewed as a trade-off between performance and ease of computation. The subsampled CI is somewhat more robust, but more challenging to compute due to subsampling. 
	\item If $(1-\tau)N>100$, a central order approximation should be used in conjunction with the   correction of \cite{Jochmans2019}.  The correction generally yields coverage close to the nominal level without significantly increasing the CI length. 
	
\end{enumerate}

 \begin{figure}[htp]
 	\centering 
 	\includegraphics[width=\linewidth]{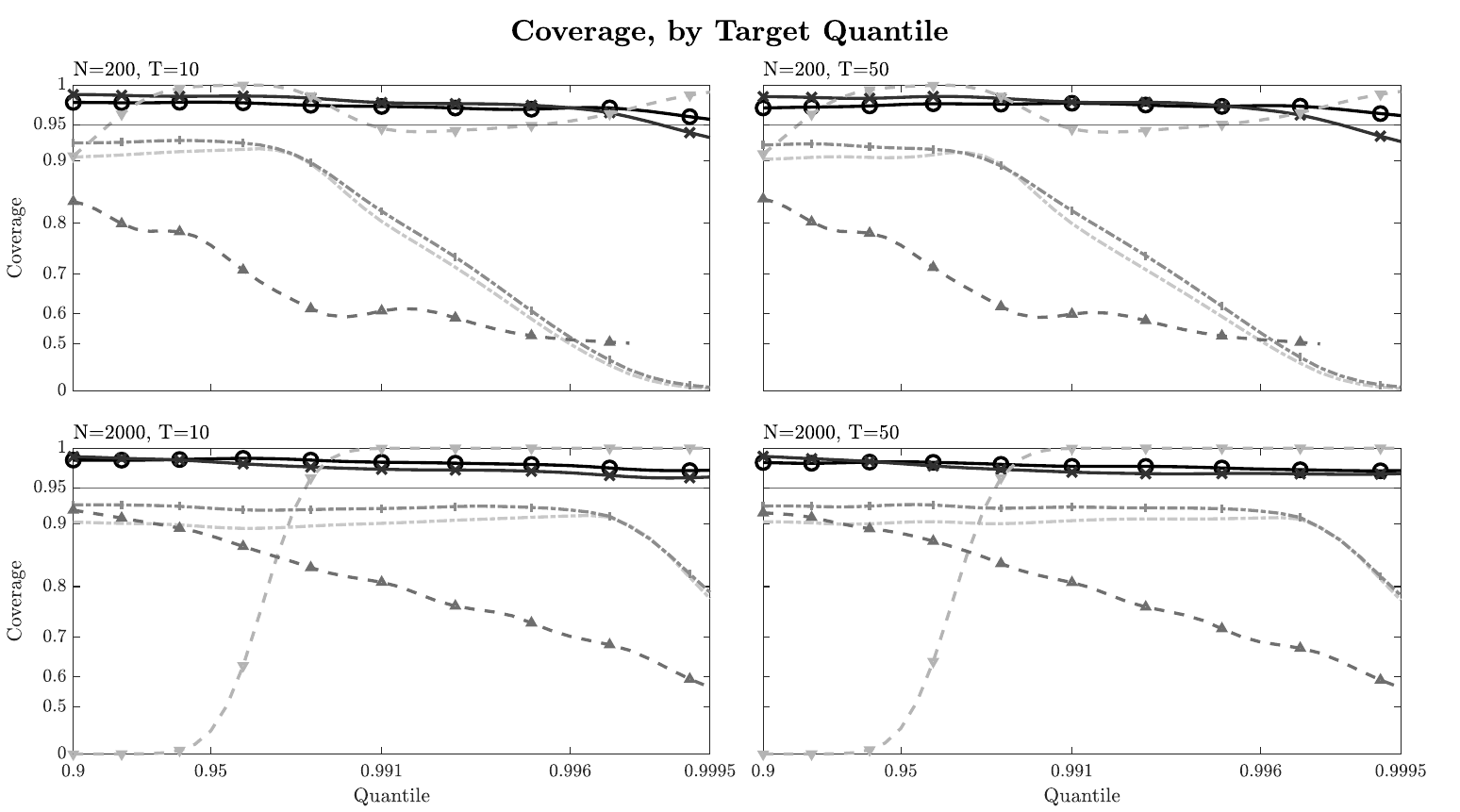}
 	\includegraphics[width=\linewidth]{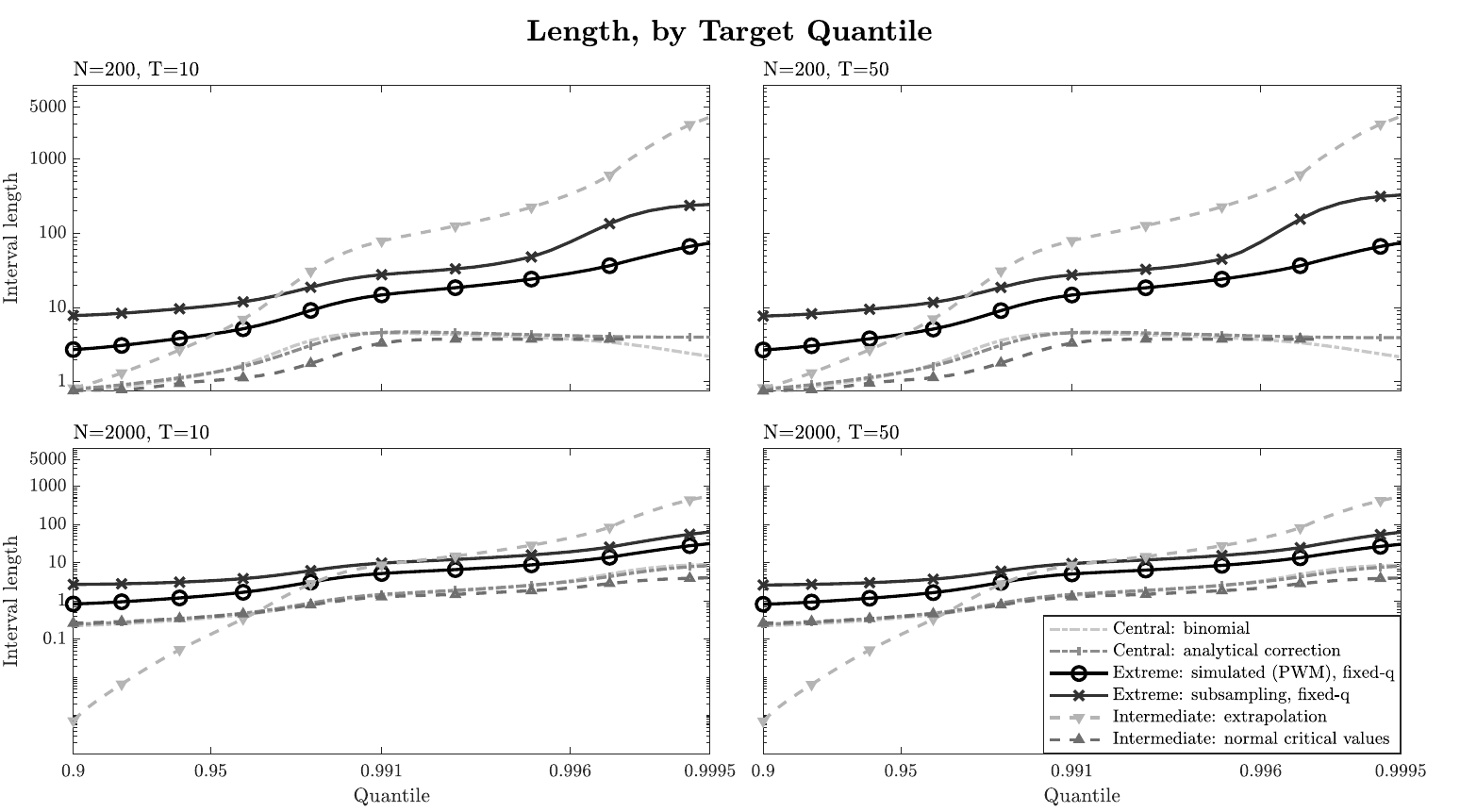}
 	
 	\caption{Coverage and length for 95\% nominal confidence interval.     $F=t_3$.   Noiseless data.   Notes: (1) nonlinear $x$ and $y$-axes; (2) intermediate CIs cannot be constructed for some quantiles (remark \ref{remark:feasibleIVT:limitation}).}  \label{figure:simulation-noiseless}
 \end{figure}
 
 The above recommendation is interchangeable only to a limited degree. As   $(1-\tau)N$ approaches zero, central CIs should be avoided as their coverage and length collapses to 0.
 The situation for extreme CIs is more delicate.  
 As $(1-\tau)N$ increases beyond 100, the distributional properties of $\vartheta_{N-\tau N, N, T}$  are  better reflected by theorem \ref{theorem:intermediateNormality} rather than theorem \ref{theorem:evtNoisy}. The associated rate conditions $(N, T)$ are accordingly typically stricter (remark \ref{remark:comparisonConditionsEVTIVT}).    If tail equivalence  holds at $\tau$, extreme CIs are valid, if wide. However, the rate conditions are progressively harder to satisfy as $\tau$ falls. Their failure may lead to size distortions  (compare the panels for $N=2000$ on figs. \ref{figure:simulation-noisy} and \ref{figure:simulation-noiseless}).

 The performance of the other three CIs is at best mixed.  First,  the binomial interval is strongly affected by estimation noise. The impact of noise is evident in the undercoverage of the binomial CI for quantiles below 0.99 (compare $(N, T)=(2000, 10)$ in figs. \ref{figure:simulation-noisy} and \ref{figure:simulation-noiseless}).  
 Second,  the  ``textbook'' intermediate extrapolation CI has good coverage properties when $(1-\tau)N\leq 20$ (e.g. $\tau\geq 0.99$ for $N=2000$).  However, this performance comes at the price of intervals that are notably longer than the extreme CIs (bottom panels of figs. \ref{figure:simulation-noisy}-\ref{figure:simulation-noiseless}).
Second, the CI based on theorem \ref{theorem:feasibleIVTnoiseless} generally has poor coverage. The issue is more pronounced for higher quantiles. 
This failure is primarily due to the slow convergence to the $N(0, 1)$ limit in statistic \eqref{equation:IVTnoisy}, as we show in the Online Appendix.

Additional simulations are reported in the Online Appendix.    We report results for $N=10000$ and $T=20$, and for additional distributions for $\theta_i$ and $u_{it}$. We also examine the impact of different choices for the tuning parameters of the extreme order CIs. Furthermore, we explore performance of the corrected quantile estimators of  \eqref{equation:medianUnbiasedFullSample} and assess speed of of convergence in theorem \ref{theorem:feasibleIVTnoiseless}. Overall, the evidence emerging from these
simulations is in line with the results presented above.

\section{Empirical Application}
\label{section:empirical}

As an empirical illustration, we revisit the relationship between firm productivity and population density, following \cite{Combes2012b} (CDGPR12). Since firm productivity must be estimated from firm-level data, this setting naturally aligns with our framework.

\paragraph{Background}

	CDGPR12 examine why firms in denser areas tend to be more productive \citep{Melo2009}, focusing on two possible explanations: agglomeration economies and firm selection. Their approach involves a two-step procedure. First, they estimate firm-specific productivities. Second, they compare the distributions of productivity in high-density (above-median density, AMD) and low-density (below-median density, BMD) areas using these estimates.
	
	A key assumption in the second step is that the true productivity distributions in AMD and BMD stem from a common latent parent distribution but differ in three parameters: mean, variance, and the extent of left-tail truncation. CDGPR12 estimate differences in these parameters to quantify the effects of agglomeration and selection. The mean and variance capture agglomeration effects --- firms in AMD tend to be more productive on average, though some benefit more than others. The truncation parameter reflects firm selection: as \cite{Asplund2006} argue, competition is tougher in larger markets, potentially leading to stronger left-tail truncation --- firms in denser areas must meet a higher productivity threshold to survive.

\paragraph{Empirical questions}
 
	We examine two questions. First, do the three parameters of CDGPR12 fully capture differences in the tails of AMD and BMD productivity distributions? An affirmative answer would support the key assumption of CDGPR12. Second, is there evidence for firm selection? This information is provided by the left tails of the productivity distributions. While CDGPR12 find that truncation must have equal strength between AMD and BMD, they do not determine the minimal productivity level or whether truncation occurs at all. 

\paragraph{Data}
We use firm-level microdata from the Banco de España's CBI dataset \citep[remote access]{BancoDeEspana2024MicrodataIndividualEnterprises} and demographic data from the Spanish National Statistics Institute (INE). The CBI covers over 50\% of non-financial Spanish firms from 1995 to 2023, including all public firms.

Our analysis focuses on three service-oriented sectors: wholesale and retail trade (NACE G), professional, scientific, and technical activities (NACE M), and administrative and support services (NACE N). These sectors are suited for analyzing agglomeration effects and firm selection dynamics, as they depend heavily on knowledge spillovers, customer proximity, and localized demand. Their relatively high firm turnover allows for firm selection to take effect quicker, allowing discovery of firm selection effects in shorter panels. 

We restrict the sample to urban areas, representing approximately 82\% of Spain's population. Urban areas are then classified as AMD or BMD based on whether they lie above or below the median experienced urban density \citep{delaRoca2017LearningWorkingBig}.

For the productivity analysis, we retain only firms observed for at least 18 years to control estimation noise. Specifically, if $i$ indexes firms, then $T_i$ ranges from 18 to 28 in eq. \eqref{equation:noisyThetaDefinition}. The number of such firms ($N$) mainly varies between 237 and 1996, depending on the sector and area type, with the exception of the trade sector in AMD (see figures below for exact values). The full sample, however, is used to estimate sector- and area-specific production functions \eqref{equation:cobb-douglas}.

\paragraph{Estimation of productivity}
 
	Firm-level productivity is estimated as follows. Let $a$ index density areas (AMD or BMD). We assume that firm $i$ in sector $s$ and area $a$ produces value added $V_{i,t}$ according to a Cobb-Douglas production function:
	\begin{equation}\label{equation:cobb-douglas}
		V_{it} = \exp(\theta_i) K_{it}^{\beta_{1, s, a}} L_{i\, t}^{\beta_{2, s, a}}\exp(u_{it}+ \beta_{0, t, s, a})
	\end{equation}
	where $\theta_i$ represents firm productivity (log total factor productivity, TFP), $K_{it}$ is capital, $L_{it}$ is labor, and $u_{it}$ captures measurement error in $V_{it}$. The parameters $\beta_{1, s, a}$ and $\beta_{2, s, a}$ are sector- and area-specific factor shares, while $\beta_{0, t, s, a}$ is a sector-, area-, and time-specific intercept. Firms in sector $s$ in AMD and BMD draw $\theta_i$ from latent distributions $F_{s, AM}$ and $F_{s, BM}$, respectively.
	
Productivity estimates $\vartheta_{i, T}$ are obtained in two steps. First, sector- and area-specific production functions \eqref{equation:cobb-douglas} are estimated using \cite{Ackerberg2015IdentificationPropertiesRecent}. Second, firm-specific log TFP is estimated with the average residual from the estimated production function:
	\begin{equation}
	\vartheta_{i, T}= T^{-1}\sum_{t=1}^T [\log V_{i\, t} - \hat{\beta}_{0, s, t} -\hat{\beta}_{1, s} \log K_{i\, t} - \hat{\beta}_{2, s} \log L_{i\, t}]	 
	\end{equation}
	To control estimation noise, we retain only firms observed for at least 18 years (see above).
 
\paragraph{Assumptions and rate conditions}
   
	Our analysis relies on two key assumptions: one concerning estimation error and the other on the underlying distributions $F_{s, AM}$ and $F_{s, BM}$. First, we assume that the estimation noise in $\vartheta_{i, T}$ has at least eight finite moments. This assumption primarily concerns measurement error, since the dominant source of noise is measurement error in $V_{it}$, while the error in estimating  $(\beta_{0, s, t}, \beta_{1, s}, \beta_{2, s})$ is negligible due to large total sample sizes. 
	Second, we assume that both the left and right tails of $F_{s, AM}$ and $F_{s, BM}$ are unbounded. This assumption is supported by the data, which exhibit heavy-tailed behavior, as discussed below.
	
	Under these assumptions, the rate conditions of Theorem \ref{theorem:feasibleEVT} hold, in line with Proposition \ref{proposition:ratesExtremeHalfUniformity} and our simulation results. A sufficient rate condition is $N/T^4 \approx 0$, which is broadly satisfied for all the sectors and areas. Simulations for $T=20$ (see Online Appendix) indicate that inference is reliable for all considered subsamples.

\paragraph{EV index estimation and evaluation of rate conditions}
 
	We begin by estimating the extreme value (EV) indices $\gamma$ for the left and right tails of the AMD and BMD distributions. Across sectors and density areas, estimates of $\gamma$ range from 0.2 to 0.28, with one exception --- the right tail of the administrative services sector in AMD, where $\gamma = 0.16$. These estimates are obtained using the PWM estimator (remark \ref{remark:gammaEstimation}), with its tuning parameter $k$ selected via the semiparametric bootstrap algorithm 4.3 in \cite{Caeiro2016ThresholdSelectionExtreme}. Results are robust to different choices of $k$.
  
	Three   observations emerge:
	\begin{enumerate}[noitemsep,topsep=0pt,parsep=0pt,partopsep=0pt, label={(\arabic*)}, leftmargin=*]
		\item 
		Heavy tails: the estimates indicate that productivity distributions are heavy-tailed, with 3-4 finite moments (except for the slightly lighter-tailed exception noted above).
		
		\item 	Support for infinite tails assumption: the estimated EV indices support the assumption that $F_{s, AM}$ and $F_{s, BM}$ have infinite support for all sectors $s$. Since measurement error is assumed to have at least eight moments,  the EV index of the estimation noise in $\vartheta_{i, T}$ must be at most 0.125. If estimation noise dominated the tails of $F_{s, AM}$ and $F_{s, BM}$, we would observe lower $\gamma$ estimates. This is not the case, supporting our assumption.
		
		\item 	No sharp survival threshold: the data does not support the existence a strict lower bound on firm productivity necessary for long-term survival (over at least 20 years), and goes against the firm selection hypothesis. This result strengthens the previous finding of \cite{Combes2012b} that  truncation must be the same between AMD and BMD, albeit at an unknown level. 
	\end{enumerate}

\begin{figure}[ht!]
 	\includegraphics[width=\linewidth]{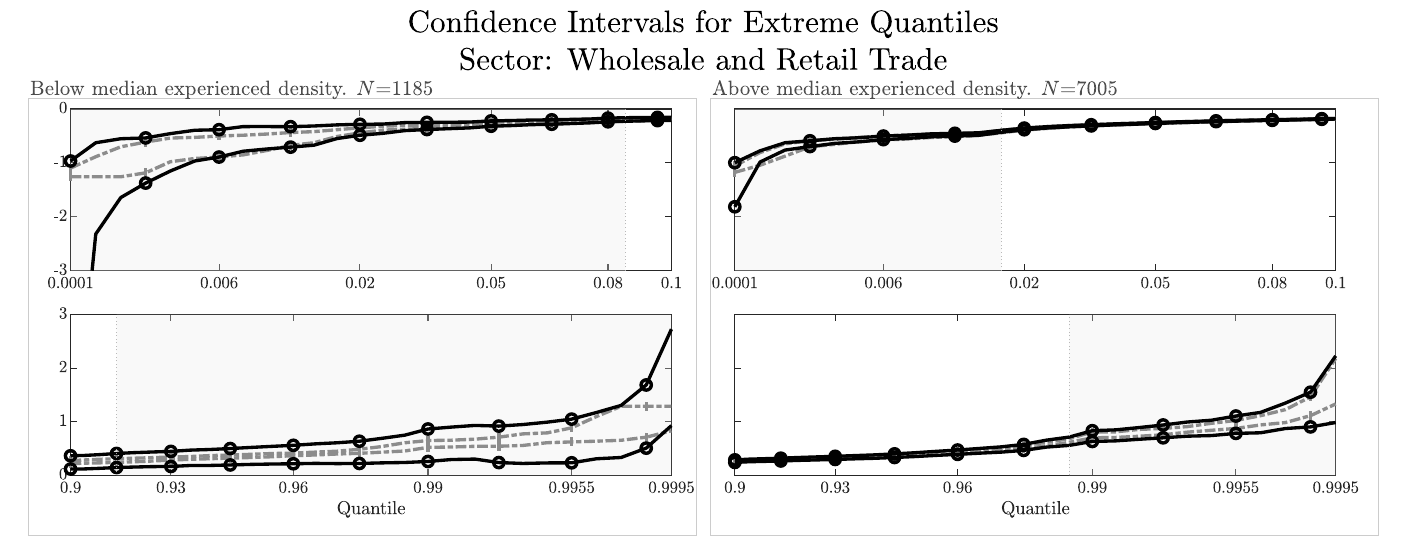}
 
 	\includegraphics[width=\linewidth]{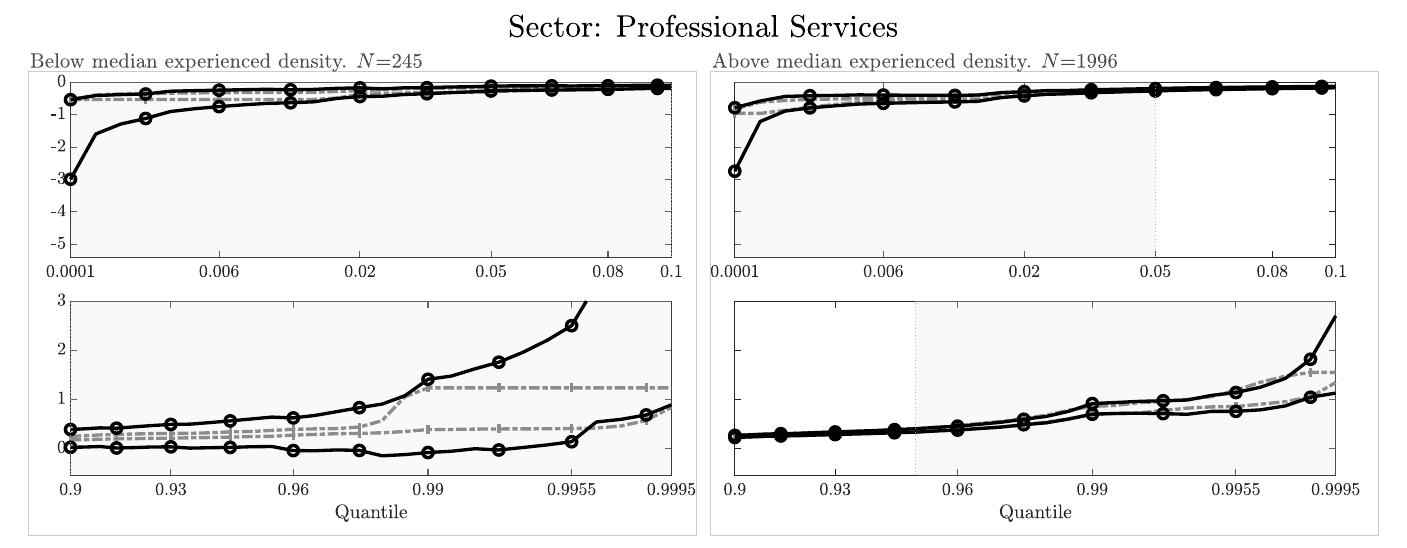}
 	
 	\includegraphics[width=\linewidth]{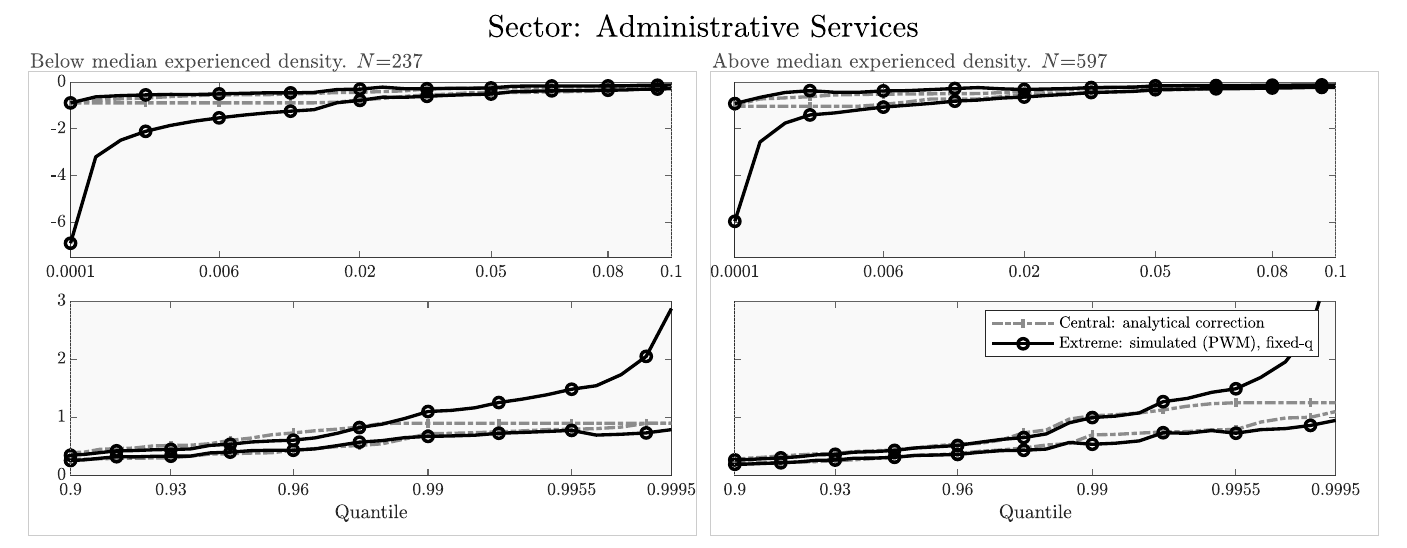}
 	
 	\caption{95\% confidence intervals for extreme quantiles of total factor productivity. Split by below and above median experience density (BME and AME, respectively); split by sectors.  Shaded area: zone where the rule of thumb of section \ref{section:inference} suggests extreme-order approximations. For areas and sectors with $N\leq 1000$, all the depicted quantiles fall into this zone. {\scriptsize [Data source: BELab, \cite{BancoDeEspana2024MicrodataIndividualEnterprises}, CBI data 1995-2023, own computations.] } }  \label{figure:empirical-split}
 \end{figure}
 
       \begin{figure}[!ht]
 	\centering
 	\includegraphics[width=\linewidth]{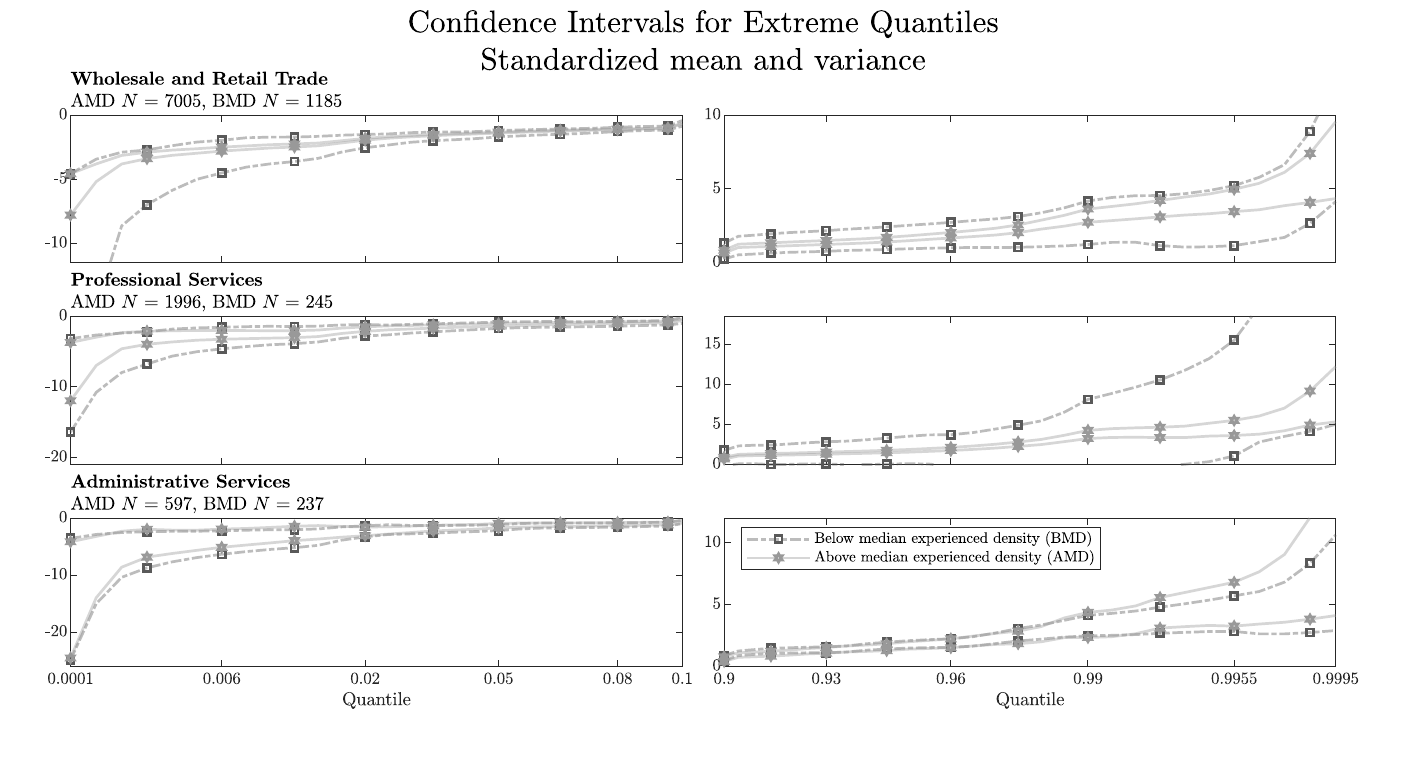}
 	\includegraphics[width=\linewidth]{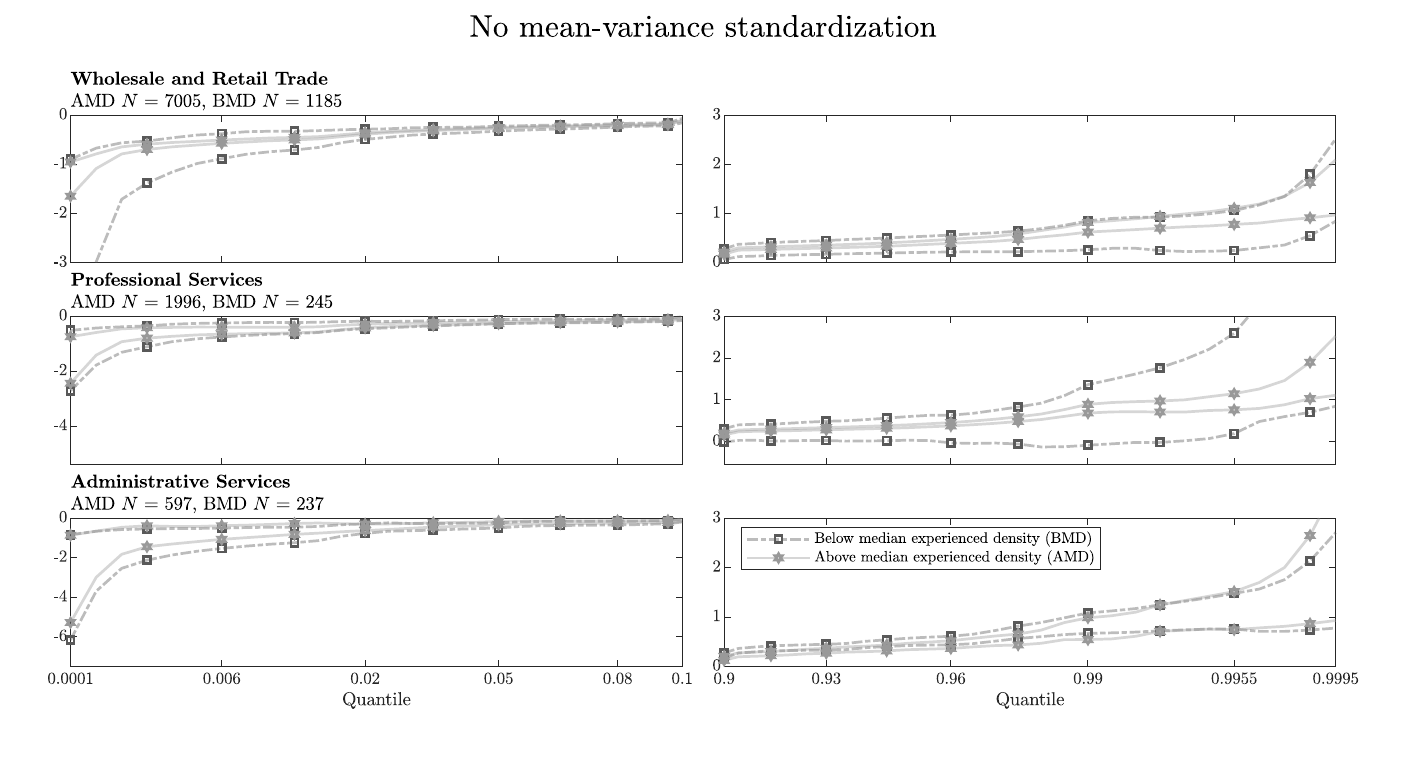}
 	
 	\caption{95\% confidence intervals for extreme quantiles of total factor productivity. Split by sector. Top panel: AMD and BMD data standardized to have the same mean and variance. Bottom panel: no standardization. CIs based on theorem \ref{theorem:feasibleEVT} with critical values estimated as in remark \ref{remark:critical-values-simulation}. {\scriptsize [Data source: BELab, \cite{BancoDeEspana2024MicrodataIndividualEnterprises}, CBI data 1995-2023, own computations.] }} \label{figure:empirical-joint}
 \end{figure}

\paragraph{Empirical results}
 
 	We now examine the tails of the productivity (TFP) distributions. We compute 95\% confidence intervals (CIs) for extreme quantiles—the 0.0001–0.1th and 0.9–0.9995th—of $F_{s, AM}$ and $F_{s, BM}$ for each sector $s$. Figures \ref{figure:empirical-split} and \ref{figure:empirical-joint} summarize the results. Figure \ref{figure:empirical-split} reports two sets of CIs: an extreme order CI (based on theorem \ref{theorem:feasibleEVT}, with critical values based on remark \ref{remark:critical-values-simulation}, implemented as in section \ref{section:simulations}) and a central order CI (based on theorem \ref{theorem:JW}). Results are split by area type and sector. The regions where extreme order approximations are recommended (see section \ref{section:inference}) are shaded in light gray. Figure \ref{figure:empirical-joint} overlays the extreme order CIs for AMD and BMD, allowing for direct comparison of the tails. We report results both with and without standardizing AMD and BMD to have the same mean and variance (in line with the key assumption of CDGPR12). Since the positive EV index estimates for the left tails of $F_{s, AM}$ and $F_{s, BM}$ are inconsistent with truncation, we do not modify the tails.

    Our main empirical finding is that the tails of \(F_{s, AM}\) and \(F_{s, BM}\) are similar across all three sectors, regardless of standardization. As shown in Figure \ref{figure:empirical-joint}, the CIs are effectively nested. This supports the key assumption of CDGPR12 --- their three parameters are sufficient to capture differences in the tails. Moreover, a stronger form of their assumption may hold: no parameters are required to explain the difference between the AMD and BMD tails (up to the available statistical precision of the data).  The data suggest that agglomeration effects are confined to non-extreme quantiles and other sectors. 	 
	
	We also make two statistical observations (figure \ref{figure:empirical-split}). First, the extreme and central order CIs agree fully, even in the regions where section \ref{section:inference} suggests   central order approximations (the non-shaded regions in figure \ref{figure:empirical-split}). Second, the behavior of the CIs aligns with the simulation results in Section \ref{section:simulations}: as the target quantile $\tau$ approaches 0 or 1, the length of the central order CI converges to zero.

\paragraph{Additional results}

Further results, including empirical estimates for all inference methods evaluated in Section \ref{section:simulations}, are provided in the Online Appendix.

\begin{remark}
	\label{remark:production-frontier}
	Our analysis (and that of CDGPR12)  follows the literature on production function estimation \citep{Ackerberg2007EconometricToolsAnalyzing}.  This should be contrasted with the literature on production frontier estimation (e.g. \cite{Kneip2015}). In our setting, each firm may be viewed as  being on its production frontier. Each frontier is characterized by eq. \eqref{equation:cobb-douglas}, and interest centers  on these firm-specific parameters, rather than an economy- or sector-wide production frontier.
\end{remark}

\begin{small}
	\setlength{\bibsep}{0pt plus 0.3ex}
\bibliographystyle{bbrr.bst}
\bibliography{../../../../eco.bib}

\end{small}

\part*{Proofs of Results in the Main Text}

\appendix

\numberwithin{equation}{section}

\allowdisplaybreaks

\section{Distributional Results}

\subsection{Proof of Theorem \ref{theorem:evtNoisy}  }

We   turn to the proof of theorem \ref{theorem:evtNoisy}.
We begin by stating some auxiliary results.
Let $H_T$ be the CDF of $\vartheta_{i, T} = \theta_i + T^{-p}\varepsilon_{i, T}$. Observe that $\curl{\vartheta_{i, T}}$ form a triangular array with rows indexed by $T$, and number of entries in each row given by $N$. In each row, the entries are IID and distributed according to the CDF $H_T$.
Define the  {auxiliary functions} $U_T$ and $U_F$ as
\begin{equation}\label{appendix:equation:definitionOfUT}
U_T= \left(\dfrac{1}{1-H_T} \right)^{-1}, \quad U_F = \left( \dfrac{1}{1-F}\right)^{-1}.
\end{equation} Using $U_T$ and $U_F$  greatly simplifies notation in subsequent proofs. 
Observe the following useful connection between $U_T$ and $H_T$.  Let $\tau\in (0, \infty)$. Let $N$ be large enough so that  $ N\tau>1$. Then
\begin{align} 
U_T(N \tau) 
&  = H_T^{-1}\left(1- \dfrac{1}{N\tau}\right) 	, \quad U_F(N\tau) = F^{-1}\left(1- \dfrac{1}{N\tau} \right)\label{equation:UHlink}.
\end{align}

We begin with a technical lemma that connects convergence of the normalized sample maximum $\vartheta_{N, N, T}= \max\curl*{\vartheta_{1, T}, \dots, \vartheta_{N, T} }$ and convergence of the quantiles of  $H_T$.

\begin{lem}\label{appendix:lemma:convergenceMaximumQuantiles}
	Let assumption \ref{assumption:independence} hold. The following are equivalent:
	\begin{enumerate}[noitemsep,topsep=0pt,parsep=0pt,partopsep=0pt, label={(\arabic*)}, leftmargin=*]
		\item As $N, T\to\infty$, for some constants $a_N, b_N$  the random variable ${(\vartheta_{N, N, T}- b_N)}/{a_N}$ converges weakly to a random variable $X$ with non-degenerate CDF $Q(x)$.
		\item As $N, T\to\infty$, for all $\tau\in (0, \infty)$ such that $\tau$ is a continuity of point  $Q^{-1}(\exp({-1/\tau}))$, it holds that  
$		( U_T(N \tau)-b_N)/{a_N} \Rightarrow Q^{-1}(\exp({-1/\tau}))$.
	\end{enumerate}
The same  is also true for $\theta_{N, N}$ with $U_F$ in place of $U_T$.
\end{lem}
The proof is completely analogous to the proof of theorem 1.1.2 in \cite{DeHaan2006}, which establishes  a similar result  for the IID case. We provide the full argument for completeness.
\begin{proof}
	
First consider the implication $(1)\Rightarrow (2)$. 
 Under assumption \ref{assumption:independence} the individual $\vartheta_{i, T}$ are IID RVs with CDF $H_T$, hence
  $P\left(({\vartheta_{N, N, T} - b_N })/{a_N}\leq x \right)= H_T^{N}(a_Nx+ b_N)$. By assumption, $H_T^{N}(a_Nx+ b_N)$  converges to some nondegenerate  CDF $Q(x)$ as $N, T\to\infty$ at all points of continuity of $Q(x)$. Let $x$ be a continuity point of $Q(x)$ such that $Q(x)\in (0, 1)$. Then at such an $x$  take logs
\begin{equation}%
N\log H_T(a_Nx+ b_N)\Rightarrow \log Q(x).
\end{equation}
For any such $x$ it must be that $H_T(a_Nx+b_N)\to 1$, otherwise the left hand side will diverge to $-\infty$. Since  $\log(1+x)\sim x$ for $x\sim 0$, this implies that   
\begin{equation} 
\lim_{N, T\to\infty} \frac{-\log H_T(a_Nx+b_N) }{1-H_T(a_Nx+b_N)}=1.
\end{equation} Replacing the logarithm  above, we obtain $
N\left(1-H_T(a_Nx+b_N) \right) \Rightarrow -\log Q(x)$.
 Taking reciprocals of both sides yields
\begin{equation}\label{appendix:equation:extremeConvergenceReciprocal}
\dfrac{1}{N\left(1-H_T(a_Nx+b_N) \right)}\Rightarrow \dfrac{1}{-\log Q(x)}.
\end{equation}
Taking inverses in eq. \eqref{appendix:equation:extremeConvergenceReciprocal} and using the definition \eqref{appendix:equation:definitionOfUT} of $U_T$, we obtain that weak convergence of the sample maximum is equivalent to the following convergence of quantiles of $H_T$:
\begin{equation}
\dfrac{ U_T(N y)-b_N }{a_N} \Rightarrow Q^{-1}(e^{-1/y}) \text{ as } N, T\to\infty~.
\end{equation}
Following the above argument in the reverse direction establishes $(2)\Rightarrow (1)$.
\end{proof}

The following lemma provides a condition on $U_T$ under which $\vartheta_{N, N, T}$  and  $\theta_{N, N}$ have the same limit properties.
\begin{lem}
 \label{lemma:weakConvergenceSameSequenceCloseQuantiles} 
If
\begin{enumerate}[noitemsep,topsep=0pt,parsep=0pt,partopsep=0pt, label={(\arabic*)}, leftmargin=*]
\item  Assumption \ref{assumption:independence} holds
\item  $a_N, b_N$ are such that as $N\to\infty$ the normalized noiseless maximum $
({\theta_{N, N}-b_N})/a_N $ converges weakly to a non-degenerate random variable $X$. 
\item   For each $\tau\in (0, \infty)$  it holds that 
${(U_T(N\tau)- U_F(N\tau))}/{a_N} \to 0 \text{ as }N, T\to\infty$
\end{enumerate}
then as $N, T\to\infty$ $ {(\vartheta_{N, N, T}-b_N)}/{a_N}\Rightarrow X$. 

\end{lem}

\begin{proof}
	Let  $Q$ be the CDF of $X$. Let $\tau\in (0, \infty)$ be a continuity point of $Q^{-1}(\exp(-1/\tau))$. Then by lemma  \ref{appendix:lemma:convergenceMaximumQuantiles}  $a_N^{-1}({ U_F(N \tau)-b_N })$ converges to $Q^{-1}(\exp(-1/\tau))$.
By assumption (3)  of the lemma for all $\tau\in(0, \infty)$   continuity points of $Q^{-1}(\exp(-1/\tau))$ it holds that
		\begin{equation}
	\dfrac{ U_T(N \tau)-{b}_N }{a_N} - \dfrac{ U_F(N \tau)-b_N }{a_N} = \dfrac{U_T(N\tau)- U_F(N\tau)}{a_N} \to 0 \text{ as }N, T\to\infty.
	\end{equation}
	From  this we conclude that $a_N^{-1}({ U_T(N \tau)-b_N })\to Q^{-1}(\exp(-1/\tau))$ for all $\tau\in(0, \infty)$ continuity points of  $Q^{-1}(\exp(-1/\tau))$.
The result follows from lemma \ref{appendix:lemma:convergenceMaximumQuantiles}.
\end{proof}

We will make use of the following quantile inequalities due to \cite{Makarov1981}.
\begin{lem}[Makarov quantile inequalities; eqs. (1) and (2) in \cite{Makarov1981}]
	\label{lemma:makarovInequalities} Suppose that $X\sim F_X$ and $Y\sim F_Y$ are a pair of random variables whose joint distribution is not restricted, and consider their sum $X+Y\sim F_{X+Y}$.   The following inequalities hold: for all $v\in[0, 1]$ 
		\begin{align}
		F^{-1}_{X+Y}(v)& \leq \inf_{w\in [v, 1]} \left(F_X^{-1}(w)  + F^{-1}_Y(1+v-w)\right), \\
		F^{-1}_{X+Y}(v) & \geq \sup_{w\in [0, v]}\left( F_X^{-1}(w) + F_Y^{-1}(v-w) \right).
		\end{align}
The bounds are pointwise sharp in the following sense: for each   $v$ there exists a joint distribution of $X$ and $Y$ such that the $v$th quantile of $X+Y$ attains the lower/upper bound at $v$.
\end{lem}

	\begin{proof}[Proof of theorem \ref{theorem:evtNoisy}]

	
	Let $H_T$ be the CDF of $\vartheta_{i, T}$.	 Fix $\tau\in(0, \infty)$. 
	Let $N$ be large enough so that $N\tau>1$. In the statement of  Makarov's inequalities (lemma \ref{lemma:makarovInequalities}), take $X=\theta_i$ and $Y= T^{-p}\varepsilon_{i, T}$,
	and $v=1-1/N\tau$ and subtract $F^{-1}(1-1/N\tau)$ on all sides to obtain
	\begin{align}
		& 	\sup_{w\in \left[0, 1-\frac{1}{N\tau}\right]}  \left( F^{-1}(w) + \dfrac{1}{T^p} G_T^{-1}\left(1-\frac{1}{N\tau}-w\right)  -F^{-1}\left(1- \dfrac{1}{N\tau} \right)   \right) 
		\\  \leq &  H^{-1}_T\left(1-\frac{1}{N\tau} \right) - F^{-1}\left(1- \frac{1}{N\tau} \right)  \label{appendix:equation:noisyEVTineqality} \\
		\leq &  	\inf_{w\in \left[1-\frac{1}{N\tau}, 1\right]} \left( F^{-1}(w)  +\dfrac{1}{T^p} G_T^{-1}\left(1+1-\frac{1}{N\tau}-w\right) - F^{-1}\left(1- \dfrac{1}{N\tau} \right) \right)~.
	\end{align}
	First,	by  definitions of $U_T$ and $U_F$ (eq. \eqref{appendix:equation:definitionOfUT}) and   eq. \eqref{equation:UHlink}, the middle of eq \eqref{appendix:equation:noisyEVTineqality} is equal to  $U_T(N \tau) - U_F(N\tau)$.
	%
	Second,  in the supremum condition define $u=1-1/(N\tau)-w$ to write the condition as
	\begin{equation}
		\sup_{u\in \left[0, 1-\frac{1}{N\tau}\right]}   \left(   F^{-1}\left(1- \dfrac{1}{N\tau}-u \right)  -F^{-1}\left(1- \dfrac{1}{N\tau} \right) + \dfrac{1}{T^p} G_T^{-1}\left(u\right)  \right)  .
	\end{equation}
		Note that eventually $[0, \epsilon] \subset [0, 1-1/(N\tau)]$, and so the expression in eq. \eqref{appendix:equation:noisyEVTineqality} may further be lower-bounded as	 $\sup_{u\in[0, \epsilon]}\curl{\dots}\leq \sup_{u\in [0, 1-1/(N\tau)]}\curl{\dots}$. Third,  define  $u=-[1-1/(N\tau)-w]	$ in the infimum condition in eq. \eqref{appendix:equation:noisyEVTineqality} to represent it as \eqref{equation:tailEquivalenceInf}.  
	 Suppose that $a_N>0$  (if not, simply reverse all inequalities below).  Combining the above arguments and multiply all sides   by $a_N^{-1}$, we obtain
		\begin{align}
		&  	\sup_{u\in \left[0, \epsilon\right]}  \dfrac{1}{a_N}\left( F^{-1}\left(1- \dfrac{1}{N\tau}-u \right)  -F^{-1}\left(1- \dfrac{1}{N\tau} \right) + \dfrac{1}{T^p} G_T^{-1}\left(u\right)  \right)
		\\  \leq &  \dfrac{U_T(N\tau)- U_F(N\tau)}{a_N}   \label{appendix:equation:noisyEVTinequalityNew} \\
		\leq &   	\inf_{u\in \left[0, \frac{1}{N\tau} \right]}   \dfrac{1}{a_N}\left(F^{-1}\left(1- \dfrac{1}{N\tau} + u\right) -  F^{-1}\left(1- \dfrac{1}{N\tau} \right) +\dfrac{1}{T^p} G_T^{-1}\left(1-u\right) \right) ~.
	\end{align}
	Then conditions \eqref{equation:tailEquivalenceInf}  and \eqref{equation:tailEquivalenceSup} imply that  $
	(U_T(N \tau) - U_F(N\tau))/{a_N} \to 0$
	for all $\tau \in (0, \infty)$. 
	By  lemma \ref{lemma:weakConvergenceSameSequenceCloseQuantiles} it follows that $(\vartheta_{N, N, T}-b_N)/{a_N}\Rightarrow X$. 
	
	
	Now we turn to the second assertion.
	First, the above shows that   \eqref{equation:tailEquivalenceSup} \eqref{equation:tailEquivalenceInf} together imply  \eqref{equation:tailEquivalenceSupGeneral}.  Now suppose that  at least one of \eqref{equation:tailEquivalenceInf} and \eqref{equation:tailEquivalenceSupGeneral} fails. We show that   the limit properties of $a_N^{-1}({\vartheta_{N, N, T}-b_N})$  differ from the limit properties of $a_N^{-1}({\theta_{N, N}-b_N})$
	for some sequence of joint distributions of $(\theta_i, \varepsilon_{i, T})$.
	Suppose it is the infimum condition  \eqref{equation:tailEquivalenceInf} that fails to hold for some $\tau$; the argument for \eqref{equation:tailEquivalenceSupGeneral} is identical.  Then along some subsequence  of $(N, T)$ it holds that	$\inf_{u\in \left[1-{1}/{N\tau}, 1\right]} a_N^{-1}\left(\cdots \right) = \delta_{N, T}$ such that $\delta_{N, T}$ are bounded away from zero. Suppose that it is possible to extract a further subsequence such that along it  $\delta_{N, T}$ converges to some $\delta\neq 0$.  Pass to that subsubsequence.
	Theorem 2 in	\cite{Makarov1981} establishes that  for each $(N, T)$ there exists a joint distribution of $\theta_i$ and $\varepsilon_{i, T}$ such that  the resulting $H_T$ attains the upper bound in inequality \eqref{appendix:equation:noisyEVTineqality}, and so
	\begin{align}
		H^{-1}_T\left( 1- \dfrac{1}{N\tau}\right) &  
		= a_N \delta_{N, T} + F^{-1}\left(1-\frac{1}{N\tau} \right).
	\end{align}
	If $({U_F(N\tau)-b_N})/{a_N}\to Q(\exp(-1/\tau))$, then for such a sequence of $H_T$ it holds that 
	$(U_T(N\tau)- b_N)/{a_N} \to Q^{-1}(\exp(-1/\tau)) + \delta$.
	Suppose that $(\vartheta_{N, N, T}-b_N)/a_N$ converges to a random variable with distribution $\tilde{Q}$ along the same subsequence. The above discussion shows that   $\tilde{Q}^{-1}(\exp(-1/\tau)) = Q^{-1}(\exp(-1/\tau)) + \delta\neq Q^{-1}(\exp(-1/\tau))$. Thus, either the limit distribution of $\vartheta_{N, N, T}$ is different from that of $\theta_{N, N}$, or $\vartheta_{N, N, T}$ does not converge.

	If we cannot extract a subsequence of $\delta_{N, T}$ converging to some finite $\delta$,  then $\delta_{N, T}$ is unbounded. In this case it is possible to extract a further monotonically increasing subsequence. Proceeding as above, we obtain that along that subsequence  it holds that
	${(U_T(N\tau)- b_N)}/{a_N} \to Q^{-1}(\exp(-1/\tau)) + \infty$, 
	and so $\left(\vartheta_{N, N, T}- b_N \right)/a_N$ does not converge.
\end{proof} 

\subsection{Proof of Proposition \ref{proposition:ratesExtremeHalfUniformity}}

Before proving proposition \ref{proposition:ratesExtremeHalfUniformity}, we state two useful results. 
First, let 
$RV_{\gamma}$ be the class of non-negative functions of regular variation with parameter $\gamma$, that is, those measurable 	$f:\R_+\to \R_+$ that satisfy $\lim_{t\to\infty}  {f(tx)}/{f(t)}= x^{\gamma}$ for any $x> 0$. $RV_{0}$ is the class of slowly varying functions.  
\begin{lem}[Karamata Characterization Theorem; Theorem  1.4.1 in \cite{Bingham1987}]
	\label{appendix:lemma:karamataCharacterization}
	Let $f\in RV_{\gamma}$. Then there exists a slowly varying function $L$ (that is, $L\in RV_{0}$) such that for all $x$ it holds that  $
	f(x)= x^{\gamma} L(x)$.
\end{lem}

 Second,  assumption \ref{assumption:EV} is equivalent to the following statement: for some sequences $\alpha_N, \beta_N$ and all $x>0$ 
 \begin{equation}\label{equation:noiselessFirstOrderAuxiliaryLimit}
 \lim_{  N\to \infty} \dfrac{ U_F(N x)-U_F(N)}{\alpha_N} = \dfrac{x^{\gamma}-1}{\gamma}  \text{ (meaning $\log(x)$ for  $\gamma=0$)},
 \end{equation}
 Since the left hand side is monotonic in $x$ and the right hand side is continuous, convergence in \eqref{equation:noiselessFirstOrderAuxiliaryLimit} is locally uniform in   $x$ (that is, for any $0<a, b<\infty$  convergence in \eqref{equation:noiselessFirstOrderAuxiliaryLimit}  is uniform on $[a, b]$). 
 See theorem 1.1.6  and corollary 1.2.4 in \cite{DeHaan2006}.
 \label{appendix:page:noiselessEVTconstants} 
We call the constants the $\alpha_N, \beta_N$ \emph{canonical} normalization constants.

\begin{proof}[Proof of proposition \ref{proposition:ratesExtremeHalfUniformity}]

Let $(a_N, b_N)$ be such that $(\theta_{N, N}- b_N)/a_N\Rightarrow X$ for some  non-degenerate random variable $X$.
	Fix $\tau\in (0, \infty)$ and define $s_{\tau, N, T}$ and $S_{\tau, N, T}$ as
		\begin{align} 
		S_{\tau, N, T}(u)  & = \dfrac{1}{a_N}\left(F^{-1}\left(1-\dfrac{1}{N\tau}+ u \right)  - F^{-1}\left(1- \dfrac{1}{N\tau} \right)  +\dfrac{1}{T^p} G_T^{-1}\left(1-u\right) \right),\\
		s_{\tau, N, T}(u)  & = \dfrac{1}{a_N}\left( F^{-1}\left( 1 - \dfrac{1}{N\tau} - u\right)  -F^{-1}\left(1- \dfrac{1}{N\tau} \right) + \dfrac{1}{T^p} G_T^{-1}\left(u\right)  \right) .
	\end{align}
 	Condition \eqref{equation:tailEquivalenceInf} for $F$ and $G_T$ can be written as $\inf_{u\in \left[0, {1}/{N\tau}\right]} S_{\tau, N, T}(u)\to 0$, while condition \eqref{equation:tailEquivalenceSup} can be written as $\sup_{u\in \left[0, \epsilon\right]} s_{\tau, N, T}(u)\to 0$ for some $\epsilon\in (0, 1)$

By	lemma \ref{lemma:makarovInequalities} and the discussion after eq. \eqref{appendix:equation:noisyEVTineqality}	\begin{align*}
		\sup_{u\in \left[0, 1-\frac{1}{N\tau}\right]}  \left[   F^{-1}\left(1-\dfrac{1}{N\tau}-u \right) + \dfrac{1}{T^p} G_T^{-1}\left(u\right)    \right] 
	\leq  	\inf_{u\in \left[0, \frac{1}{N\tau}\right]}  \left[ F^{-1}\left(1-\dfrac{1}{N\tau}+u \right)  +\dfrac{1}{T^p} G_T^{-1}\left(1-u\right)\right].
	\end{align*}
	Suppose that $a_N>0$ (if not, simply reverse all inequalities below). The above inequality implies that  $
	 	\sup_{u\in \left[0, 1-{1}/{N\tau}\right]} s_{\tau, N, T}(u)\leq 	\inf_{u\in \left[0, {1}/{N\tau}\right]} S_{\tau, N, T}(u) $.  Further, fix an arbitrary $\epsilon\in (0, 1)$. 
	 	Then eventually $[0, \epsilon]\subset [0, 1-1/N\tau]$, and correspondingly $	\sup_{u\in \left[0, \epsilon\right]} s_{\tau, N, T}(u)\leq 	\sup_{u\in \left[0, 1-{1}/{N\tau}\right]} s_{\tau, N, T}(u)$.
	 	 If $u_{s, \tau, N, T}\in [0, \epsilon]$ and $u_{S, \tau, N, T}\in [0, 1/N\tau]$, then the  following chain of inequalities  holds:
	 	\begin{equation}\label{appendix:equation:evtSufficientInequalityChain}
	 	s_{\tau, N, T}(u_{s, \tau, N, T})\leq 	\sup_{u\in \left[0, \epsilon\right]} s_{\tau, N, T}(u)\leq 	\inf_{u\in \left[0, \frac{1}{N\tau}\right]} S_{\tau, N, T}(u) \leq S_{\tau, N, T}(u_{S, \tau, N, T})~.
	 	\end{equation} 
	 	
	Let  $F$  satisfy assumption \ref{assumption:EV} with EV index $\gamma>\gamma'$.  Define 
	\begin{equation}
	u_{S, \tau, N, T}  = \frac{1}{N\tau} \frac{1}{\log (T)+1} \in \left[ 0,  \dfrac{1}{N\tau}  \right].
	\end{equation}
	Suppose that   under the conditions of the proposition it holds that
	\begin{align}
	\dfrac{1}{a_N}\left(F^{-1}\left( 1- \dfrac{1}{N\tau}+ 	u_{S, \tau, N, T} \right)    - F^{-1}\left(1- \dfrac{1}{N\tau} \right)\right) & \to 0 \label{appendix:equation:evtSufficientF}, \\
	\dfrac{1}{a_N}\dfrac{1}{T^p}G_T^{-1}\left(1-  	u_{S, \tau, N, T}  \right) & \to 0.\label{appendix:equation:evtSufficientG}
	\end{align} 
	Under  \eqref{appendix:equation:evtSufficientF} and \eqref{appendix:equation:evtSufficientG} it holds that    $S_{\tau, N, T}(u_{S, \tau, N, T})\to 0$ as $N, T\to\infty$. 
By inequality \eqref{appendix:equation:evtSufficientInequalityChain} we conclude that $\limsup_{N, T\to\infty} \inf_{u\in \left[0, {1}/{N\tau}\right]} S_{\tau, N, T}(u)\leq 0$. 
An identical argument shows that   $s_{\tau, N, T}(u_{s, \tau, N, T})\to 0$ where  $u_{s, \tau,N, T} =  ({1}/{N\tau}) (1/{\log (T)})$ eventually lies in  $[0, \epsilon]$ for any $\epsilon\in(0, 1)$; hence $\liminf_{N, T\to\infty} \sup_{u\in \left[0, \epsilon\right]} s_{\tau, N, T}(u)\geq 0$.  
Further,  $\limsup_{N, T\to\infty}  \sup_{u\in \left[0, \epsilon\right]} s_{\tau, N, T}(u)\leq \liminf_{N, T\to\infty} \inf_{u\in \left[0,  {1}/{N\tau}\right]} S_{\tau, N, T}(u) $ by inequality \eqref{appendix:equation:evtSufficientInequalityChain}. Combining the results,  we obtain that both $\inf_{u\in \left[0, {1}/{N\tau}\right]} S_{\tau, N, T}(u)$ and  $\sup_{u\in \left[0, \epsilon\right]} s_{\tau, N, T}(u)$ tend to 0, proving the result.
	
	\air 
	
	It remains to establish  \eqref{appendix:equation:evtSufficientF} and \eqref{appendix:equation:evtSufficientG}. 
		We split the proof of \eqref{appendix:equation:evtSufficientF} by sign of $\gamma$.

			Suppose $\gamma>0$. We begin by making  two observations.
			First,  by the convergence to types theorem \citep[proposition 0.2]{Resnick1987} $a_N\sim \alpha_N$ for the canonical scaling $\alpha_N$ of eq. \eqref{equation:noiselessFirstOrderAuxiliaryLimit}.   Further, by lemma 1.2.9 of \cite{DeHaan2006} $\alpha_N\sim U_F(N)$. 
			 Second, 						   by corollary 1.2.10 in \cite{DeHaan2006}  $(U_F(x))^{-1}  \in RV_{-\gamma}$. 
						   Then by lemma \ref{appendix:lemma:karamataCharacterization} we can write $1/U_F(x)  =   x^{-\gamma} L(x)$ where $L$ is a slowly varying function (that depends on $F$).
			Combining the above observations, definition of $u_{S, \tau, N, T}$ and  eq.  \eqref{equation:UHlink}, we see that 
			\begin{align}
			& 	\dfrac{1}{a_N}\left(F^{-1}\left(1- \frac{1}{N\tau} \frac{\log(T)}{\log (T)+1} \right)    - F^{-1}\left(1- \dfrac{1}{N\tau} \right)\right)\\
			& \sim  \dfrac{1}{U_F(N)}\left(U_F\left(N\tau\dfrac{\log(T)+1}{\log(T)} \right)- U_F(N\tau) \right)\\
			&=  N^{-\gamma}L(N)\left(  \left(N\tau \dfrac{\log (T)+1}{\log (T)} \right)^{\gamma} \dfrac{1}{L\left( N\tau \frac{\log (T)+1}{\log (T)} \right)}     - (N\tau)^{\gamma} \dfrac{1}{L(N\tau)} \right)\\
			& \propto 	\left( \dfrac{\log(T)+1}{\log (T)} \right)^{\gamma} \dfrac{L(N) }{L\left( N\tau \frac{\log (T)+1}{\log (T)} \right)} - \dfrac{L(N)}{L(N\tau)}   \to  0.
			\end{align}
		Convergence follows since $L$ is slowly varying on infinity: $ {L(N)}/{L(N\tau)}\to 1$.  By local  uniform convergence (proposition 0.5 in \cite{Resnick1987}) $ {L(N) }/{L\left( N\tau  {(\log (T)+1)}/{\log (T)} \right)} \to 1$.

	Suppose $\gamma<0$. As above, $a_N\sim \alpha_N$. By lemma 1.2.9 in \cite{DeHaan2006} in turn $\alpha_N\sim (U_F(\infty)-U_F(N))$ (note that $U_F(\infty)<1$ when $\gamma<0$).  By corollary 1.2.10 in \cite{DeHaan2006}, $ (U_F(\infty)- U(x))^{-1}\in RV_{-\gamma}$   and we can write $1/(U_F(\infty)- U(x))= x^{-\gamma}L(x)$ for some slowly varying function $L$ by lemma \ref{appendix:lemma:karamataCharacterization}.  Hence proceeding as for $\gamma>0$
	\begin{align}
	& 	\dfrac{1}{a_N}\left(F^{-1}\left(1- \frac{1}{N\tau} \frac{\log(T)}{\log (T)+1} \right)    - F^{-1}\left(1- \dfrac{1}{N\tau} \right)\right)\\
	  &\sim  \dfrac{1}{U_F(\infty)-U_F(N)}\left( \left( U_F\left(N\tau\dfrac{\log(T)+1}{\log(T)} \right)- U_F(\infty) \right)- ( U_F(N\tau)- U(\infty)) \right)\\
	& = N^{-\gamma}L(N)\left(  \left(N\tau \dfrac{\log (T)+1}{\log (T)} \right)^{\gamma} \dfrac{1}{L\left( N\tau \frac{\log (T)+1}{\log (T)} \right)}     - (N\tau)^{\gamma} \dfrac{1}{L(N\tau)} \right)\\
	 &  \propto	\left( \dfrac{\log(T)+1}{\log (T)} \right)^{\gamma} \dfrac{L(N) }{L\left( N\tau \frac{\log (T)+1}{\log (T)} \right)} - \dfrac{L(N)}{L(N\tau)} \to  0.
	\end{align}

	Suppose $\gamma=0$. As above, $a_N\sim \alpha_N$. By eq. \eqref{equation:noiselessFirstOrderAuxiliaryLimit} 	$\lim_{N\to\infty} {(U_F(Nx)-U_F(N))}/{\alpha_N} = \log (x)$ locally uniformly in $x$.  Then 
	\begin{align} 
	& 	\dfrac{1}{a_N}\left(F^{-1}\left(1- \frac{1}{N\tau}  \frac{\log(T)}{\log (T)+1} \right)    - F^{-1}\left(1- \dfrac{1}{N\tau} \right)\right)\\
	& \sim \dfrac{1}{\alpha_N} \left(U_F\left( N\tau \frac{\log (T)+1}{\log (T)} \right)- U_F(N\tau) \right) \to   \log\left(\lim_{T\to\infty} \dfrac{\log(T)+1}{\log(T)}  \right)  =0.
	\end{align} 
	
	Now  we focus on the $G^{-1}_T$ term in eq. \eqref{appendix:equation:evtSufficientG}.
	First suppose that  $\sup_T\E\abs*{\varepsilon_{i, T}}^{\beta}<\infty$. Then by Markov's inequality we obtain for any $\tau\in (0, 1)$ that  $G_T^{-1}(\tau) \leq \left({\sup_T \E\abs{\varepsilon_{i, T}}^{\beta} }/({1-\tau}) \right)^{1/\beta}$. Hence  \begin{equation}\label{appendix:equation:sufficientExtremeGorder}
	 \frac{1}{T^p}G^{-1}_T \left( 1 - \frac{1}{N\tau} \frac{1}{\log (T)+1}   \right) = O\left(  \frac{N^{1/\beta} (\log (T))^{1/\beta}}{T^{p}} 	\right)
	 \end{equation}
	 where the  $O$ term is uniform in $T$.
	First suppose that $\gamma\neq 0$. 
	As above,  $a_{N}^{-1}\sim x^{-\gamma}L(x)$ where $L$ is a slowly varying function (that depends on $F$). For eq. \eqref{appendix:equation:evtSufficientG} to hold   it is sufficient that  ${N^{1/\beta- \gamma}L(x)(\log (T))^{1/\beta}}/{T^{p}}\to 0$ (by eq. \eqref{appendix:equation:sufficientExtremeGorder}). Write $\gamma= \gamma'+\delta$, $\delta>0$. Then
	 \begin{equation} 
	\frac{N^{1/\beta- \gamma}L(x)(\log (T))^{1/\beta}}{T^{p}} = \left[ \frac{N^{1/\beta- \gamma'}(\log (T))^{1/\beta}}{T^{p}}\right] \left[ \frac{L(x)}{N^{\delta}}\right] \to 0
	\end{equation}since the condition holds for $\gamma'$ and  ${L(x)}/{x^{\delta}}\to 0$ for any $\delta>0$ ($L$ is slowly varying).
	Second, 	
let $\gamma=0$. 
	Then \eqref{appendix:equation:evtSufficientG} holds  if ${N^{1/\beta}  (\log (T))^{1/\beta}}/(\alpha_N T^{p})\to 0$.  Fix an arbitrary $x>0$ and write this as
	\begin{align}
	 	\dfrac{ N^{1/\beta} (\log T)^{1/\beta}}{\alpha_N T^p} = \dfrac{  N^{-1/\beta-\gamma'}(\log (T))^{1/\beta}  }{T^{p}}  \dfrac{U_F(Nx) -U_F(N)}{\alpha_N} \dfrac{1}{N^{-\gamma'} (U_F(Nx) - U_F(N))}
	\end{align}
	By assumption, the first fraction tends to zero. The second fraction tends to $\log(x)$ by eq. \eqref{equation:noiselessFirstOrderAuxiliaryLimit}. Last, there are two possibilities for $U_F(Nx)-U_F(N)$. The first one is that it is bounded away from zero. The second one is that it converges to zero (possibly along a subsequence); in this case convergence is slower than $N^{-\kappa}$ for all $\kappa>0$ by problem 1.1.1(b) in \cite{Resnick1987} (recall that $U_F(N) = F^{-1}(1-1/N)$ by eq. \eqref{equation:UHlink}). In both cases the last fraction converges to zero, as $\gamma'<0$.

	The proof is identical  if $G_T\sim N(\mu_T, \sigma_T^2)$. Assumption \ref{assumption:tightness} implies that $\mu_T$ and $\sigma_T^2$ are bounded. In this case
	\begin{equation}
	\dfrac{1}{T^p}G^{-1}_T\left(1- \dfrac{1}{N\tau}\left(1-\frac{\log(T)}{\log(T)+1} \right) \right) = O\left(\dfrac{ \sqrt{\log (N)}}{T^p} \right),
	\end{equation}
	where the $O$ term is uniform in $T$.  The rest of the argument proceeds as above.
\end{proof}

\subsection{Proof of Theorem \ref{theorem:intermediateNormality}}

\begin{proof}[Proof of theorem \ref{theorem:intermediateNormality}]
	
	By  theorem 2.2.1 in \cite{DeHaan2006}
	\begin{equation}\label{appendix:equation:intermediateNormality}
		\sqrt{k}\dfrac{ \theta_{N-k, N} - U_F\left(\frac{N}{k}\right)  }{\frac{N}{k}U'_F\left(\frac{N}{k} \right)} \Rightarrow N(0, 1).
	\end{equation}
	where $({N}/{k})U'_F\left({N}/{k} \right) = (N/k)\times \left(\left({1}/({1-F})  \right)^{-1}\right)'\left( N/k \right)\equiv c_N$ (by eqs. \eqref{appendix:equation:definitionOfUT} and \eqref{equation:UHlink}).
	We now transfer this convergence property to $\vartheta_{N-k, N}$.

	It is convenient for the purposes of the proof to replace the uniform random variables $U_i$ with $1/U_i$. 	Let $Y_1, \dots, Y_N$ be IID random variables with CDF $F_Y(y)=1-1/y, y>1$, $F_Y(y)=0$ for $y\leq 1$. Observe that $Y_i \overset{d}{=} 1/U_i$ where $U_i$ is Uniform[0, 1].
	Let $Y_{1, N}\leq \dots \leq Y_{N, N}$ be the order statistics, then $U_{k, N} \overset{d}{=} 1/Y_{N-k, N}$.  As pointed out by \cite{DeHaan2006} (p. 50)
	\begin{align}
		\theta_{N-k, N} & \overset{d}{=} U_F(Y_{N-k, N}), \quad \label{appendix:equation:ivtProofUFtheta}
		\vartheta_{N-k, N, T}  \overset{d}{=} U_T(Y_{N-k, N}).
	\end{align}
	Let $c_N = ({N}/{k})U'_F\left({N}/{k} \right)$, as in the statement of the theorem.  Then
	\begin{align} \label{appendix:equation:intermediateDecomposition}
		\sqrt{k}\dfrac{\vartheta_{N-k, N, T} - U_F\left(\frac{N}{k}\right)  }{c_N} 
		& \overset{d}{=} \sqrt{k}\dfrac{ U_T\left(  Y_{N-k, N} \right) - U_F\left(\frac{N}{k}\right)  }{c_N} \\
		& =  \sqrt{k}\dfrac{ U_F\left(  Y_{N-k, N} \right)- U_F\left(\frac{N}{k}\right) }{c_N}   + \dfrac{\sqrt{k}}{c_N}\left( U_T\left(  Y_{N-k, N} \right) -   U_F\left(  Y_{N-k, N} \right)    \right)
	\end{align}
	By eqs.  \eqref{appendix:equation:intermediateNormality} and \eqref{appendix:equation:ivtProofUFtheta} it follows that as $N\to\infty$ the first term converges weakly to a $N(0, 1)$ variable.
	The conclusion of the theorem follows if  
	\begin{equation}\label{appendix:equation:ivtProofTailZero}
		\dfrac{\sqrt{k}}{c_N }\left( U_T\left(  Y_{N-k, N} \right) -   U_F\left(  Y_{N-k, N} \right)    \right) \xrightarrow{p} 0~. 
	\end{equation}	
	We establish eq. \eqref{appendix:equation:ivtProofTailZero} by an argument similar to the one used in the proof of theorem \ref{theorem:evtNoisy}. We use lemma 
	\ref{lemma:makarovInequalities} to bound $U_T(Y_{N-k, N})= F^{-1}\left(1- 1/Y_{N-k, N} \right)$.  In lemma    \ref{lemma:makarovInequalities} take $X=\theta_i$, $Y= T^{-p}\varepsilon_{i, T}$, $v= 1-1/Y_{N-k, N}$, 
	subtract $U_F(Y_{N-k, n})= F^{-1}\left(1-1/Y_{N-k, N} \right)$ on all sides  and multiply by $\sqrt{k}/c_N$ to obtain
	\begin{align}
		&	\sup_{w\in \left[0, 1-\frac{1}{Y_{N-k, N}}\right]}  \dfrac{\sqrt{k}}{c_N}  \left( F^{-1}(w) + \dfrac{1}{T^p} G^{-1}_T\left(1-\frac{1}{Y_{N-k, N}}-w\right)  -F^{-1}\left(1- \dfrac{1}{Y_{N-k, N}} \right) \right) \\
		& \leq   \dfrac{\sqrt{k}}{c_N} \left( U_T(Y_{N-k, N})- U_F(Y_{N-k, N}) \right)\\
		& \leq  	\inf_{w\in \left[1-\frac{1}{Y_{N-k, N}}, 1\right]}  \dfrac{\sqrt{k}}{c_N}  \left(F^{-1}(w)  +\dfrac{1}{T^p} G^{-1}_T\left(1+1-\frac{1}{Y_{N-k, N}}-w\right) - F^{-1}\left(1- \dfrac{1}{Y_{N-k, N}} \right)\right),
	\end{align}
	if  $c_N$ is non-negative; the opposite inequalities hold if $c_N$  is negative.
	Since $1/Y_{N-k, N}\overset{d}{=} U_{k, N}$, we obtain 
	\begin{align}
		& 	\inf_{w\in \left[1-\frac{1}{Y_{N-k, N}}, 1\right]}  \dfrac{\sqrt{k}}{c_N}  \left(F^{-1}(w)  +\dfrac{1}{T^p} G^{-1}_T\left(1+1-\frac{1}{Y_{N-k, N}}-w\right) - F^{-1}\left(1- \dfrac{1}{Y_{N-k, N}} \right)\right)\\
		& \overset{d}{=}	\inf_{w\in \left[1-U_{k, N}, 1\right]}  \dfrac{\sqrt{k}}{c_N }  \left(F^{-1}(w)  +\dfrac{1}{T^p} G_T^{-1}\left(1+1-U_{k, N}-w\right) - F^{-1}\left(1-  U_{k, N} \right)\right)   \xrightarrow{p} 0
	\end{align}
 Define $u = -[1-1/Y_{N-k, N} -w]$ to write the above   as condition \eqref{equation:intermediateConditionRandomInf}.
	For the supremum, instead define $u= 1- U_{k, N}-u$ and proceed as in the proof of theorem \ref{theorem:evtNoisy}, noting that $U_{k, N}\xrightarrow{p}0$ and hence with probability approaching 1 $[0, \epsilon]\subset [0, 1- U_{k, N}]$.
	 Eq. \eqref{appendix:equation:ivtProofTailZero} follows as in the proof of theorem \ref{theorem:evtNoisy}.
	
	Sharpness of conditions \eqref{equation:intermediateConditionRandomSup} and \eqref{equation:intermediateConditionRandomInf}  is established as in the proof of theorem \ref{theorem:evtNoisy}.  Suppose that  condition \eqref{equation:intermediateConditionRandomInf} fails (the case for condition \eqref{equation:intermediateConditionRandomSup} is analogous). There is some subsequence of $(N, T)$ and some $\delta_{N, T}$ such  that $\inf_{u\in \left[1-U_{k, N}, 1\right]} \sqrt{k}c_N^{-1}(\cdot) = \delta_{N, T}$ and  $\delta_{N, T}$ is bounded away from zero. Suppose that it is possible to extract a further subsequence such that $\delta_{N, T}$ converges to some $\delta \neq 0$. By theorem 2 in \cite{Makarov1981}, there exists a joint distribution of $\theta_i$ and $\varepsilon_{i, T}$ such that the infimum is attained. Then along this subsequence for this joint distribution $
	{\sqrt{k}}c_N^{-1}\left( U_T\left(  Y_{N-k, N} \right) -   U_F\left(  Y_{N-k, N} \right)    \right) \xrightarrow{p} \delta$. Then from equations \eqref{appendix:equation:intermediateNormality}, \eqref{appendix:equation:ivtProofUFtheta}, and \eqref{appendix:equation:intermediateDecomposition}  it follows that $\sqrt{k}c_N^{-1}[ \vartheta_{N-k, N,  T} - U_F(N/k) ]\Rightarrow N(\delta, 1)$.
	This convergence result may or may not hold for the overall original sequence. 
	If no convergent subsequence of $\delta_{N, T}$ can be extracted, $\delta_{N, T}$ is unbounded. Extract a further monotonically increasing subsequence. There exists a sequence of joint distributions of $\theta_i$ and $\varepsilon_{i, T}$ such that along that subsequence  $\vartheta_{N-k, N, T}$ diverges.
\end{proof}

\subsection{Proof of Proposition \ref{proposition:ratesIntermediateHalfUniformity}}

\begin{proof}[Proof of proposition \ref{proposition:ratesIntermediateHalfUniformity}]
	
	The proof proceeds similarly to that of proposition \ref{proposition:ratesExtremeHalfUniformity}.  As in the proof  of proposition \ref{proposition:ratesExtremeHalfUniformity}, it is sufficient to prove that for some $\tilde{u}_{S, N, T}\in \left[0, U_{k, N}\right]$ and $\tilde{u}_{s,  N, T}$ that with probability approaching 1 lies in  $[0, \epsilon]$ for some $\epsilon\in(0, 1)$.
	\begin{align} 
& 	\dfrac{\sqrt{k}}{c_N }  \left(F^{-1}(1-U_{k, N} + \tilde{u}_{S, N, T}) - F^{-1}\left(1-  U_{k, N} \right)    +\dfrac{1}{T^p} G_T^{-1}\left(1-\tilde{u}_{S, N, T}\right)  \right) \xrightarrow{p} 0, \label{appendix:equation:intermediateSufficientI}\\
 & \dfrac{\sqrt{k}}{c_N} \left( F^{-1}(1-U_{k, N}- \tilde{u}_{s, N, T})  -F^{-1}\left(1- U_{k, N} \right)   + \dfrac{1}{T^p} G_T^{-1}\left(\tilde{u}_{s, N, T}\right)\right) \xrightarrow{p} 0. \label{appendix:equation:intermediateSufficientS}
		\end{align}
	We only show that    eq. \eqref{appendix:equation:intermediateSufficientI} holds, eq. \eqref{appendix:equation:intermediateSufficientS} follows analogously.

	Let $\rho=\delta/2+\nu$, and set 
	\begin{equation}
	\tilde{u}_{S, N , T} =  U_{k, N} \frac{1}{N^{\rho}+1} \in [0, U_{k, N}].
	\end{equation}
	As in the proof of proposition \ref{proposition:ratesExtremeHalfUniformity}, we first show that the scaled $F^{-1}$ terms in eq. \eqref{appendix:equation:intermediateSufficientI} decay, and then that  the scaled $G^{-1}_T$ term decays.
	First we establish that
	\begin{equation}\label{appendix:equation:intermediateSufficientF}
	\dfrac{\sqrt{k}}{c_N}\left(  F^{-1}\left(  1- U_{k, N} \frac{N^\rho}{N^{\rho}+1} \right) - F^{-1}\left(1- U_{k, N} \right) \right)\xrightarrow{p} 0.
	\end{equation}
	Let $Y_1, \dots, Y_N$ be IID random variables with CDF $F_Y(y)=1-1/y, y>1$, $F_Y(y)=0$ for $y\leq 1$. We will use that $Y_i\overset{d}{=} 1/U_i$, and correspondingly $Y_{N-k, N}\overset{d}{=} 1/U_{k, N}$, as in the proof of theorem \ref{theorem:intermediateNormality}.
	Observe that $c_N$ can be written as $c_N= (N/k)U'_F(N/k)$. Then using eq. \eqref{equation:UHlink} we obtain that
	\begin{align*}
		& 	\dfrac{\sqrt{k}}{c_N}\left(  F^{-1}\left(  1- U_{k, N} \frac{N^\rho}{N^{\rho}+1} \right) - F^{-1}\left(1- U_{k, N} \right) \right)\\
			& \overset{d}{=} \dfrac{\sqrt{k}}{\frac{N}{k}U'_F\left(\frac{N}{k} \right)} \left(U_F\left( Y_{N-k, N} \frac{N^{\rho}+1}{N^{\rho}}  \right)  - U_F(Y_{N-k, N})  \right)		\\
			& = \dfrac{\sqrt{k}}{\frac{N}{k}U'_F\left(\frac{N}{k} \right)} \left(U_F\left(\frac{N}{k} \left(\frac{k}{N} Y_{N-k, N}\right) \frac{N^{\rho}+1}{N^{\rho}}  \right)  - U_F\left(\frac{N}{k} \left(\frac{k}{N} Y_{N-k, N}\right) \right)  \right)	\\
			& =   \sqrt{k} \left(\dfrac{N^{\rho}+1}{N^{\rho}}-1 \right)\left(\dfrac{k}{N}Y_{N-k, N} \right)\dfrac{\frac{N}{k} U_F'\left(\frac{N}{k}x_N \right)}{\frac{N}{k}U'_F\left(\frac{N}{k} \right)}, \quad x_N\in \left[\frac{k}{N} Y_{N-k, N}, \left(\frac{k}{N} Y_{N-k, N}\right)  \frac{N^{\rho}+1}{N^{\rho}}    \right]
	\end{align*} 
	where the last line follows by the mean value theorem.   We now deal with the last two terms in the above expression.
By corollary 2.2.2 in \cite{DeHaan2006}, $ ({k}/{N})Y_{N-k, N}\xrightarrow{p} 1$.
	By corollary 1.1.10 in \cite{DeHaan2006}  under assumptions \ref{assumption:independence} and \ref{assumption:vonMises} it holds that $
	\lim_{t\to\infty} {U'_F(tx)}/{U'_F(t)} = x^{\gamma-1}$
	locally uniformly in $x$. Since $x_N \to 1$ as $N\to\infty$ and $k\to\infty, k=o(N)$, we conclude that $
	{ U_F'\left(\frac{N}{k}x_N \right)}/{ U'_F\left(\frac{N}{k} \right)}\to 1.$
	Combining these observations, 
	we obtain that
	\begin{align}
& 	\dfrac{\sqrt{k}}{\frac{N}{k}U'_F\left(\frac{N}{k} \right)} \left(U_F\left(\frac{N}{k} \left(\frac{k}{N} Y_{N-k, N}\right) \frac{N^{\rho}+1}{N^{\rho}}  \right)  - U_F\left(\frac{N}{k} \left(\frac{k}{N} Y_{N-k, N}\right) \right)  \right)\\
	& =  O_p\left(    \sqrt{k}\left[  \dfrac{N^{\rho}+1}{N^{\rho} }-1  \right]  \right) = O_p(N^{\delta/2} N^{-\left(\delta/2+\nu \right)})= O_p(N^{-\nu})=o_p(1).
	\end{align}
	where we apply   the assumption that $k=N^{\delta}$ in the third equality.
 
	Now we show that
	\begin{equation}\label{appendix:equation:intermediateSufficientG}
	\dfrac{\sqrt{k}}{c_N} \dfrac{1}{T^p} G_T^{-1}\left( 1-\tilde{u}_{S, N, T} \right) \xrightarrow{p} 0.
	\end{equation}
	Suppose that  $\sup_T \E\abs{\varepsilon_{i, T}}^{\beta}<\infty$ (proof 
 is analogous if $G_T\sim N(\mu_T, \sigma^2_T)$).
	As in eq. \eqref{appendix:equation:sufficientExtremeGorder}
 \begin{equation}
 \dfrac{1}{T^p}G_T^{-1}\left(1- U_{k, N}\frac{1}{N^{\rho}+1}  \right) \sim  O\left(  \dfrac{N^{\rho/\beta} }{U_{k, N}^{1/\beta} T^p}  \right)
 \end{equation}
 By corollary 1.1.10 in \cite{DeHaan2006} $U'_F\in RV_{\gamma-1}$, hence for some slowly varying function $L$ we can write $U'_F(x)=x^{\gamma-1}L(x)$ (be lemma \ref{appendix:lemma:karamataCharacterization}). Then $c_N= ({N}/{k})U_F'\left({N}/{k} \right)= \left( {N}/{k}\right)^{\gamma}L\left( {N}/{k} \right)$. 
Hence	
	\begin{align}
 & 	\dfrac{\sqrt{k}}{c_N} \dfrac{1}{T^p} G_T^{1}\left( 1-\tilde{u}_{S, N, T} \right)  = O\left(\sqrt{k} \frac{1}{\left(\frac{N}{k}\right)^{\gamma}L\left(\frac{N}{k} \right) } \dfrac{N^{\rho/\beta} }{U_{k, N}^{1/\beta} T^p}   \right)\\
& =  O_p\left( k^{1/2+\gamma-1/\beta}  N^{-\gamma+\rho/\beta+1/\beta} \dfrac{1}{L\left( \frac{N}{k}\right)} \dfrac{1}{T^p} \right) 
 = O_p\left( 	\dfrac{ N^{\delta/2(1+1/\beta) + (1-\delta)(-\gamma'+1/\beta) + \nu/\beta }}{T^{p}}\dfrac{ N^{-\varkappa(1-\delta)}}{L\left(N^{1-\delta} \right)}    \right)
	\end{align}
	where  in the second equality we again use corollary 2.2.2 in \cite{DeHaan2006} to conclude that $ ({N}/{k}) \times 1/ U_{k, N}\overset{d}{=}\left(  {k}/{N}\right)Y_{N-k, N} \xrightarrow{p}1$; in the third line we write $\gamma=\gamma'+\varkappa$ where $\varkappa>0$ by assumption.  The  above expression is now $o_p(1)$ by assumptions of the proposition and since $L$ is slowly varying.

	Combining together equations \eqref{appendix:equation:intermediateSufficientF} and \eqref{appendix:equation:intermediateSufficientG} shows that \eqref{appendix:equation:intermediateSufficientI} holds. To prove that \eqref{appendix:equation:intermediateSufficientS} holds, proceed as above $
		\tilde{u}_{s, N, T} =  U_{k, N}/N^{\rho}$; observe that $\tilde{u}_{s, N, T}$ lies in $[0, \epsilon]$ with probability approaching 1 for any $\epsilon\in (0, 1)$. 
\end{proof}

\section{Inference}

%

\subsection{Proof of Lemmas \ref{corollary:evtNoisyCenteredQuantile}- \ref{theorem:jointEVT}, Theorems \ref{theorem:feasibleEVT}-\ref{theorem:subsampling}, Remark \ref{remark:gammaEstimation}}

\begin{proof}[Proof of lemma \ref{corollary:evtNoisyCenteredQuantile}]
	
	We split the proof by sign of $\gamma$. First,  let $\gamma>0$. Let $Fr$ be  an RV  with $P(Fr\leq x) = \exp(-x^{-1/\gamma})$ for $x\geq 0$ and $0$ for $x<0$; note that $(E_1^*)^{-\gamma}\overset{d}{=} Fr$. Then
	\begin{align}
		& 	\dfrac{1}{ F^{-1}\left(  1 - \frac{1}{N} \right)} \left[  \vartheta_{N, N, T}- F^{-1}\left( 1-\dfrac{l}{N} \right)\right]
	 = 	\dfrac{1}{U_F(N)} \left[  \vartheta_{N, N, T}- U_F\left( \dfrac{N}{l} \right)\right]\\
	& = \dfrac{1}{U_F(N)}\vartheta_{N, N, T} - \dfrac{U_F(N \times l^{-1})}{U_F(N)} \Rightarrow Fr -( l^{-1})^{\gamma}= Fr -\dfrac{1}{l^{\gamma}},
	\end{align}
	since  by corollary 1.2.10 in \cite{DeHaan2006}  $U_F(x)\in RV_{\gamma}$ and $\vartheta_{N, N, T}/U_F(N)\Rightarrow Fr$ by corollary 1.2.4 in \cite{DeHaan2006}.
	
	For $\gamma=0$ and let $Gu$ be an RV  with $P(Gu\leq x) = \exp(-e^{-x})$ for  $x\in \R$; note that $-\log(E_1^*)\overset{d}{=}Gu$. By corollary 1.2.4 in \cite{DeHaan2006} there exists a $\hat{f}(\cdot)$ such that $\alpha_N  = \hat{f}(F^{-1}(1-1/N))$ for $\alpha_N$ of eq. \eqref{equation:noiselessFirstOrderAuxiliaryLimit}. For this $\hat{f}$ we have
	\begin{align*}
	& 	\dfrac{1}{\hat{f}\left(F^{-1}\left(  1 - \frac{1}{N} \right)\right)} \left( \vartheta_{N, N, T} - F^{-1}\left(  1 - \frac{l}{N} \right) \right)\\
	&  = 	\dfrac{1}{\alpha_N} \left( \vartheta_{N, N, T} - U_F(N)\right)  + 	\dfrac{1}{\alpha_N} \left(  U_F\left(\dfrac{N}{l} \right) - U_F(N) \right)\Rightarrow Gu - \log(l),
	\end{align*}
	where the first term converges by corollary 1.2.4 in \cite{DeHaan2006} and the second term by eq. \eqref{equation:noiselessFirstOrderAuxiliaryLimit}.

		For $\gamma<0$ necessarily $F^{-1}(1) = U_F(\infty)<\infty$. Let $W$ be an RV  with $P(W\leq x) = \exp( -(-x)^{-1/\gamma} )$ for  $x<0$ and $P(W\leq x)=1$ for $x\geq 0$, note that $-(E_1^*)^{-\gamma} \overset{d}{=}W$.
		 Then
	\begin{align} 
& 	\dfrac{1}{F^{-1}(1) - F^{-1}\left(1- \frac{1}{N} \right) }\left(\vartheta_{N, N, T} -  F^{-1}\left(  1- \frac{l}{N}   \right) \right)\\
 & =  	\dfrac{1}{U_F(\infty) - U_F(N) }\left(\vartheta_{N, N, T}- U_F(\infty) \right) + \dfrac{U_F(\infty) - U_F(N/l)}{U_F(\infty) - U_F(N)}	\Rightarrow  W +  \frac{1}{l^{\gamma}},
	\end{align} where the first term converges by corollary 1.2.4 in \cite{DeHaan2006} and  $U_F(\infty)-U_F(x)\in RV_\gamma$ by corollary 1.2.10 in \cite{DeHaan2006}.
\end{proof}

\begin{proof}[Proof of lemma \ref{theorem:jointEVT}]
	
	Let $\alpha_N$ be as in eq. \eqref{equation:noiselessFirstOrderAuxiliaryLimit} and let $\beta_N = U_F(N)$.
%
	First we show that  for the  constants $\alpha_N, \beta_N$ it holds that
	\begin{multline}\label{equation:jointEVTcanonical}
	\begin{pmatrix}
	\dfrac{\vartheta_{N, N, T}-\beta_N}{\alpha_N}, \dfrac{\vartheta_{N-1, N, T} -\beta_N}{\alpha_N}, \dots, \dfrac{\vartheta_{N-q, N, T} -\beta_N}{\alpha_N}
	\end{pmatrix} \\ \Rightarrow 
\gamma^{-1}\begin{pmatrix}
	 {(E_1^*)^{-\gamma}-1 },  {(E_1^*+E_2^*)^{-\gamma}-1}, \dots,  {(E_1^*+E_2^*+ \dots + E^*_{q+1})^{-\gamma} -1 }
\end{pmatrix}.
	\end{multline}	
	
	Let $E_1, \dots, E_{q+1}$ be IID standard exponential RVs, and $E_1^*, \dots, E_{q+1}^*$ another IID set of standard exponential RVs. Observe that  $
	P\left( U_T\left(1/(1-\exp({-E_i}))\right)\leq x \right) = H_T(x).$
	Then 
		\begin{align*} 
	& (\vartheta_{N, N, T}, \vartheta_{N-1, N, T}, \dots, \vartheta_{N-q, N,T}) \\ \overset{d}{=}  &  \left( U_T\left(\frac{1}{1-\exp({-E_{1, n}})} \right), U_T\left(\frac{1}{1-\exp{(-E_{2, n})}} \right),  \dots, U_T\left(\frac{1}{1-\exp({-E_{q+1, n}}) } \right)\right)\\
	\overset{d}{=} & \left( U_T\left(\dfrac{1}{1- \exp\left( {-\frac{E_1^*}{N}}\right)} \right) ,  U_T\left(\dfrac{1}{1- \exp\left({-\frac{E_1^*}{N} - \frac{E_2^*}{N-1} }\right)} \right) , \dots, \right.\\
	& \quad, \dots, \left.  U_T\left(\dfrac{1}{1- \exp\left({-\frac{E_1^*}{N} - \frac{E_2^*}{N-1} - \dots - \frac{E_{q+1}^*}{N-q} }\right)} \right) \right),
	\end{align*}
	where the second equality follows by the \cite{Renyi1953} representation  of order statistics from an exponential sample (see expression (1.9) in \cite{Renyi1953}). We conclude that
	\begin{align}%
	& \left(\dfrac{\vartheta_{N, N, T}-\beta_N}{\alpha_N}, 
	\dots, 
	\dfrac{\vartheta_{N-q, N}-\beta_N}{\alpha_N} \right) \\ 
	 \overset{d}{=}  & \begin{pmatrix}
\dfrac{U_T\left(\frac{1}{1- \exp\left( {-\frac{E_1^*}{N}}\right)} \right) -\beta_N}{\alpha_N} , 
	&   \dots, & \dfrac{U_T\left(\frac{1}{1- \exp\left({-\frac{E_1^*}{N} - \frac{E_2^*}{N-1} - \dots - \frac{E_{q+1}^*}{N-q} }\right)} \right)-\beta_N}{\alpha_N}   
	  \end{pmatrix} . \label{appendix:equation:jointEVTnoisyThetasExponentials}
	\end{align}
	Examine the first coordinate in the above vector:
	\begin{align}
& \dfrac{U_T\left(\frac{1}{1- \exp\left( {-\frac{E_1^*}{N}}\right)} \right) -\beta_N}{\alpha_N}  
   =  \dfrac{U_T\left(\frac{N}{N\left( 1- \exp\left( {-\frac{E_1^*}{N}}\right) \right)} \right) -\beta_N}{\alpha_N} \\
	& = \dfrac{U_F\left(\frac{N}{N\left( 1- \exp\left( {-\frac{E_1^*}{N}}\right) \right)} \right) -\beta_N}{\alpha_N} +  \dfrac{ U_T\left(\frac{N}{N\left( 1- \exp\left( {-\frac{E_1^*}{N}}\right) \right)} \right)- U_F\left(\frac{N}{N\left( 1- \exp\left( {-\frac{E_1^*}{N}}\right) \right)} \right) }{\alpha_N}.
	\end{align}
	We can rewrite each term in eq. \eqref{appendix:equation:jointEVTnoisyThetasExponentials} as above by decomposing  it into a $U_F$ a component and a difference term involving $U_T$ and $U_F$.
	
	 We separately analyze the two terms.
		First,  by theorem 2.1.1 in  \cite{DeHaan2006} \begin{multline}\label{appendix:equation:jointEVTUFconvergence}
\begin{pmatrix} \dfrac{U_F\left(\frac{1}{1- \exp\left( {-\frac{E_1^*}{N}}\right)} \right) -\beta_N}{\alpha_N} , 
& \dots, & \dfrac{U_F\left(\frac{1}{1- \exp\left({-\frac{E_1^*}{N} - \frac{E_2^*}{N-1} - \dots - \frac{E_{q+1}^*}{N-q} }\right)} \right)-\beta_N}{\alpha_N}   
\end{pmatrix} \\ \Rightarrow  
\gamma^{-1}\begin{pmatrix}
	{(E_1^*)^{-\gamma}-1 },  {(E_1^*+E_2^*)^{-\gamma}-1}, \dots,  {(E_1^*+E_2^*+ \dots + E^*_{q+1})^{-\gamma} -1 }
\end{pmatrix}.
		\end{multline}
		Second, the difference terms converge to zero. We show this for the first term only, the result follows analogously for the other terms.
		First, write the difference term as   
		\begin{align} 
&	\dfrac{ U_T\left(\frac{N}{N\left( 1- \exp\left( {-\frac{E_1^*}{N}}\right) \right)} \right)- U_F\left(\frac{N}{N\left( 1- \exp\left( {-\frac{E_1^*}{N}}\right) \right)} \right) }{\alpha_N} \label{appendix:equation:jointEVTdifferenceTerm}  \\
&	=\dfrac{ U_T\left(\frac{N}{N\left( 1- \exp\left( {-\frac{E_1^*}{N}}\right) \right)} \right)-U_T(N) }  {\alpha_N}   - \dfrac{U_F\left(\frac{N}{N\left( 1- \exp\left( {-\frac{E_1^*}{N}}\right) \right)} \right) - U_F(N)}{\alpha_N} + \dfrac{U_F(N) - U_T(N)}{\alpha_N} . 
		\end{align}
	
We show that above expression  is $o(1)$ in two steps.	First, we show that the difference of the first two terms in the above display tends to zero. 	Define $	\tilde{h}_{N, T}(x) = 	 ( U_T\left(Nx\right)-U_T(N) )/{\alpha_N}$. 
			$\tilde{h}_{N, T}$ converges pointwise to $\tilde{h}(x)= {(x^{\gamma}-1)}/{\gamma}$  as $N, T\to\infty $ by theorem \ref{theorem:evtNoisy}, lemma \ref{appendix:lemma:convergenceMaximumQuantiles},  and eq. \eqref{equation:noiselessFirstOrderAuxiliaryLimit} with $U_T$ in place of $U_F$.
			Since the limit is continuous, and $\tilde{h}_{N, T}(x)$ is monotonic in $x$, convergence is locally uniform in $y$ (see section 0.1 in \cite{Resnick1987}). Define $x_N=N(1-\exp({-E_1^*/N}))$. Then $x_N\to  E_1^*$ and  $x_N$ is     a bounded sequence. Then as $N, T\to\infty$ it holds that $\tilde{h}_{N, T}(x_N^{-1})\to \tilde{h}((E_1^*)^{-1})$ (observe that $E_1^*$ does not depend on $N$ or $T$). 
		Similarly, define $
		\tilde{f}_N(x) = ({U_F\left(Nx\right) - U_F(N)})/{\alpha_N}$. 
		$\tilde{f}_N(x)$ converges to the same limit $\tilde{h}(x)$ by  eq. \eqref{equation:noiselessFirstOrderAuxiliaryLimit}. As $\tilde{f}_N(x)$ is   monotonic,  convergence is also locally uniform in $x$, so analogously $\tilde{f}_N(x_N^{-1})\to \tilde{h}((E_1^*)^{-1})$. Thus the difference between the first two terms in eq. \eqref{appendix:equation:jointEVTdifferenceTerm} tends to zero. 
		Second,  ${(U_F(N) - U_T(N))}/{\alpha_N}=o(1)$  as  in the proof of theorem \ref{theorem:evtNoisy} (take $\tau=1$ and recall that conditions of theorem \ref{theorem:evtNoisy} are assumed to hold).

Combining the above argument with eq. \eqref{appendix:equation:jointEVTUFconvergence}, we obtain  eq. 	\eqref{equation:jointEVTcanonical}.

	Last, we  translate from the   constants $(\alpha_N, \beta_N)$ to the  constants of the  statement of the lemma. We only show this for the first coordinate in eq.  	\eqref{equation:jointEVTcanonical}, the argument for the other coordinates is identical, and we split the proof by sign of $\gamma$. 
	For $\gamma<0$
	\begin{align} 
	\dfrac{\vartheta_{N, N, T}-U_F(\infty)}{U_F(\infty)-U(N)} 
	= \dfrac{\vartheta_{N, N, T} - \beta_N}{\alpha_N} \dfrac{\alpha_N}{U_F(\infty)-U(N)}   -1 \Rightarrow -\dfrac{(E_1^*)^{-\gamma}-1}{\gamma}\times \gamma  -1. 
	\end{align}
 where the result follows from lemma  1.2.9 in \cite{DeHaan2006} .
	For $\gamma>0$
	\begin{align}
		\dfrac{\vartheta_{N, N, T}}{U_F(N)}  
		= \dfrac{\vartheta_{N, N, T} - \beta_N}{\alpha_N} \dfrac{\alpha_N}{U_F(N)} +1
		 \Rightarrow \dfrac{(E_1^*)^{-\gamma}-1}{\gamma}\times \gamma  +1 = (E_1^*)^{-\gamma}.
	\end{align}
	where the result follows from  lemma 1.2.9 in \cite{DeHaan2006}.
	For $\gamma=0$ the constants in the lemma statement are  in fact $\alpha_N, \beta_N$ (corollary 1.2.4 in \cite{DeHaan2006}) We only need to represent $((E_1^*)^{-\gamma}-1))/\gamma$ in the form given in the theorem statement. Observe that that for $x>0$ ${(x^{-\gamma}-1)}/{\gamma} \to -\log(x)$ as $\gamma\to 0$ (note the minus). The conclusion follows.
\end{proof}

 \begin{proof}[Proof of theorem \ref{theorem:feasibleEVT}]

Follows   immediately from lemmas \ref{corollary:evtNoisyCenteredQuantile} and  \ref{theorem:jointEVT} and the continuous mapping theorem.\end{proof}

 
 We now turn towards the proof of theorem \ref{theorem:subsampling}.
We begin by computing the order of the difference between the noisy and the noiseless maximum.		
Define $EA_{N, T}  = \max\curl*{    
T^{-p}\abs{\varepsilon_{1, T}}, \dots, T^{-p}\abs{\varepsilon_{N, T}}}$.
Since $\vartheta_{i, T}= \theta_i + T^{-p} \varepsilon_{i, T}$, the following elementary inequality holds:
\begin{equation}\label{appendix:equation:EAmaxInequality}
\theta_{N, N}- EA_N \leq \vartheta_{N, N, T}\leq \theta_{N, N}+ EA_N   \text{ or }  \abs{\vartheta_{N, N, T} -\theta_{N, N} }\leq EA_{N, T}.
\end{equation} 

\begin{lem}
	
	\begin{enumerate}[noitemsep,topsep=0pt,parsep=0pt,partopsep=0pt, label={(\arabic*)}, leftmargin=*]
		\item If  $\sup_{T}\E\abs{\varepsilon_{i, T}}^{\beta}<\infty$ for some  $\beta> 0$, then $\vartheta_{N, N, T} -\theta_{N, N}$ is $O_p ( {N^{1/\beta}}/{T^{p}}  )$. 
		\item Let assumption \ref{assumption:tightness} hold and   $\varepsilon_{i, T}\sim N(\mu_T, \sigma_T^2)$ for all $T$.  Then $\vartheta_{N, N, T}-\theta_{N, N}=O_p( \sqrt{\log (N)}/T^p)$.
	\end{enumerate}
	
	\label{lemma:stochasticOrderOfDifference}

\end{lem}

\begin{proof} Consider (1),
Let $t>0$. We compare $EA_{N, T}$ to  ${N^s}/{T^p}$ for $s\geq 0$:
	\begin{align*}
	P\left( \dfrac{EA_{N, T}}{N^s/T^p} \geq t\right)  & = P\left(\bigcup_{i=1}^N \curl*{\dfrac{1}{T^p}\dfrac{\abs{\varepsilon_{i, T}}}{N^s/T^p}\geq t   } \right) \leq \sum_{i=1}^NP\left( \dfrac{\abs{\varepsilon_{i, T}}}{N^s/T^p} \geq t T^p \right) \leq NP\left( \abs{\varepsilon_{i, T}} \geq t N^s \right)\\
	& \leq  N  \dfrac{ \E\abs{\varepsilon_{i, T}}^{\beta} }{t^{\beta} N^{\beta s}   }  \leq  \dfrac{\sup_T\E\abs{\varepsilon_{i, T}}^{\beta}}{t^{\beta}} N^{1-\beta s},
	\end{align*} 
where we use Markov's inequality in the penultimate   step. Setting $s=1/\beta$ shows that these probabilities are bounded, uniformly in $T$, hence 
	$EA_{N, T}= O_p\left(  {N^{1/\beta}}/{T^p} \right)$. The result follows by inequality \eqref{appendix:equation:EAmaxInequality}.

Consider (2).
	By assumption \ref{assumption:tightness}, $(\mu_T, \sigma_T^2)$ is a bounded sequence. $EA_{N, T}$  is a maximum of independent normal variables of bounded mean and variance. Then  $EA_{N, T}=O_p\left(\sqrt{\log (N)}/T^p\right)$. The result then follows by inequality \eqref{appendix:equation:EAmaxInequality}.   
\end{proof}

 \begin{proof}[Proof of theorem \ref{theorem:subsampling}]
 	
 	The proof changes  depending on whether $l=0$ or $l>0$. 
	We begin with $l=0$ and $\gamma<0$.
	Label 
	\begin{align*} 
	J(x)  & = P\left(  \dfrac{  (E_1^* + \dots + E_{r+1}^*)^{-\gamma} }{(E_1^*+ \dots + E_{q+1}^*)^{-\gamma}-(E_1^*)^{-\gamma} }\leq x \right),  	J_{N, T}(x) & = P\left(\dfrac{\vartheta_{N-r, N, T} -F^{-1}(1)}{\vartheta_{N-q, N, T} - \vartheta_{N, N, T}} \leq x  \right),
	\end{align*}
	using notation of theorem \ref{theorem:feasibleEVT}. Theorem \ref{theorem:feasibleEVT} shows that  $J_N \Rightarrow J$.

	Add and subtract $F^{-1}(1)$ in $L_{b, N, T}$ to obtain
	\begin{equation}
	L_{b, N, T}(x) = \dfrac{1}{\binom{N}{b}}\sum_{s=1}^{\binom{N}{b}}  \I\curl*{ \dfrac{ \vartheta_{b-r, b, T}^{(s)} - F^{-1}(1) }{\vartheta_{b-q, b, T}^{(s)}  - \vartheta_{b, b, T}^{(s)} }  + \dfrac{ F^{-1}(1) - \vartheta_{N, N, T} }{\vartheta_{b-q, b, T}^{(s)}  - \vartheta_{b, b, T}^{(s)} } \leq x   }.
	\end{equation}
	Fix  an arbitrary $\epsilon>0$ and define the event $	E_{N, T} = \curl*{ \abs*{  { ( F^{-1}(1) - \vartheta_{N, N, T}) }/{(\vartheta_{b-q, b, T}^{(s)}  - \vartheta_{b, b, T}^{(s)}  )} }\leq  \epsilon   }.$
	The goal is to show that $P(E_{N, T})\to 1$ for any $\varepsilon>0$ as $N, T\to\infty$. 
	Write
	\begin{align} 
	\dfrac{ F^{-1}(1) - \vartheta_{N, N, T} }{\vartheta_{b-k, b, T}^{(s)}  - \vartheta_{b, b, T}^{(s)} }  & =  \dfrac{ F^{-1}(1) - \theta_{N, N} }{\vartheta_{b-q, b, T}^{(s)}  - \vartheta_{b, b, T}^{(s)}  }+  \dfrac{ \theta_{N, N} - \vartheta_{N, N, T} }{  \vartheta_{b-q, b, T}^{(s)}  - \vartheta_{b, b, T}^{(s)} }. \label{equation:subsamplingProofTail}
	\end{align}
	We show that both terms are $o_p(1)$, which allows us to conclude that $P(E_{N, T})\to 1$.
	Focus on the first   term in eq. \eqref{equation:subsamplingProofTail}. Recall that $F^{-1}(1)= U_F(\infty)$ and  write the term as 
	\begin{equation} 
	\dfrac{ F^{-1}(1) - \theta_{N, N} }{\vartheta_{b-q, b, T}^{(s)}  - \vartheta_{b, b, T}^{(s)}  }= \dfrac{U_F(\infty)- \theta_{N, N} }{ U_F(\infty)-U_F(N) } \dfrac{U_F(\infty)-U_F(b) }{ \vartheta_{b-q, b, T}^{(s)}  - \vartheta_{b, b, T}^{(s)}  }  \dfrac{ U_F(\infty) - U_F(N)}{U_F(\infty)-U_F(b)}.
	\end{equation}
	\begin{enumerate}[noitemsep,topsep=0pt,parsep=0pt,partopsep=0pt, label={(\arabic*)}, leftmargin=*]
		\item The first term is $O_p(1)$ under assumption \ref{assumption:EV} by corollary 1.2.4 in \cite{DeHaan2006}. This term does not depend on $T$.
		
		\item Second term is $O_p(1)$ by lemma \ref{theorem:jointEVT}. Lemma \ref{theorem:jointEVT} applies  to subsample $s$, since $(N, T)$ satisfy conditions of proposition \ref{proposition:ratesExtremeHalfUniformity}, and $b=o(N)$. 
		
		\item  Last,  $(U_F(\infty)-U_F(t))\in RV_{\gamma}$ by corollary 1.2.10 in \cite{DeHaan2006}.  By proposition 0.5 in \cite{Resnick1987} $ {(U_F(\infty)- U_F(xt))}/{(U_F(\infty)- U_F(t))}\to x^{\gamma}$ uniformly on intervals of the form $(b, \infty)$. Hence, using
		$b= N^m$, $m<1, \gamma<0$, we obtain
		\begin{equation} 
		\frac{U_F(\infty)-U_F(N)}{U_F(\infty)-U_F(N^m)} =   \frac{U_F(\infty)-U_F((N^{1-m}N^m))}{U_F(\infty)-U_F(N^m)}  \sim (N^{1-m})^{\gamma} \to 0.
		\end{equation}
 The last term is $o(1)$.
	\end{enumerate} 
	Overall the first  term  in eq. \eqref{equation:subsamplingProofTail} is   $o_p(1)$

	Now focus on the second term in eq. \eqref{equation:subsamplingProofTail}.   Let condition (1) of proposition \ref{proposition:ratesExtremeHalfUniformity} hold (the proof is analogous if condition (2) holds instead). 
	Since $\sup_T \E\abs{\varepsilon_{i, T} }^{\beta}<\infty$, by lemma \ref{lemma:stochasticOrderOfDifference} we conclude that	
	 $\theta_{N, N}-\vartheta_{N, N, T} = O_p({N^{1/\beta}}/{T^{p}})$. 
	  By lemma \ref{theorem:jointEVT}, ${(\vartheta_{b-q, b, T}^{(s)}   - \vartheta_{b, b, T}^{(s)})  }/{(U_F(\infty) - U_F(b))}$ is $O_p(1)$.  In addition,  by corollary 1.2.10 in \cite{DeHaan2006} ${1}/{(U_F(\infty)-U_F(t))}$ is $RV_{-\gamma}$, so by lemma \ref{appendix:lemma:karamataCharacterization} we can write ${1}/{(U_F(\infty)-U_F(t))} = t^{-\gamma} L(t)$ for some slowly varying $L$.   Since  $b=N^m$, we obtain
	\begin{align}
	\dfrac{ \theta_{N, N} - \vartheta_{N, N, T} }{  \vartheta_{b-q, b, T}^{(s)}   - \vartheta_{b, b, T}^{(s)} }  &  = \left( \theta_{N, N} - \vartheta_{N, N, T}\right) 	\dfrac{ U_F(\infty) - U_F(b)}{  \vartheta_{b-q, b, T}^{(s)}   - \vartheta_{b, b, T}^{(s)}}  \dfrac{1}{U_F(\infty)- U_F(b)} \\
	& =  O_p\left(\frac{N^{1/\beta}}{T^{p}}\right) O_p\left(   N^{-\gamma m} L(N^m) \right)  . \label{equation:subsamplingProofTailSecondTerm}
	\end{align}  
	 $L(N^m)$ diverges at rate slower than any power of $N^m$. Then the expression in  \eqref{equation:subsamplingProofTailSecondTerm} is $o_p(1)$ if   for some   $\kappa>0$ it holds that 
	$N^{1/\beta -\gamma m+\kappa}{T^{-p}} \to 0$. 
	However, such  a $\kappa>0$ exists  since assumptions of proposition \ref{proposition:ratesExtremeHalfUniformity} hold and $\gamma<\gamma m$.
	

	The remainder of the proof now proceeds as in \cite{Politis1994}.
	Define  
	\begin{equation} 
	\tilde{L}_{b, N, T} = \dfrac{1}{\binom{N}{b}}\sum_{s=1}^{\binom{N}{b}}  \I\curl*{ \dfrac{ \vartheta_{b, b, T}^{(s)} - U_F(\infty) }{\vartheta_{b-q, b, T}^{(s)} - \vartheta_{b, b, T}^{(s)}  } \leq x   }.
	\end{equation}
	On the event $E_{N, T}$ it holds that $
	\tilde{L}_{b, N, T}(x-\epsilon) \leq L_{b, N, T}(x)\leq \tilde{L}_{b, N, T}(x+\epsilon)$. 
	Since $P(E_{N, T})\to 1$, the above also holds with probability approaching one.	Observe that $
	\E(\tilde{L}_{n, b}(x)) = J_{b, T}(x) \Rightarrow J(x)$. 
	$\tilde{L}_{b, N, T}$ is a U-statistic of order $b$ with kernel bounded between 0 and 1. By theorem A on p. 201 in \cite{Serfling1980} it holds that $\tilde{L}_{b, N, T}(x) - J_{b, T}(x)\xrightarrow{p} 0$.
	 Then, as in \cite{Politis1994}, for any $\epsilon>0$ with probability approaching it holds that   $
	J(x-\epsilon)-\epsilon\leq L_{b, N, T}\leq J(x+\epsilon) +\epsilon$. 
	Letting $\epsilon\to 0$  shows that $L_{b, N, T}(x)\to J(x)$ at all continuity points $x$ of $J(x)$. 
	This also shows that $\hat{c}_{\alpha} = L^{-1}_{b, N, T}(\alpha)\to J^{-1}(\alpha)=c_{\alpha}$ since weak convergence of CDFs  is equivalent to weak convergence of quantiles.

	\air 
	
		Now consider the case of $l>0$. Add and subtract $U_F(b/l)$ in the subsampling estimator:
		\begin{equation} 
	L_{b, N, T}(x) = \dfrac{1}{\binom{N}{b}}\sum_{s=1}^{\binom{N}{b}}  \I\curl*{ \dfrac{ \vartheta_{b-r, b, T}^{(s)} - U_F(b/l) }{\vartheta_{b-q, b, T}^{(s)}  - \vartheta_{b, b, T}^{(s)} }  + \dfrac{ U_F(b/l) -\vartheta_{N- {Nl}/{b}, N, T} }{\vartheta_{b-q, b, T}^{(s)}  - \vartheta_{b, b, T}^{(s)} } \leq x   }.
	\end{equation}
		First, since $b$ satisfies the conditions of theorem \ref{theorem:feasibleEVT}
	\begin{equation}
	\dfrac{ \vartheta_{b-r, b, T}^{(s)} - U_F(b/l) }{\vartheta_{b-q, b, T}^{(s)}  - \vartheta_{b, b, T}^{(s)} }\Rightarrow    \dfrac{  (E_1^*+ \dots + E_{r+1}^*)^{-\gamma} +   l^{-\gamma}}{(E_1^*+ \dots + E_{q+1}^*)^{-\gamma}-(E_1^*)^{-\gamma} }. 
	\end{equation}
	Similarly to the above, 
	fix some $\epsilon>0$ and define the event
	\begin{equation} \label{appendix:equation:subsamplingEntNotMaximum}
	E_{N, T} = \curl*{ \abs*{ \dfrac{ U_F(b/l) - \vartheta_{N- {Nl}/{b}, N, T}}{\vartheta_{b-q, b, T}^{(s)}  - \vartheta_{b, b, T}^{(s)} } }\leq  \epsilon   }.
	\end{equation}
	The goal is is to show that $P(E_{N, T})\to 1$ for any $\epsilon>0$ under the assumptions of the theorem, the proof will then proceed as above.
  We show this separately for different signs of $\gamma$. 
	Let $\gamma<0$ and write the expression under the absolute value as
	\begin{align} 
	& \label{appendix:equation:subsamplingExtraTermLPositive}    \left[ \sqrt{  \frac{Nl}{b} } \frac{U_F(b/l) -\vartheta_{N- {Nl}/{b}, N, T}    }{\frac{b}{l}U'_F(\frac{b}{l})} \right] \left[ \frac{U_F(\infty)- U_F(b)}{ \vartheta_{b-q, b, T}^{(s)}  - \vartheta_{b, b, T}^{(s)}} \right] \left[\frac{  \frac{b}{l}U'_F\left(\frac{b}{l}\right)    }{\sqrt{\frac{Nl}{b}} (U_F(\infty)-U_F(b))}\right].
	\end{align}
	The first term is $O_p(1)$ by theorem \ref{theorem:intermediateNormality} taken with $k=Nl/b$.  Conditions of theorem \ref{theorem:intermediateNormality} hold, since  $k={Nl}/{b}\sim N^{1-m}$ and conditions of proposition \ref{proposition:ratesIntermediateHalfUniformity} are assumed to hold for $\delta=1-m$.  The second term is $O_p(1)$ by lemma    \ref{theorem:jointEVT} as the conditions of proposition \ref{proposition:ratesExtremeHalfUniformity} hold for $N$ (and hence for $b$).  Finally,  by corollary 1.1.14 in \cite{DeHaan2006} under assumption \ref{assumption:vonMises}  ${  ({b}/{l})U'_F\left(\frac{b}{l}\right)    }/{  (U_F(\infty)-U_F(b))}\to -\gamma$. Multiplying this by $\left(Nl/b \right)^{-1/2}\to0$ shows that overall the last term is $o(1)$.
%
%
We conclude that  overall the expression in eq. \eqref{appendix:equation:subsamplingExtraTermLPositive} is $o_p(1)$.

	For $\gamma>0$ instead write the expression under the absolute value in eq.  \eqref{appendix:equation:subsamplingEntNotMaximum} as 
	\begin{equation} 
	\dfrac{  U_F(b/l) -\vartheta_{N-\frac{Nl}{b}, N, T} }{\vartheta_{b-q, b, T}^{(s)}  - \vartheta_{b, b, T}^{(s)} }  =  \left[\sqrt{ \frac{Nl}{b} } \frac{U_F(b/l) -\vartheta_{N-\frac{Nl}{b}, N, T}    }{\frac{b}{l}U'_F(\frac{b}{l})} \right]\left[  \frac{ U_F(b)}{ \vartheta_{b-q, b, T}^{(s)}  - \vartheta_{b, b, T}^{(s)}}  \right] \left[\frac{  \frac{b}{l}U'_F\left(\frac{b}{l}\right)    }{\sqrt{\frac{Nl}{b}} U_F(b)}\right].
	\end{equation}
	The last term is $o(1)$ by corollary 1.1.12 in \cite{DeHaan2006},   other terms are as above.
	
	For $\gamma=0$   write the term of interest as 
	\begin{equation}
	\dfrac{  U_F(b/l) -\vartheta_{N-\frac{Nl}{b}, N, T} }{\vartheta_{b-q, b, T}^{(s)}  - \vartheta_{b, b, T}^{(s)} }  =  \left[ \sqrt{ \frac{Nl}{b} } \frac{U_F(b/l) -\vartheta_{N-\frac{Nl}{b}, N, T}    }{\frac{b}{l}U'_F(\frac{b}{l})}  \right]\left[ \frac{ \frac{b}{l}U'_F\left(\frac{b}{l}\right)   }{ \vartheta_{b-q, b, T}^{(s)}  - \vartheta_{b, b, T}^{(s)}} \right]\left[ \frac{  1 }{\sqrt{\frac{Nl}{b}} }\right].
	\end{equation}
	The first term is $O_p(1)$ as above. The second term is $O_p(1)$ by lemma   \ref{theorem:jointEVT} as by corollaries 1.1.10 and 1.2.4 in \cite{DeHaan2006} under assumption \ref{assumption:vonMises}  we may  take $\hat{f}\left(U_F(N)\right) = NU'_F(N)$.  Finally, the last term is $o(1)$.\end{proof}



 
 \begin{proof}[Proof of remark \ref{remark:gammaEstimation}]
 	\label{appendix:page:proof-convergence-gamma-hat}
 	
 Let $\gamma>0$ and	consider $\hat{\gamma}_H$.  Under conditions of theorems \ref{theorem:evtNoisy}  eq. \eqref{equation:noiselessFirstOrderAuxiliaryLimit} holds with $U_T$ in place of $U_F$. By lemma 1.2.10 in \cite{DeHaan2006}, this convergence is equivalent to $U_T(tx)/U_T(x)\to x^{\gamma}$ for all $x>0$.  Now proof of theorem 3.2.2 in \cite{DeHaan2006} applies with $U_T$ in place of $U_F$ ($U$ in their notation).
 
 Consistency of $\hat{\gamma}_{PWM}$ for $\gamma<1$ holds by theorem 3.6.1 in \cite{DeHaan2006} as eq.  \eqref{equation:noiselessFirstOrderAuxiliaryLimit} holds with $U_T$ in place of $U_F$ under the conditions of theorems \ref{theorem:evtNoisy} and \ref{theorem:intermediateNormality}.
  \end{proof}
 
\subsection{Proof of Theorem \ref{theorem:feasibleIVTnoiseless}}

We begin by establishing several supporting lemmas.
\begin{lem}
	
\label{lemma:uniformSameLimitPair}  Let $U_{1, N}\leq \dots \leq U_{N, N}$ be the order statistics from an IID sample of size $N$ from a Uniform$[0, 1]$ distribution.  If $k=o(N), s=\floor{\sqrt{k}}$, and $k\to \infty$, then
	\begin{align}
	\begin{pmatrix}
	\sqrt{k}\left(\frac{N}{k}U_{k+1, N}-1 \right) \\
	\sqrt{k}\left(\frac{N}{k+s}U_{k+s+1, N}-1 \right) 
	\end{pmatrix}	\Rightarrow  N\left(0, \begin{pmatrix}
	1 & 1\\
	1 & 1
	\end{pmatrix}  \right).
	\end{align}  

\end{lem}

\begin{proof}
	By lemma 2.2.3 in \cite{DeHaan2006} $
	\sqrt{k}\left( \frac{N}{k}U_{k+1, N}-1 \right)\Rightarrow N(0, 1) \equiv Z.$
	To show the result, we only need to show that the suitable scaled difference between
	$U_{k+1, N}$ and $U_{k+s-1, N}$  converges to zero in probability.   Consider
	\begin{align} 
	& 	\sqrt{k}\left(\dfrac{N}{k+s}U_{k+s+1, N}-1 \right) - 	\sqrt{k}\left(\dfrac{N}{k}U_{k+1, N}-1 \right)\\ & =  \sqrt{k}N\left(  \dfrac{k}{k} \dfrac{1}{k+s} U_{k+s+1, N} - \dfrac{1}{k} U_{k+1, N}     \right) 	\\
	&  = 	 \dfrac{N}{\sqrt{k}}\left(   U_{k+s+1}-  U_{k+1} - \dfrac{s}{N+1} - \dfrac{s}{k+s}U_{k+s+1, N} + \dfrac{s}{N+1}   \right)\\
&
  = 	 \dfrac{N}{\sqrt{k}}\left(   U_{k+s+1}-  U_{k+1} - \dfrac{s}{N+1} \right) -\dfrac{N}{\sqrt{k}}\left( \dfrac{s}{k+s}U_{k+s+1, N} - \dfrac{s}{k+s} \dfrac{k+s+1}{N+1} \right)\\
 & \quad  + \dfrac{N}{\sqrt{k}}\left(\dfrac{s}{N+1} - \dfrac{s}{k+s}\dfrac{k+s+1}{N+1} \right). \label{appendix:equation:uniformSameLimitPairDecomposition}
	\end{align}
We show that each of  the  terms in the last equality in eq. \eqref{appendix:equation:uniformSameLimitPairDecomposition}   is $o_p(1)$.  The last term:
	\begin{equation}
	\dfrac{N}{\sqrt{k}} \dfrac{s}{N+1}\left(1- \dfrac{k+s+1}{k+s} \right) \sim \dfrac{s}{\sqrt{k}}\left( 1- \dfrac{k+s+1}{k+s} \right) \to 0 \text{ as }s\sim\sqrt{k}, k\to\infty.
	\end{equation}
Consider the first term. A difference of order statistics from the uniform distribution follows a beta distribution: if $p>r$, then $U_{p, N}-U_{r, N} \sim $ Beta($p-r, N-p+r+1)$.  
Let $\delta>0$, then
	\begin{align}
& 	P\left( \abs*{ \dfrac{N}{\sqrt{k}}\left( U_{k+s+1, N}-  U_{k+1, N}  - \dfrac{s}{N+1} \right)} \geq \delta  \right) \\ 
	& = P\left( \abs*{\text{Beta}(s, N-s+1) - \E\left(\text{Beta}(s, N-s+1) \right) }\geq  \dfrac{\sqrt{k}}{N}\delta  \right)\\
	 &\leq  \dfrac{ \var(\text{Beta}(s, N-s+1) )   }{\delta^2 {k}/{N^2}}
	=  \dfrac{  \frac{ s(N-s+1)}{(N+1)^2 (N+2)}  }{\delta^2 \frac{k}{N^2}}    
\sim  \dfrac{\frac{s}{N^2}}{\delta^2 \frac{k}{N^2}} 
	=	 \frac{s}{\delta^2 k} \to 0.
	\end{align}
Last, turn to the second term. Since $U_{k+s+1, N}\sim $Beta($k+s+1, n-k-s$), for $\delta>0$ we have 
	\begin{align} 
& 	P\left( \dfrac{N}{\sqrt{k}}\abs*{\dfrac{s}{k+s}U_{k+s+1, N} -\dfrac{s}{k+s}\dfrac{k+s+1}{N+1}}\geq \delta    \right) \\
	& = P\left(\abs*{ \text{Beta}(k+s+1, N-k-s) - \E\left(\text{Beta}(k+s+1, N-k-s) \right) }\geq \delta \dfrac{\sqrt{k}(k+s)  }{Ns}   \right)\\
	& \leq \var(\text{Beta}(k+s+1, N-k-s) )\dfrac{ N^2s^2}{ \delta^2 k(k+s)^2  }   =  \dfrac{(k+s+1)(N-k-s) }{(N+1)^2(N+2)}\dfrac{ N^2s^2}{ \delta^2 k(k+s)^2  } 
	\\
	& \sim  \dfrac{(k+s+1)s^2 }{k(k+s)^2}   \sim \dfrac{1}{k+s} \to 0. 
	\end{align}The assertion of the lemma now follows.
\end{proof}

\begin{lem}
	\label{lemma:intermediatePairNoiseless} 
Let $\theta$ be sampled IID from $F$, let assumption	 \ref{assumption:vonMises} hold.  Let $k=o(N), k\to\infty, s=\floor{\sqrt{k}}$.
  Let $\theta_{N-k, N}, \theta_{N-k-s, N}$ be the order statistics from $F$.  Then as $N\to\infty$
	\begin{equation} 
	\sqrt{k}	\begin{pmatrix}
	\dfrac{ \theta_{N-k, N}- U_F\left(\frac{N}{k} \right) }{\frac{N}{k}U_F'\left(\frac{N}{k} \right)}
	\\
	\dfrac{ \theta_{N-k-s, N}- U_F\left(\frac{N}{k+s} \right) }{\frac{N}{k+s}U_F'\left(\frac{N}{k+s} \right)}
	\end{pmatrix} \Rightarrow N\left(0, \begin{pmatrix}
	1 & 1\\
	1 & 1
	\end{pmatrix}  \right).
	\end{equation}

\end{lem}

\begin{proof}
	
	We use the Cramer-Wold device together with a technique used by \cite{DeHaan2006} in proving an asymptotic normality result for a single statistic  under a von Mises condition (see proof of theorem 2.2.1 therein).
	Observe that  (in notation of lemma \ref{lemma:uniformSameLimitPair}) $
	(\theta_{N-k, N}, \theta_{N-k-s, N})\overset{d}{=} \left(U_F\left(1/{U_{k+1, N}} \right), U_F\left( {1}/{U_{k+s+1, N}} \right) \right).$
	Let $(c_1, c_2)\in \R^2$ and examine
	\begin{align} 
	& c_1\sqrt{k}\dfrac{ \theta_{N-k, N}- U_F\left(\frac{N}{k} \right) }{\frac{N}{k}U_F'\left(\frac{N}{k} \right)} + 	 c_2\sqrt{k}\dfrac{ \theta_{N-k-s, N}- U_F\left(\frac{N}{k+s} \right) }{\frac{N}{k+s}U_F'\left(\frac{N}{k+s} \right)}\\ \overset{d}{=} & 	 c_1\sqrt{k}\dfrac{ U_F\left(\frac{N}{k} \frac{k}{N U_{k+1, N}} \right)- U_F\left(\frac{N}{k} \right) }{\frac{N}{k}U_F'\left(\frac{N}{k} \right)} + 	 c_2\sqrt{k}\dfrac{ U_F\left(\frac{N}{k+s} \frac{k+s}{N U_{k+s+1, N}} \right)- U_F\left(\frac{N}{k+s} \right) }{\frac{N}{k+s}U_F'\left(\frac{N}{k+s} \right)}\\
	=  & c_1\sqrt{k}\int_1^{k/(NU_{k+1 N})} \dfrac{U_F'\left( \frac{N}{k}t \right)}{U'_F\left(\frac{N}{k} \right)}dt +  c_2\sqrt{k}\int_1^{(k+s)/(NU_{k+s+1, N})} \dfrac{U_F'\left( \frac{N}{k+s}t \right)}{U'_F\left(\frac{N}{k+s} \right)}dt. \label{appendix:equation:feasibleIVTintegralLemma}
	\end{align}
	
	Under assumption \ref{assumption:vonMises} $U'_F\in RV_{\gamma-1}$ by corollary 1.1.10 in \cite{DeHaan2006} (up to sign). Then  	by Potter's inequalities (proposition B.1.9 (5) in \cite{DeHaan2006})  for any $\varepsilon, \varepsilon'>0$ starting from some $N_0$ for $t\geq 1$ it holds that
	\begin{align} 
	& (1-\varepsilon) t^{\gamma-1-\varepsilon'} < \dfrac{U_F'\left(\frac{N}{k} t\right)}{U'_F\left(\frac{N}{k} \right)}< (1+\varepsilon)t^{\gamma-1+\varepsilon'},\\
	& (1-\varepsilon) t^{\gamma-1-\varepsilon'} < \dfrac{U_F'\left(\frac{N}{k+s} t\right)}{U'_F\left(\frac{N}{k+s} \right)}< (1+\varepsilon)t^{\gamma-1+\varepsilon'}.
	\end{align} 
Multiplying by $\sqrt{k}$ and taking integrals with limits of integration as in in eq. \eqref{appendix:equation:feasibleIVTintegralLemma},  we obtain
	for $c_1, c_2\geq 0$
	\begin{multline}
	c_1 (1-\varepsilon)\sqrt{k} \dfrac{  \left(\frac{k}{NU_{k+1, N}} \right)^{\gamma-\varepsilon'}-1}{\gamma-\varepsilon'} + c_2 	 (1-\varepsilon)\sqrt{k} \dfrac{  \left(\frac{k+s}{NU_{k+s+1, N}} \right)^{\gamma-\varepsilon'}-1}{\gamma-\varepsilon'} \\
	\leq c_1  \sqrt{k}\dfrac{U_F\left(\frac{1}{U_{k+1, N}} \right)- U_F\left(\frac{N}{k} \right) }{\frac{N}{k}U_F'\left(\frac{N}{k} \right)}+ c_2\sqrt{k}\dfrac{U_F\left(\frac{1}{U_{k+s+1, N}} \right)- U_F\left(\frac{N}{k+s} \right) }{\frac{N}{k}U_F'\left(\frac{N}{k+s} \right)} \\
	\leq c_1 (1+\varepsilon)\sqrt{k} \dfrac{  \left(\frac{k}{NU_{k+1, N}} \right)^{\gamma+\varepsilon'}-1}{\gamma+\varepsilon'}+ c_2 (1+\varepsilon)\sqrt{k} \dfrac{  \left(\frac{k+s}{NU_{k+s+1, N}} \right)^{\gamma+\varepsilon'}-1}{\gamma+\varepsilon'}.
	\end{multline}
	Similar inequalities apply for different combinations of signs of $c_1, c_2$, though with $(1-\varepsilon)$ replaced by $(1+\varepsilon)$ for the terms with $c_i<0$.
	 By lemma  \ref{lemma:uniformSameLimitPair} and  the delta method 
	 we get that for any $\varepsilon'>0$ that satisfies $\varepsilon'\neq \pm\gamma$ 
	\begin{equation}
	\sqrt{k} \begin{pmatrix}
	\frac{  \left(\frac{k}{NU_{k+1, N}} \right)^{\gamma+\varepsilon'}-1}{\gamma+\varepsilon'}\\
	\frac{  \left(\frac{k}{NU_{k+s+1, N}} \right)^{\gamma+\varepsilon'}-1}{\gamma+\varepsilon'}
	\end{pmatrix} \Rightarrow N\left(0, \begin{pmatrix}
	1 & 1\\
	1 & 1
	\end{pmatrix}  \right), \quad \sqrt{k} \begin{pmatrix}
	\frac{  \left(\frac{k}{NU_{k+1, N}} \right)^{\gamma-\varepsilon'}-1}{\gamma-\varepsilon'}\\
	\frac{  \left(\frac{k}{NU_{k+s+1, N}} \right)^{\gamma-\varepsilon'}-1}{\gamma-\varepsilon'}
	\end{pmatrix} \Rightarrow N\left(0, \begin{pmatrix}
	1 & 1\\
	1 & 1
	\end{pmatrix}  \right).
	\end{equation}
	with convergence being joint for the two vectors.
	Since $\varepsilon>0$ is arbitrary, we obtain that  
\begin{align} 
	c_1  \sqrt{k}\frac{ \theta_{N-k, N}- U_F\left(\frac{N}{k} \right) }{\frac{N}{k}U_F'\left(\frac{N}{k} \right)}+ c_2\sqrt{k}\frac{ \theta_{N-k-s, N}- U_F\left(\frac{N}{k+s} \right) }{\frac{N}{k}U_F'\left(\frac{N}{k+s} \right)}\end{align}
	has the same asymptotic distribution as \begin{align}
	c_1  \sqrt{k} \dfrac{  \left(\frac{k}{NU_{k+1, N}} \right)^{\gamma\pm \varepsilon'}-1}{\gamma\pm\varepsilon'}+ c_2 \sqrt{k} \dfrac{  \left(\frac{k+s}{NU_{k+s+1, N}} \right)^{\gamma\pm\varepsilon'}-1}{\gamma\pm\varepsilon'}.\end{align}
	Finally, the conclusion of the lemma follows by the Cramer-Wold device,
\end{proof}

\begin{proof}[Proof of theorem \ref{theorem:feasibleIVTnoiseless}]
	
	First, we focus on part (1) and establish the result for  $\theta$.
Since $F$ is differentiable, so is $U_F$, and  $
	U'_F(t)  = [{(1-F(U_F(t)))^2}]/[{f(U_F(t))}]$.
	 Since $F(U_F(t))$ is monotonic, the monotonicity assumption on $f$ implies that eventually also $U'_F$ is non-increasing/non-decreasing.
	Let $s=\floor{\sqrt{k}}$.
	Recall $F^{-1}(1-k/N) = U_F(N/k)$, then
	\begin{align} 
 	& \dfrac{\theta_{N-k, N}- F^{-1}\left(1- \frac{k}{N} \right)}{\theta_{N-k, N}-\theta_{N-k-s, N}}
 		  =    \dfrac{\frac{\sqrt{k}}{\frac{N}{k}U'_F(N/k)} \left(  \theta_{N-k, N}- U_F(N/k) \right) }{\frac{\sqrt{k}}{\frac{N}{k}U'_F(N/k)} \left(\theta_{N-k, N}-\theta_{N-k-s, N}  \right) }. \label{appendix:equation:feasibleIVTnoiselessSelfNormalizedDecomposition}
	\end{align}
	By lemma \ref{lemma:intermediatePairNoiseless} the numerator weakly converges to $Z\equiv N(0, 1)$
	
	We show that the denominator converges to 1 in probability.
	 Rewrite the denominator as 
	\begin{align} 
\label{appendix:equation:feasibleIVTnoiselessDenominatorDecomposition}	& \frac{\sqrt{k}}{\frac{N}{k}U'_F(N/k)} \left(\theta_{N-k, N}-\theta_{N-k-s, N}  \right) \\
	= &  \frac{\sqrt{k}}{\frac{N}{k}U'_F(N/k)} \left(\theta_{N-k, N} -U_F\left(\dfrac{N}{k} \right) \right)  - \frac{\sqrt{k}}{\frac{N}{k}U'_F(N/k)} \left(\theta_{N-k-s, N} - U_F\left(\dfrac{N}{k+s} \right) \right)  \\ &  + \frac{\sqrt{k}}{\frac{N}{k}U'_F(N/k)} \left(U_F\left(\dfrac{N}{k} \right) - U_F\left(\dfrac{N}{k+s} \right) \right).
	\end{align} 
	The second term can be written as
	\begin{align}
	&    \frac{\sqrt{k}}{\frac{N}{k+s}U'_F(N/(k+s))} \left(\theta_{N-k-s, N} - U_F\left(\dfrac{N}{k+s} \right) \right)  \dfrac{\frac{N}{k+s}U'_F(N/(k+s))}{\frac{N}{k}U'_F(N/k)}.	
	\end{align}
	By assumption \ref{assumption:vonMises} and corollary 1.1.10 in \cite{DeHaan2006} 
	 ${U'_F(tx)}/{U'_F}(t) \to x^{\gamma-1}$ as $t\to\infty$ locally uniformly in $(0, \infty)$. Hence 
	\begin{equation}
	\dfrac{\frac{N}{k+s}U'_F(N/(k+s))}{\frac{N}{k}U'_F(N/k)} = \dfrac{k}{k+s} \dfrac{U'_F\left(\frac{N}{k}\frac{k}{k+s}\right)}{U'_F\left(  \frac{N}{k}\right)}  \to 1.
	\end{equation}
	since ${k}/{k+s} \to 1$. 
	Thus, by lemma \ref{lemma:intermediatePairNoiseless}
	\begin{equation} 
	\frac{\sqrt{k}}{\frac{N}{k}U'_F(N/k)} \left(\theta_{N-k, N} -U_F\left(\dfrac{N}{k} \right) \right)  - \frac{\sqrt{k}}{\frac{N}{k}U'_F(N/k)} \left(\theta_{N-k-s, N} - U_F\left(\dfrac{N}{k+s} \right) \right)\xrightarrow{p} 0.
	\end{equation}
	Last, we examine the residual. Observe that $U_F'\geq 0$. First  suppose $U'_F$ is eventually non-increasing, in which case
	\begin{align}
	\left(U_F\left(\dfrac{N}{k} \right)  - U_F\left(\dfrac{N}{k+s} \right) \right) & = \int_{(N/k)\times k/(k+s)}^{N/k} U'_F\left(t \right)dt \leq \dfrac{s}{k+s} \frac{N}{k} U'_F\left(\dfrac{N}{k} \dfrac{k}{k+s} \right).
	\end{align}
	Using the above expression,   we obtain an upper bound for the residual term
	\begin{align} 
	\frac{\sqrt{k}}{\frac{N}{k}U'_F\left(\frac{N}{k} \right)} \left(U_F\left(\dfrac{N}{k} \right) - U_F\left(\dfrac{N}{k+s} \right) \right)  & \leq \frac{\sqrt{k}}{\frac{N}{k}U'_F\left(\frac{N}{k} \right)} \left( \dfrac{s}{k+s} \frac{N}{k} U'_F\left(\dfrac{N}{k} \dfrac{k}{k+s} \right)   \right)
	\\& =    \dfrac{\sqrt{k}s}{k+s} \frac{U'_F\left(\frac{N}{k} \frac{k}{k+	s} \right)   }{U'_F\left(\frac{N}{k} \right)} \to 1
	\end{align}
	since $s= \floor{\sqrt{k}}$ and by local uniform convergence of the ratio of $U'_F$. 
	At the same time, since $U_F'$ is eventually non-increasing, we obtain a lower bound
	\begin{equation} 
	\int_{(N/k)\times k/(k+s)}^{N/k} U'_F\left(t \right)dt \geq \dfrac{s}{k+s} \frac{N}{k} U'_F\left(\dfrac{N}{k}   \right)
	\end{equation}
which shows that	
	\begin{equation} 
	\frac{\sqrt{k}}{\frac{N}{k}U'_F\left(\frac{N}{k} \right)} \left(U_F\left(\dfrac{N}{k} \right) - U_F\left(\dfrac{N}{k+s} \right) \right) \geq  \dfrac{\sqrt{k}s}{k+s} \to 1.
	\end{equation}
	Hence, the residual term converges to 1.
	If instead $U'_F$ is eventually non-decreasing, swap the $N/(k+s)$ and $N/k$ terms. 
	
	By combining the above arguments and eq. \eqref{appendix:equation:feasibleIVTnoiselessDenominatorDecomposition}, we conclude that the denominator of eq.  \eqref{appendix:equation:feasibleIVTnoiselessSelfNormalizedDecomposition} converges to 1 i.p.
We conclude that
	\begin{equation} 
	\dfrac{\theta_{N-k, N}- U_F(N/k)}{\theta_{N-k, N}-\theta_{N-k-s, N}} \Rightarrow \dfrac{Z}{1}=Z\sim N(0, 1)
	\end{equation}

 Now turn to part (2) and consider the noisy estimates $\vartheta_i$.  If the conclusion of lemma  \ref{lemma:intermediatePairNoiseless}  holds with $\vartheta$ in place of $\theta$,  
	then  the proof of part (1)  applies with  order statistics of $\vartheta$ replacing their noiseless counterparts $\theta$. 
	In light of this, it is sufficient to establish the result of lemma \ref{lemma:intermediatePairNoiseless} for $\vartheta$.
	We proceed similarly to proof of theorem \ref{theorem:intermediateNormality} and we apply the Cramer-Wold device.

Observe that $
	(\vartheta_{N-k, N, T}, \vartheta_{N-k-s, N, T})\overset{d}{=} \left( U_T(Y_{N-k, N}),  U_T(Y_{N-k-s, N})\right) $
	for $Y$ as in the proof of theorem  \ref{theorem:intermediateNormality}.
	Let 	 $c_1, c_2\in \R$. Then  as in eq. \eqref{appendix:equation:intermediateDecomposition}
	\begin{align} 
	& c_1\sqrt{k}\dfrac{ \vartheta_{N-k, N, T}- U_F\left(\frac{N}{k} \right) }{\frac{N}{k}U_F'\left(\frac{N}{k} \right)} + 	 c_2\sqrt{k}\dfrac{ \vartheta_{N-k-s, N, T}- U_F\left(\frac{N}{k+s} \right) }{\frac{N}{k+s}U_F'\left(\frac{N}{k+s} \right)}\\ 
	\overset{d}{=} &  c_1\sqrt{k}\dfrac{ U_F\left(Y_{N-k, N} \right)- U_F\left(\frac{N}{k} \right) }{\frac{N}{k}U_F'\left(\frac{N}{k} \right)} + 	 c_2\sqrt{k}\dfrac{ U_F\left(Y_{N-k-s, N} \right)- U_F\left(\frac{N}{k+s} \right) }{\frac{N}{k+s}U_F'\left(\frac{N}{k+s} \right)} \\ &+ c_1\sqrt{k}\dfrac{ U_T\left(Y_{N-k, N} \right)- U_F\left(Y_{N-k, N} \right) }{\frac{N}{k}U_F'\left(\frac{N}{k} \right)} + 	 c_2\sqrt{k}\dfrac{ U_T\left(Y_{N-k-s, N} \right)- U_F\left(Y_{N-k-s, N} \right) }{\frac{N}{k+s}U_F'\left(\frac{N}{k+s} \right)} \\ 
		\overset{d}{=} &  c_1\sqrt{k}\dfrac{ \theta_{N-k, N}- U_F\left(\frac{N}{k} \right) }{\frac{N}{k}U_F'\left(\frac{N}{k} \right)} + 	 c_2\sqrt{k}\dfrac{ \theta_{N-k-s, N}- U_F\left(\frac{N}{k+s} \right) }{\frac{N}{k+s}U_F'\left(\frac{N}{k+s} \right)} \\ &+ c_1\sqrt{k}\dfrac{ U_T\left(Y_{N-k, N} \right)- U_F\left(Y_{N-k, N} \right) }{\frac{N}{k}U_F'\left(\frac{N}{k} \right)} + 	 c_2\sqrt{k}\dfrac{ U_T\left(Y_{N-k-s, N} \right)- U_F\left(Y_{N-k-s, N} \right) }{\frac{N}{k+s}U_F'\left(\frac{N}{k+s} \right)}  .
	\end{align}
	The result now follows:  the first two terms converge to the desired limit by lemma \ref{lemma:intermediatePairNoiseless}; the third and the fourth term converge to zero i.p., this convergence follows as in the proof of theorem \ref{theorem:intermediateNormality} applied at $k$ and $k+s$.
\end{proof}

\end{document}